\documentclass[aap]{imsart}

\RequirePackage{amsthm,amsmath,amsfonts,amssymb}
\RequirePackage[numbers, sort]{natbib}
\RequirePackage[colorlinks,citecolor=blue,urlcolor=blue]{hyperref}
\RequirePackage{graphicx}
\RequirePackage{array, tikz, subcaption}
\RequirePackage{enumitem}

\startlocaldefs
\numberwithin{equation}{section}
\theoremstyle{plain}
\newtheorem{theorem}{Theorem}[section]
\newtheorem{propo}[theorem]{Proposition}
\newtheorem{lemma}[theorem]{Lemma}
\newtheorem{corollary}[theorem]{Corollary}
\theoremstyle{definition}
\newtheorem{defi}{Definition}[section]
\newtheorem{remark}[theorem]{Remark}
\newtheorem{assum}{Assumption}[section]
\newtheorem{example}{Example}[section]

\newcommand{\E}{\mathbb{E}}
\newcommand{\Prob}{\mathbb{P}}

\newcommand{\Y}{\mathcal{Y}}

\newcommand{\bs}[1]{\boldsymbol{#1}}

\newcommand\numberthis{\addtocounter{equation}{1}\tag{\theequation}}
\newcommand{\weak}{\textbf{D}}

\endlocaldefs

\begin{document}

\begin{frontmatter}
\title{Large sample scaling analysis of the Zig-Zag algorithm for Bayesian inference}
\runtitle{Scaling analysis of Zig-Zag}

\begin{aug}
\author[A]{\fnms{Sanket}~\snm{Agrawal}\ead[label=e1]{sanket.agrawal@math.ku.dk}},
\author[B]{\fnms{Joris}~\snm{Bierkens}\ead[label=e2]{joris.bierkens@tudelft.nl}}
\and
\author[C]{\fnms{Gareth O.}~\snm{Roberts}\ead[label=e3]{gareth.o.roberts@warwick.ac.uk}}
\address[A]{University of Copenhagen\printead[presep={,\ }]{e1}}
\address[B]{Delft University of Technology\printead[presep={,\ }]{e2}}
\address[C]{University of Warwick\printead[presep={,\ }]{e3}}
\end{aug}

\begin{abstract}
    Piecewise deterministic Markov processes provide scalable methods for sampling from the posterior distributions in big data settings by admitting principled subsampling strategies that do not bias the output. An important example is the Zig-Zag process of [{\it Ann. Stats.} {\bf 47} (2019) 1288 - 1320] where clever subsampling has been shown to produce an essentially independent sample at a cost that does not scale with the size of the data. However, subsampling also leads to slower convergence and poor mixing of the process, a behaviour which questions the promised scalability of the algorithm. We provide a large sample scaling analysis of the Zig-Zag process and its subsampling versions in settings of parametric Bayesian inference. In the transient phase of the algorithm, we show that the Zig-Zag trajectories are well approximated by the solution to a system of ODEs. These ODEs possess a drift in the direction of decreasing KL-divergence between the assumed model and the true distribution and are explicitly characterized in the paper. In the stationary phase, we give weak convergence results for different versions of the Zig-Zag process. Based on our results, we estimate that for large data sets of size $n$, using suitable control variates with subsampling in Zig-Zag, the algorithm costs $O(1)$ to obtain an essentially independent sample; a computational speed-up of $O(n)$ over the canonical version of Zig-Zag and other traditional MCMC methods.
\end{abstract}



\end{frontmatter}


\section{Introduction}
\label{sec:intro}
    Piecewise deterministic Markov processes (PDMPs, \cite{davis_piecewise-deterministic_1984}) are a class of non-reversible Markov processes which have recently emerged as alternatives to traditional MCMC methods based on the Metropolis-Hastings algorithm. Practically implementable PDMP algorithms such as the Zig-Zag sampler \citep{bierkens_zig-zag_2019}, the Bouncy Particle Sampler \citep{bouchard-cote_bouncy_2018} and others \citep{corbella_automatic_2022,vasdekis_speed_2023,wu_coordinate_2020,bierkens_boomerang_2020,bertazzi_approximations_2022,pagani_nuzz_2024} have provided highly encouraging results for Bayesian inference as well as in Physics (\cite{krauth_event-chain_2021}, for example).
    As a result, there has been a surge of interest in the Monte Carlo community to better understand the theoretical properties of PDMPs (see, e.g., \citep{holderrieth_cores_2021,deligiannidis_exponential_2019,andrieu_hypocoercivity_2021,bierkens_limit_2017,bierkens_ergodicity_2019,bierkens_large_2021,bierkens_high-dimensional_2022, bierkens_scaling_2025}). This body of work demonstrates powerful convergence properties with impressive scalability as a function of the dimension. 

    For largely good reasons, much of MCMC theory has focused on the behaviour of algorithms as a function of the dimension.
    However, within applications in Bayesian Statistics, it is at least as important to consider how sampler properties vary with data. Theoretical studies to focus on this problem are relatively  rare (though see \cite{sahu_convergence_1999, johndrow_no_2024, ascolani_scalability_2026, agrawal_micromodes_2026}). The main reason for this is that studying the data effect is fraught with additional technical complications making rigorous results hard to obtain. To our knowledge this paper is the first to attempt to make progress on this problem for PDMPs. Furthermore, the emergence of {\em subsampling} methodology as outlined in the next paragraph, makes analysis of algorithm complexity as function of data size $n$ of great practical importance. 

    From a practical MCMC perspective, the most exciting property of PDMPs, at least for statistical applications is that of {\em principled subsampling} \citep{bierkens_zig-zag_2019}, whereby randomly chosen subsets of data can be used to estimate likelihood values within the algorithm without
    biasing its invariant distribution. The authors only cover subsampling for the Zig-Zag algorithm, although the unbiasedness property holds for a rather general class of PDMP MCMC algorithms. This unbiasing property is in sharp contrast to most MCMC methods \citep{welling_bayesian_2011, nemeth_stochastic_2021, bardenet_markov_2017}, although see \cite{cornish_scalable_2019}. In fact, \cite{bierkens_zig-zag_2019} gives experiments to show that even for data sets of large size $n$, subsets of size $1$ can be used to estimate the likelihood, producing a computational speed-up of $O(n)$. This leads to the property of {\em superefficiency} whereby, after initialization, the entire running cost of the algorithm is dominated for large $n$ by the cost of evaluating the likelihood once.
    In \cite{bierkens_zig-zag_2019}, the authors provide two algorithms which utilise this subsampling idea: the straightforward ZZ-SS algorithm together with a control-variate variant ZZ-CV which will be formally introduced in Section \ref{s:ZZBayes}.

    In \cite{bierkens_zig-zag_2019}, the authors also point out that although principled subsampling gives these dramatic improvements in implementation cost, the dynamics of ZZ-SS and ZZ-CV differ from that of canonical Zig-Zag: sample paths exhibit more diffusive behaviour due to the extra velocity switches induced by the variation in the estimate of the likelihood. This might lead to slower convergence and inferior mixing of the underlying stochastic process (see Figure \ref{fig:transient}). For one-dimensional Zig-Zag, \cite{bierkens_limit_2017} studies the diffusive behaviour resulting from excess switches outside of a statistical context. In \cite{andrieu_peskuntierney_2021}, the authors show that canonical Zig-Zag has superior mixing (in the sense of minimizing asymptotic variance for estimating a typical $L_2$ function) than any alternative with excess switches. Since both ZZ-SS and ZZ-CV can be viewed as Zig-Zag algorithms with excess switches, this raises the question of whether the inferior mixing of subsampling Zig-Zag algorithms relative to their canonical parent might negate their computational advantage.

    \begin{figure}
    \centering
        \begin{subfigure}[t]{\linewidth}
        \centering
            \begin{subfigure}[t]{0.49\linewidth}
                \includegraphics[scale = 0.35]{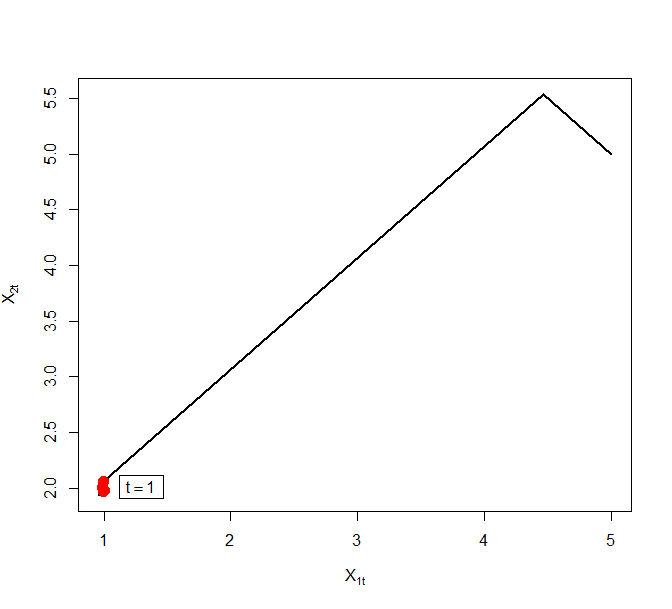}
            \end{subfigure}
            \begin{subfigure}[t]{0.49\linewidth}
                \includegraphics[scale = 0.35]{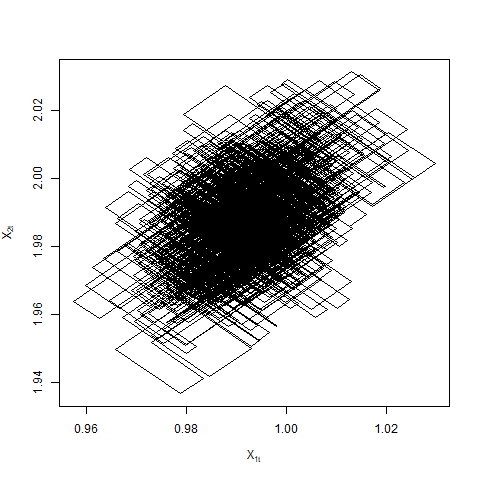}
            \end{subfigure}
        \subcaption{Canonical Zig-Zag (ZZ)}
        \end{subfigure}
        \begin{subfigure}[t]{\linewidth}
        \centering
            \begin{subfigure}[t]{0.49\linewidth}
                \includegraphics[scale = 0.35]{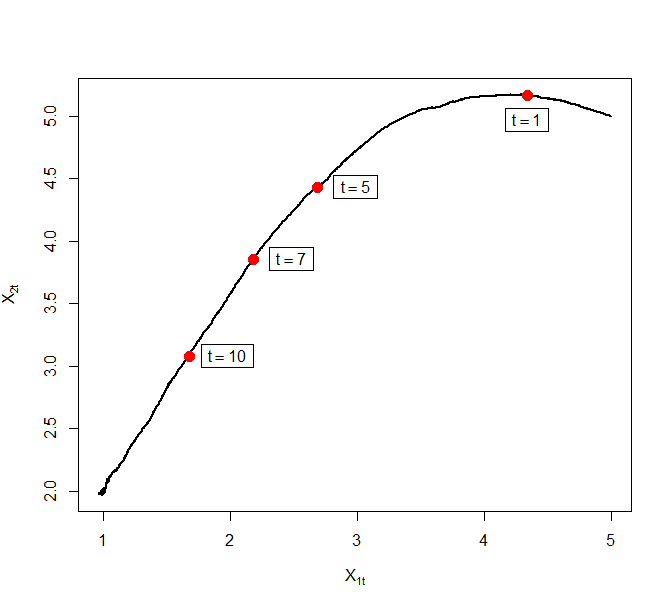}
            \end{subfigure}
            \begin{subfigure}[t]{0.49\linewidth}
                \includegraphics[scale = 0.35]{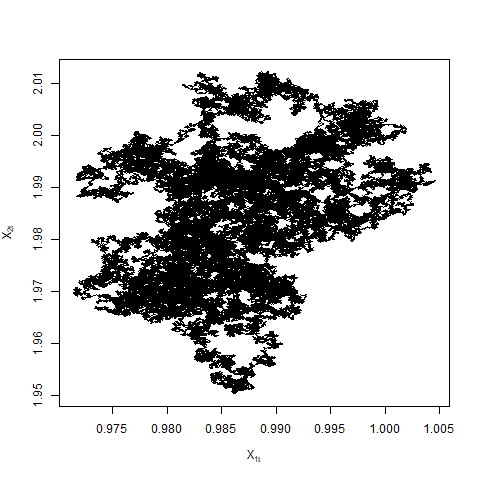}
            \end{subfigure}
        \subcaption{Zig-Zag with vanilla subsampling (ZZ-SS)}
        \end{subfigure}
        \begin{subfigure}[t]{\linewidth}
        \centering
            \begin{subfigure}[t]{0.49\linewidth}
                \includegraphics[scale = 0.35]{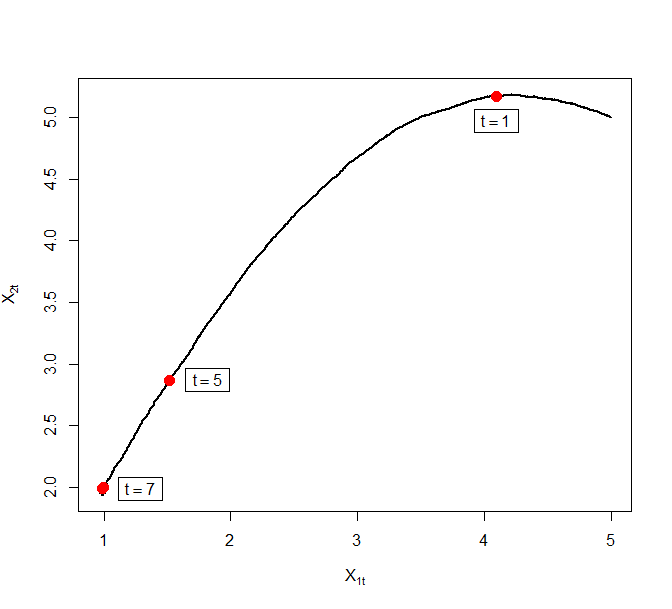}
            \end{subfigure}
            \begin{subfigure}[t]{0.49\linewidth}
                \includegraphics[scale = 0.35]{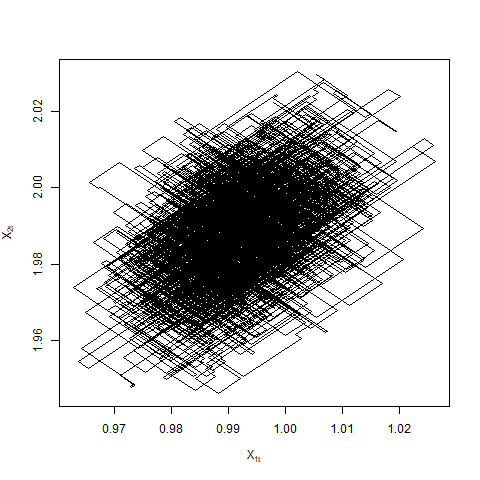}
            \end{subfigure}
        \subcaption{Zig-Zag with control variates (ZZ-CV)}
        \end{subfigure}
        \caption{Trajectories for different versions of the Zig-Zag algorithm targeting a Bayesian posterior for the 2-dimensional Bayesian logistic regression example with true parameter $(1, 2)$. The left panel displays the trajectories in the transient phase with red points marking the location of the process at different process times. The right panel illustrates the process in the stationary phase.}
        \label{fig:transient}
    \centering
    \end{figure}

    \subsection{Our contributions and summary of the results}
    \label{sec:contri}
        The goal of this work is to give the first rigorous theoretical results concerning the algorithm mixing deterioration caused by the excess switches induced by subsampling. To that end, we carry out a scaling limit analysis of Zig-Zag algorithm as the data set size $n \to \infty$, with dimension $d$ being fixed. Interplaying closely with the large-sample Bayesian asymptotics, we prove weak convergence results for the underlying PDMP for different subsampling versions of the Zig-Zag algorithm. The scaling limit approach has been a standard and impactful one in the literature to study mixing properties of MCMC algorithms with respect to dimension (\cite[e.g.][]{roberts_weak_1997,roberts_optimal_2001,agrawal_optimal_2023}) and dataset size \citep{schmon_optimal_2022, schmon_large-sample_2021}. Recently, there has been a growing interest in the non-asymptotic analysis of traditional MCMC methods (\cite{altmeyer_polynomial_2025, chewi_optimal_2021, dwivedi_log-concave_2019}), however, the non-asymptotic literature on PDMPs remains relatively limited (though see \cite{lu_explicit_2022, andrieu_hypocoercivity_2021}).   
        
        In parametric Bayesian inference, under sufficient regularity conditions, it is well known that the posterior distribution converges to a degenerate distribution as $n \to \infty$ \citep{berk_consistency_1970, kleijn_bernstein-von-mises_2012}. The point of concentration, say $x_0$ if it exists and is unique, is the one minimising the Kullback--Leibler (KL) divergence between the true data-generating distribution and the assumed model. We analyze the algorithm in two different phases: first, the {\it transient phase} when the process distance from $x_0$ bounded away from $0$. In this phase in the limit, any MCMC procedure is equivalent to an optimization routine seeking to find $x_0$ and so is Zig-Zag. In Section \ref{sec:transient}, we show that as $n \to \infty$, the Zig-Zag trajectories behave like deterministic smooth flows with a drift pushing them in the direction of decreasing KL-divergence. In particular, Theorem \ref{thm:transient} establishes weak convergence of the Zig-Zag trajectories to the solutions of a system of ODEs. These ODEs are explicitly characterized in equation \eqref{eq:limiting_drift}. We observe that while the canonical Zig-Zag attains optimal drift ($\pm 1$ in each coordinate) in the limit, the subsampling versions remain sub-optimal with a damping factor that depends on the model. In particular, we find that the vanilla subsampling starts with the optimal drift at infinity, suffering from an extreme slowdown as the process gets closer to the minimiser. On the other hand for subsampling with control variates, the drift remains positive at stationarity achieving optimal drift for model densities that are log-concave in the parameter.
        
        The fluid limit analysis in the transient phase captures large-scale movements (of $O_{p}(1)$) of the process. Under a sufficiently regular Bayesian model, for large $n$, the posterior concentrates within a ball of radius $n^{-1/2}$ around $x_0$ and resembles a Gaussian measure \citep{kleijn_bernstein-von-mises_2012}. The high concentration of posterior mass restricts the movement of the Zig-Zag process. Consequently, close to $x_0$ the fluid limit analysis yields a trivial limit (see Theorem~\ref{thm:canonical}). Let $(\hat{x}_n)_{n=1}^{\infty}$ be a consistent sequence of estimators of $x_0$ around which the finite-$n$ posteriors can be centered so that, by the Bernstein--von Mises Theorem, the rescaled coordinate $\xi(x)=n^{1/2}(x-\hat{x}_n)$ has an asymptotically Gaussian distribution \citep{kleijn_bernstein-von-mises_2012} (see also Assumption \ref{assum:mle_consistent}). The second phase we study is when the process is close to $\hat{x}_n$, i.e., when $\xi$ is $O_{p}(1)$. We call this the {\it stationary phase}. When looking at the evolution of Zig-Zag process under the rescaled coordinate $\xi$, the process trajectories have a non-trivial limit, which we characterize in Section \ref{sec:stationary}. In particular, Theorem \ref{theo:zzss_stat_fixed} establishes weak convergence of ZZ-SS trajectories to the solution of a reversible Ornstein-Uhlenbeck SDE given in \eqref{eq:limiting_sde}. The SDE is invariant with respect to the limiting Gaussian distribution and, as in the transient phase, has a damping factor owing to the subsampling. On the other hand, the ZZ-CV and the canonical Zig-Zag are slowed down by a factor $n^{1/2}$ to give a non-trivial limit. Particularly, the slowed down ZZ-CV and canonical Zig-Zag in the $\xi$-parameter converges weakly to a Zig-Zag process invariant with respect to the asymptotic Gaussian distribution (see Theorems \ref{theo:zzcv_stat_fixed} and \ref{theo:zzcv_stat_varying} respectively). The limiting switching rate for ZZ-CV is given by \eqref{eq:limrate} and is, in general, not canonical once again owing to the subsampling. 
        
        The weak convergence results further imply results on complexity bounds (Corollary \ref{corr:complexity_ss} and \ref{corr:complexity_cv}). In the following, and throughout the paper, `time' always refers to `process time'. Since the Zig-Zag sampler is implemented in continuous time, this is also same as `algorithmic time'. However, both are different from the `clock time'. In the stationary phase, the canonical Zig-Zag and ZZ-CV require $O(n^{-1/2})$ time to obtain an essentially independent sample in contrast to $O(1)$ for ZZ-SS. The computational effort of a sampling algorithm to obtain an essentially independent sample is the product of time taken to obtain an essentially independent sample and the computational effort per unit time. For PDMP-like algorithms this further breaks down to,
        \begin{align*}
            \text{computational effort} &:= (\text{time to obtain an essentially independent sample}) \\
            &\quad \quad \quad \times (\text{no. of proposed switch per unit time}) \\
            &\quad \quad \quad \times (\text{computational cost per proposed switch}).
        \end{align*} 
        
        \begin{table}
        \caption{Computational effort of different versions of Zig-Zag algorithm to obtain an essentially independent sample in the stationary phase.}
        \label{tab:summary}
        \begin{tabular}{@{}ccccc@{}}
        \hline
        Algorithm & time to an essentially ind. sample & $\sharp$events$/$time & effort/event  & total effort \\
        \hline\\[-2pt]
        ZZ    & $O\left(n^{-1/2}\right)$  & $O\left(n^{1/2}\right)$ & $O(n)$ & $\bs{O(n)}$ \\
              & (Thm.~\ref{theo:zzcv_stat_varying}, Corr.~\ref{corr:complexity_cv}) & (Sec.~\ref{sec:cost}) && \\
              \\[-2pt]
        ZZ-SS & $O(1)$ & $O(n)$ & $O(1)$ & $\bs{O(n)}$ \\
              & (Thm.~\ref{theo:zzss_stat_fixed}, Corr.~\ref{corr:complexity_ss}) & (Sec.~\ref{sec:cost}) && \\
              \\[-2pt]
        ZZ-CV & $O\left(n^{-1/2}\right)$  & $O\left(n^{1/2}\right)$ & $O(1)$ & $\bs{O(1)}$ \\
              & (Thm.~\ref{theo:zzcv_stat_fixed}, Corr.~\ref{corr:complexity_cv}) & (Sec.~\ref{sec:cost}) && \\[2pt]
        \hline
        \end{tabular}
        \end{table}
        Similar arguments for other methods and/or in different settings have been previously used to estimate algorithmic complexities of Monte Carlo methods (see e.g. \cite{bierkens_high-dimensional_2022, cornish_scalable_2019, pollock_quasi-stationary_2020}). The canonical Zig-Zag incurs an $O(n)$ cost at each proposed switch while the subsampling versions are designed to only cost $O(1)$. The number of proposed switches depend on the switching rates of the algorithm. In the $\xi$-parameter, the effective switching rates are $O(n^{1/2})$ for canonical Zig-Zag and ZZ-CV and $O(n)$ for ZZ-SS. The computational effort of different versions of the Zig-Zag algorithm in the stationary phase is summarized in Table~\ref{tab:summary}. The overall computational effort of the canonical Zig-Zag and ZZ-SS per independent sample is $O(n)$. In contrast, ZZ-CV achieves a computational speed-up of $O(n)$ in the stationary phase provided the reference point is chosen sufficiently close to $x_0$. Although, note that ZZ-CV will have an initialization cost associated with finding a suitable reference point and then evaluating the likelihood at it. This will incur a one-off $O(n)$ cost after which, essentially independent samples can be obtained at a cost that does not increase with $n$.   

        Finally, we also investigate the choice of subsample size. In the transient phase, increasing the subsample size from $1$ to $m$ increases each limiting drift component’s magnitude by at most a factor $m$ (Remark~\ref{rem:optimal_m}). An analogous comparison holds for the limiting diffusion coefficients of ZZ-SS (Remark~\ref{rem:zzss_complexity}) and limiting switching rates for ZZ-CV (Remark~\ref{rem:zzcv_complexity}) in the stationary phase. These comparisons support using subsamples of size $1$ when the computational cost per proposed switch grows linearly with subsample size, in the absence of implementation advantages such as parallelization.
    
        The rest of the paper is organised as follows. Section \ref{s:ZZBayes} details the Zig-Zag algorithm applied to the task of sampling from a Bayesian posterior distribution under the assumption of independent and identically distributed (i.i.d.)\ observations. Section~\ref{sec:transient} and Section~\ref{sec:stationary} contains our main results analysis of the transient phase and stationary phase respectively of the algorithm. In Section \ref{sec:examples}, we discuss the results of the transient phase in more detail along with some illustrations on simple models. We conclude with a final discussion in Section \ref{sec:conclusion}. All the proofs are deferred to Appendix.

\section{Zig-Zag process for parametric Bayesian inference with i.i.d. observations}
\label{s:ZZBayes}
    The $d$-dimensional Zig-Zag process, introduced in \cite{bierkens_zig-zag_2019}, is a continuous-time Markov process with state space $\mathsf{E} = \mathbb{R}^d \times \mathcal{V}$ where $\mathcal{V} = \{-1,+1\}^d$. It is characterized by a collection, $\{\lambda_i\}_{i=1}^d$, of continuous {\it switching rates} where, for each $i$, $\lambda_i: \mathsf{E} \to [0, \infty)$ is such that $t \mapsto \lambda_i(x + vt, v)$ is locally integrable in $(0, \infty)$ for all $(x, v) \in \mathsf{E}$. Given $\{\lambda_i\}_{i=1}^d$, Zig-Zag process $(Z_t)_{t\ge 0} = ((X_t, V_t))_{t \ge 0}$ evolves as a deterministic flow $\Phi_t(x,v) = (x + vt, v)$ until at a random time, which occurs at rate $\lambda_i(X_t, V_t)$, the $i$-th coordinate of $V_t$ changes sign. We refer to $X_t$ and $V_t$ as {\it position} and {\it velocity} process respectively. In this section, we review how Zig-Zag process can be designed to sample from a Bayesian posterior distribution in the setting where the data are assumed to be generated independently from a fixed probability model. For details, we refer to \cite{bierkens_zig-zag_2019}.

\subsection{Zig-Zag sampler for Bayesian inference}
\label{sec:zz_bayes_1}
        Let $\bs{y}^{(n)} = (y_1, \dots, y_n)$ be $n$ i.i.d. observations modeled via a parametric family $\left\{f(\cdot; x), x \in \mathbb{R}^d\right\}$ of probability densities. Consider the classical Bayesian posterior density, $\pi^{(n)}$, given by
        \begin{equation}
        \label{eq:posterior}
            \pi^{(n)}(x) \propto \pi_0(x)\prod_{j=1}^n f(y_j; x); \quad x \in \mathbb{R}^d,
        \end{equation}
        for some prior density $\pi_0$. Let $U_n(x) = -\log \pi^{(n)}(x)$, $x \in \mathbb{R}^d$.
        
        For each $i - 1, \dots, d$, let $\mathsf{F}_i: \mathcal{V} \to \mathcal{V}$ be the {\it flip-operator} that flips the sign of $i$-th component of $v$, i.e.,
        \[
            (\mathsf{F}_i(v))_{j} = \begin{cases}
                v_j, & j \neq i,\\
                -v_j, & j = i.
            \end{cases} \numberthis \label{eq:flip_operator}
        \]
         Let $\gamma^n: \mathsf{E} \to [0, \infty)^d$ continuous be such that $\gamma^n_i(x, v) = \gamma_i(x, \mathsf{F}_i(v))$ for all $(x, v) \in \mathsf{E}$. Let $(\cdot)_{+}$ denote the positive part function, i.e., for any $a \in \mathbb{R}$, $(a)_{+} = \max\{x, 0\}$. Define
        \begin{equation}
	\label{eq:rates}
	        \lambda^n_i(x, v) = (v_i \nabla_iU_n(x))_{+} + \gamma^n_i(x, v), \quad (x, v) \in \mathsf{E}, \ i = 1, \dots, d.
	\end{equation}
        The Zig-Zag process with switching rates $\{\lambda^n_i\}_{i=1}^d$ given by \eqref{eq:rates} has an invariant distribution whose marginal distribution on $\mathbb{R}^d$ is equivalent to $\pi^{(n)}$ \citep{bierkens_zig-zag_2019}. The switching rates for which $\gamma^n_i \equiv 0$ in \eqref{eq:rates} are called {\it canonical switching rates}. When $\gamma^n_i \equiv 0$ for all $i$, we call the corresponding Zig-Zag process the {\it canonical Zig-Zag process}. We denote the $i$-th canonical rate by $\lambda^n_{\text{can}, i}$. The function $\gamma^n$ itself will be called the {\it excess switching rate.}

        The implementation of the Zig-Zag sampler boils down to simulating events from inhomogeneous Poisson processes with intensity $\lambda^n_i(t) = \lambda^n_i(x + vt, v)$ for each initial condition $(x, v) \in \mathsf{E}$ and each $i = 1, \dots, d$. This is typically done by thinning another Poisson process with a bounding intensity as described in the following (see also Algorithm 1 in \cite{bierkens_zig-zag_2019}). Given initial condition $(x, v)$, suppose $\lambda^n_i(t) \le M^n_i(t)$ for all $i = 1, \dots, d$. For each $i$, propose an event time $\tau_i$ according to the Poisson process with intensity $M^n_i(t)$. Let $i_0 = \arg\min \tau_i$, and $\tau = \tau_{i_0}$. Accept the proposed event and flip the $i_0$-th component of velocity with probability $\lambda^n_{i_0}(x + v\tau, v)/M^n_{i_0}(\tau)$. Repeat the procedure from the new starting point $(x + v\tau, F_{i_0}(v))$ if accepted and $(x + v\tau, v)$ if rejected. The collection $\{M^n_i\}_{i=1}^d$, termed as {\it computational bounds} in \cite{bierkens_zig-zag_2019}, are crucial to the algorithm's performance.

    \subsection{Exact-approximate Zig-Zag}
        The implementation of Zig-Zag sampler described above requires evaluation of one of the $\lambda^n_i$s, and hence one of the coordinates of $\nabla U_n$, at each proposed event time to carry out the accept-reject step. In terms of the computation cost, this is an $O(n)$ computation since it requires processing the whole dataset, although parallelization might be possible. As $n$ becomes large, the $O(n)$ cost per proposed event renders the algorithm prohibitively expensive to run for longer durations.


        Consider the canonical rates $\lambda^n_{\text{can}, i}(x, v) = (v_i\nabla_i U_n(x))_{+}$. 
        For all $x$, suppose there exists an unbiased estimator, $\zeta(x)$ say, of $\nabla U_n(x)$. Further, suppose the cost of obtaining one random realization of $\zeta(x)$ for any $x$ is an order of magnitude cheaper than that of processing the whole dataset and it can be executed independent of the algorithm. Given initial conditon $(x, v)$, let the computational bounds $\{M^n_i\}$ now satisfy $(v_i\zeta_i(x + vt))_{+} \le M^n_i(t)$ for all possible realizations of $\zeta$. Let $\tau$ and $i_0$ be the proposed event time and coordinate respectively sampled as described before. At the accept-reject step, sample a random realization of $\zeta(x + v\tau)$ and flip the $i_0$-th component of velocity with probability $(v_i\zeta_i(x + v\tau, v))_{+}/M^n_{i_0}(\tau)$. The resulting algorithm is again a Zig-Zag process with the effective switching rate of the new algorithm given by, 
        \begin{equation}
        \label{eq:rates_eff}
            \lambda^{n}_{\text{eff}, i}(x, v) := \mathbb{E}\left[(v_i\zeta_i(x))_{+} \mid {\bs{y}^{(n)}}\right],
        \end{equation}
        for all $(x, v) \in \mathsf{E}, \ i = 1, \dots, d$ (see, e.g., \cite{bierkens_zig-zag_2019, vanetti_piecewise-deterministic_2018}). Since $\zeta$ is an unbiased estimator of $\nabla U_n(x)$, it is easy to see that $\lambda^n_{\text{eff}, i}$ defined by \eqref{eq:rates_eff} satisfies \eqref{eq:rates} for all $i$ and $(x,v)$. Thus, the procedure described above preserves the invariance of the target distribution. Moreover, the computational cost per proposed event for this procedure is bounded above by the cost of generating one random realization of $\zeta$. However, from Jensen's inequality and the convexity of $a \mapsto (a)_{+}$,
        \[
            \lambda^{n}_{\text{eff}, i}(x, v) = \mathbb{E}\left[(v_i\zeta_i(x))_{+} \mid {\bs{y}^{(n)}}\right] \ge \left(v_i\mathbb{E}(\zeta_i(x) \mid {\bs{y}^{(n)}})\right)_{+} = \lambda^n_{\text{can}, i}(x, v), \numberthis \label{eq:effective_rate}
        \]
        for all $(x, v)$. So, the above procedure is equivalent to adding a positive excess switching rate to the underlying Zig-Zag process and, hence, has sub-optimal mixing properties \cite{andrieu_peskuntierney_2021}. 

        In this paper, we analyze the Zig-Zag process for a particular class of unbiased estimators used in \cite{bierkens_zig-zag_2019}. These estimators use a randomly chosen subset of data at each step to act as a proxy for the whole dataset and are particularly useful when $\nabla U_n$ can be decomposed as a sum. We call these as {\it subsampling estimators} and the associated Zig-Zag process as {\it subsampling Zig-Zag}. We detail subsampling Zig-Zag alongwith the technical setup for this paper in the next section.

    \subsection{Subsampling with Zig-Zag and technical setup} 
        Let $(\Omega, \mathcal{F}, \Prob)$ be an arbitrary complete probability space, rich enough to accommodate all the random elements defined in the following. This probability space will be the primary source of randomness. The data sequence is defined on it via a data-generating process (see below) and the Zig-Zag process is viewed as a random element on $(\Omega, \mathcal{F}, \Prob)$ conditioned on the data. For notational convenience, we will suppress the dependence on $\omega$ throughout the manuscript. 
        
        Denote the observation space by $\Y$ and the space of probability measures on $\Y$ by $\mathcal{M}(\Y)$. Let $Y: \Omega \to \mathcal{Y}$ be a data-generating process such that $Y \sim P$ for some $P \in \mathcal{M}(\mathcal{Y})$, i.e., $P$ is the push-forward of $\Prob$ by $Y$. Fix a parametric model $\{F_x, x \in \mathbb{R}^d\} \subset \mathcal{M}(\mathcal{Y})$ for $P$. Suppose $P$ and $\{F_x, x \in \mathbb{R}^d\}$ are absolutely continuous with respect to a common non-atomic $\sigma$-finite reference measure $\nu$ on $\Y$. Let $p := dP/d\nu$ be the continuous density of $P$ with respect to $\nu$ and for each $x \in \mathbb{R}^d$, let $f(\cdot ; x) := dF_x/d\nu$ be the respective density of $F_x$. 

        Let $\bs{Y}^{(\mathbb{N})} \sim P^{\otimes \mathbb{N}}$ be a sequence of random elements defined on $\Omega$. Let $n \in \mathbb{N}$ denote the {\it dataset size} and for each $n$, let $\bs{Y}^{(n)} = (Y_1, \dots, Y_n)$ denote the dataset. For all $n$, let $\Pi^{(n)}$ be the random measure on $\mathbb{R}^d$ with Lebesgue density $\pi^{(n)}$ such that, 
	    \[
	        \pi^{(n)}(x) \propto \exp\left(\sum_{j=1}^n \log f(Y_j; x)\right); \quad x \in \mathbb{R}^d. \numberthis \label{eq:pin}
	    \] 


        \begin{remark}
            Throughout this paper, we will assume there exists an $n_0$ such that for all $n \ge n_0$, $\int_{\mathbb{R}^d} \prod_{j=1}^n f(Y_j; x)\ dx < \infty$, $\Prob$-almost surely. Sufficient conditions on $\{F_x, x \in \mathbb{R}^d\}$ and $P$ to ensure this are given, e.g., in \cite{berk_consistency_1970}. Under this assumption, for all $n \ge n_0$, $\Pi^{(n)}$ is the classical Bayesian posterior for $x$ under the uniform prior (see e.g. \cite{ghosh_introduction_2006} for details). Since we will be concerned with asymptotic behaviour of the Zig-Zag process for large $n$, we can work under the assumption of a uniform prior without loss of generality. See also the discussion in Section~\ref{sec:conclusion}.
        \end{remark}

        Define the potential function $U_n := -\log \pi^{(n)}(x)$ and the negative score function, $S$, by        
        \begin{equation}
        \label{eq:score}
            S(x, y) := -\nabla_x\log f(y; x), \quad x \in \mathbb{R}^d, y \in \mathcal{Y},
        \end{equation}
        and let $S_i$ denote the $i$-th component of $S$.
        For all $n$, let $(E^j)_{j=1}^{n}$ be a sequence of functions such that,
        \begin{equation}
        \label{eq:additive_decomp}
            \nabla U_n(x) = \nabla (-\log \pi^{(n)})(x) = \sum_{j=1}^n E^j(x), \quad x \in \mathbb{R}^d,\  \Prob\text{-almost surely}.
        \end{equation}
        For example, from \eqref{eq:pin}, the simplest choice of $E^j$ would be $E^j(x) = S(x, Y_j)$.


        Let $\mathcal{J}_{(n,m)}$ be the collection of all possible subsets of size $m$ with distinct elements from $\{1, \dots, n\}$. Given the decomposition in \eqref{eq:additive_decomp}, simple unbiased estimators of $\nabla U_n(x)$ are based on random batches of size $m(n) \le n$ from the whole data set, i.e.,
        \[
            \zeta(J, x) = \frac{n}{m(n)}\sum_{j \in J}E^{j}(x), \numberthis \label{eq:zetaJ}
        \]
        where $J: \Omega \to \mathcal{J}_{(n, m)}$ is a simple random sample of size $m$ drawn without replacement, i.e, $J \sim \text{Uniform}(\mathcal{J}_{(n, m)})$. In our asymptotic study in later sections, $m$ may be chosen to depend explicitly on $n$, e.g., $m(n) = \lceil\log n\rceil$ for $n \ge 2$, or be eventually held fixed, i.e. $m(n) = m$ for $n \ge m$. 
        In this paper, we analyze the following two choices of $E^j$ suggested in \cite{bierkens_zig-zag_2019}. Our results, however, can be easily extended to other choices.

     \subsubsection{Zig-Zag with subsampling (ZZ-SS)}
        This is the vanilla subsampling strategy where we choose for $j = 1, \dots, n$, 
    	\begin{equation}
            \label{eq:zzss_ej1}
    	    E^j(x) = E^j_{\text{ss}}(x) := S(x, Y_j) = -\nabla_x\log f(Y_j; x); \quad \ x \in \mathbb{R}^d.
    	\end{equation}
        This choice is suitable when the gradient of the model log density is globally bounded i.e. there exist positive constants $c_i$ such that $|\nabla_{x,i}\log f(y;x)| \le c_i$ for all $i, x$ and $y$. In accordance with the nomenclature used in \cite{bierkens_zig-zag_2019}, we call the Zig-Zag process resulting from setting $E^j(\cdot ) = S(\cdot, Y_j)$ the Zig-Zag with subsampling (ZZ-SS).

      \subsubsection{Zig-Zag with Control Variates (ZZ-CV)}  
        The second choice is suitable when the gradient term $\nabla \log f(y; x)$ is globally Lipschitz. The Lipschitz condition allows the use of control variates to reduce variance in the vanilla subsampling estimator of $\nabla U_n$ \cite{bierkens_zig-zag_2019}. 
            
        For each $n$, let $T_n: \Omega \times \mathcal{Y}^n \to \mathbb{R}^d$ be a jointly measurable function. Define the sequence of {\it reference points},
        \begin{equation}
        \label{eq:reference_points}
            (X^*_n)_{n=1}^{\infty} := (T_n(\bs{Y}^{(n)}))_{n = 1}^{\infty}
        \end{equation} 
        to be a random element on $\Omega$ taking values in $(\mathbb{R}^d)^{\mathbb{N}}$. The control variate idea is to set, for $j = 1, \dots, n$, 
        \begin{align*}
           E^j(x) = E^j_{\text{cv}} &:= E^j_{\text{ss}}(x) - E^j_{\text{ss}}(X^*_n) + n^{-1}\sum_{k=1}^n E^k_{\text{ss}}(X^*_n) \\
           &= S(x, Y_j) - S(X_n^*, Y_j) + n^{-1}\nabla U_n(X_n^*); \quad x \in \mathbb{R}^d. \numberthis \label{eq:zzcv_ej1}
        \end{align*}
        Once again, as in \cite{bierkens_zig-zag_2019}, We call the corresponding Zig-Zag process Zig-Zag with control variates (ZZ-CV).    
            

        \begin{remark}[Reference points]
        \label{remark:ref_points}
            From~\eqref{eq:reference_points}, observe that the reference points are allowed to depend on the observed dataset. In~\cite{bierkens_zig-zag_2019}, the authors argue that the reference points should be chosen to lie close to the posterior mode. Given the form of the posterior \eqref{eq:pin}, a natural choice for $T_n$ is the maximum likelihood estimator (MLE) for the model parameter $x$. Under sufficient regularity of the model, both MLE and posterior mode will converge to the unique minimizer, say $x_0$, of the KL-divergence between $P$ and $\{F_x,\ x \in \mathbb{R}^d\}$. In practice, any consistent estimator of $x_0$ will suffice. In particular, for the consistent sequence $(\hat{x}_n)_{n=1}^{\infty}$ in Section \ref{sec:contri}, it is enough to chose reference points to satisfy $X_n^* - \hat{x}_n = o_{p}(1)$. Note that this is not necessary for our results in the next section.
             
            We will assume that the sequence $\{T_n\}_{n=1}^{\infty}$ is chosen and fixed as soon as the data is observed and a model is identified. Thus, the value of $(X^*_n)_{n=1}^{\infty}$ is known and fixed before the Zig-Zag algorithm is initialized.
        \end{remark}

        For unbiased estimator $\zeta$ taking the form in \eqref{eq:zetaJ}, the effective switching rates for the resulting Zig-Zag process are given by,
        \begin{align*}
            \lambda^n_i(x, v; m(n)) &= \mathbb{E} \left[(v_i\cdot \zeta_i(J, x))_{+} \mid \bs{Y}^{(n)}, X^*_n\right] \\
            &= \frac{1}{|\mathcal{J}_{(n, m(n))}|}\sum_{J \in \mathcal{J}_{(n, m(n))}}\left(v_i\cdot \frac{n}{m(n)}\sum_{j \in J}E_i^j(x)\right)_{+}, \numberthis \label{eq:subsampling_rates}
        \end{align*}
        for all $(x, v) \in \mathsf{E}, \ i = 1, \dots, d$.  

        Let $\mathsf{a} \in \{\text{can}, \text{ss}, \text{cv}\}$ denote the subsampling scheme. Let $\{\lambda^n_{\text{can}, i}\}_{i=1}^d$ denote the canonical rates, i.e., $\lambda^n_{\text{can}, i}(x,v) = (v_i\nabla U_n(x))_{+}$. For $\mathsf{a} \in \{\text{ss}, \text{cv}\}$, let $\{\lambda^n_{\mathsf{a}, i}\}_{i=1}^d$ denote the effective switching rates in \eqref{eq:subsampling_rates} with $E^j = E^j_{\mathsf{a}}$, where $E^j_{\text{ss}}$ and $E^j_{\text{cv}}$ are given by \eqref{eq:zzss_ej1} and \eqref{eq:zzcv_ej1} respectively. When $m(n) = n$, the subsampling rates are equivalent to the canonical rates, i.e., for all $(x,v) \in \mathsf{E}$, and $i = 1, \dots, d$,
        \[
            \lambda^n_{\text{ss},i}(x,v; n) = \lambda^n_{\text{cv},i}(x,v; n) = \lambda^n_{\text{can},i}(x,v).
        \]

        For each $\mathsf{a} \in \{\text{can}, \text{ss}, \text{cv}\}$, and each $n$, let $Z^n_{\mathsf{a}} = (Z^n_{\mathsf{a}, t})_{t\ge 0} =  (X^n_{\mathsf{a},t}, V^n_{\mathsf{a},t})_{t \ge 0}$ be a Zig-Zag process conditional on $\bs{Y}^{(n)}$ and, when applicable $X^*_n$, with switching rates $\{\lambda^n_{\mathsf{a},i}\}_{i=1}^d$.

        

    \subsection{Notation and assumptions}
        We will use $\|\cdot\|$ to denote the Euclidean norm in $\mathbb{R}^d$. For any $a \in \mathbb{R}$, $(a)_{+}$ denotes the positive part of $a$ i.e. $(a)_{+} = \max\{0, a\}$. 
        
        Unless stated otherwise, the derivatives are always considered with respect to the parameter $x$.
        If $h:\mathbb{R}^d \to \mathbb{R}$ is differentiable, we denote the gradient vector by $\nabla h = (\nabla_1h, \dots, \nabla_dh)$. If $h$ is twice differentiable, we use $\nabla^{2}h$ to denote the Hessian of $h$ i.e. $\nabla^2 h = [\nabla^2_{\cdot 1}h \ \dots \ \nabla^2_{\cdot d}h]$ where $\nabla^2_{\cdot i}h = \nabla (\nabla_i h)$ denotes the $i$-th column of $\nabla^2h$ for all $i = 1, \dots, d$. 
        If $g$ is a function from $\mathbb{R}^d \times \mathcal{Y}$ to $\mathbb{R}$, $\nabla g$ and $\nabla^{2}g$ denote the first and second derivatives respectively of $g(x, y)$ with respect to $x \in \mathbb{R}^d$. Similarly, if $h$ is a real-valued function on $\mathbb{R}^d \times \mathcal{V}$, $\nabla h$ and $\nabla^{2}h$ denote the first and second derivatives respectively of $h(x, v)$ with respect to $x \in \mathbb{R}^d$. We denote the $d$-dimensional vector of all $1$s by $\bs{1}_d$ and the $d \times d$ identity matrix by $\bs{I}_d$.
        
        Given $\omega \in \Omega$, we will use lower-case letters to denote the observed value of random quantities such as $\bs{y}^{(n)} = (y_{1}, \dots, y_n)$ for $\bs{Y}^{(n)}$ and $x_n^{*} = T_n(\omega, \bs{y}^{(n)})$ for $X_n^*$. 

        For $m \in \mathbb{N}$ and a measurable function $g$ of $Y_1,\dots,Y_m$ and possibly other arguments, we write $\mathbb{E}_{P}\left[g(Y_1,\ldots,Y_m,\dots)\right] := \int g(y_1,\ldots,y_m,\dots)P^{\otimes m}(d y_1\cdots d y_m)$. Thus, $\mathbb{E}_{P}$ denotes expectation only with respect to $Y_1,\ldots,Y_m \sim P^{\otimes m}$; any other arguments of $g$, including random ones, are held fixed. The number of integrated variables is inferred from the argument.

        We make the following assumptions on the model $\{F_x, x \in \mathbb{R}^d\}$ and observed data.

        \begin{assum}[\it Smoothness]
        \label{assum:smoothness}
           For each $y$, $f(y; \cdot) \in \mathcal{C}^3(\mathbb{R}^d)$. There exists $M' > 0$ such that for all $i = 1,\dots, d$,
            \[
                \|\nabla^{2}S_i(x, y)\|_{op} \le M', \quad x \in \mathbb{R}^d, y \in \mathcal{Y}.
            \]
            Here $\|\cdot\|_{op}$ denotes the Euclidean operator norm.
        \end{assum}
        
        \begin{assum}[\it Moments]
        \label{assum:moment}
            For all $x \in \mathbb{R}^d$, following moment conditions hold:
            \[
                0 < \mathbb{E}_{P} \left\|\nabla_x \log f(Y; x)\right\| < \infty, \quad \mathbb{E}_{P} \left\|\nabla^2_x \log f(Y; x)\right\| < \infty.
            \]
        \end{assum}

        \begin{assum}[\it Unique minimizer of the KL-divergence] 
        \label{assum:unique_minimiser}    
            The KL-divergence of the assumed model relative to the true data-generating distribution is uniquely minimized at $x_0 \in \mathbb{R}^d$, i.e.,
            \[
                -\mathbb{E}_{P}\log f(Y; x_0) = \inf_{x \in \mathbb{R}^d} -\mathbb{E}_{P}\log f(Y; x) < \infty.
            \]
            Define $I(x_0) := -\mathbb{E}_{P} \left(\nabla^2_x \log f(Y; x_0)\right)$. Then $I(x_0)$ is positive definite. Moreover,
            \[
                \mathbb{E}_{P} \left\|\nabla_x \log f(Y; x_0)\right\| < \infty, \quad  \mathbb{E}_{P} \left\|\nabla^2_x \log f(Y; x_0)\right\| < \infty. 
            \]
        \end{assum}

        \begin{assum}[\it Bernstein--von Mises Theorem]
        \label{assum:mle_consistent}
            There exists a consistent sequence of estimators $(\hat{x}_n)_{n=1}^{\infty}$ such that, 
            \begin{enumerate}[label=(\roman*)]
                \item $\hat{x}_n \to x_0$ in $\Prob-$probability as $n \to \infty$.
                \item $\sum_{j=1}^n S(\hat{x}_n, Y_j) = o_{p}(n^{1/2})$ as $n \to \infty$, where $o_p$ is with respect to $\Prob-$probability. 
                \item Let $\mathcal{B}$ denote the collection of Borel sets on $\mathbb{R}^d$. Then,
                \[
                    \sup_{B \in \mathcal{B}}\left|\Pi^{(n)}(n^{-1/2}B + \hat{x}_n) - N_{\bs{0}, (I(x_0))^{-1}}(B)\right| \to 0
                \]
                as $n \to \infty$ in $\Prob-$probability, where $I(x_0)$ is as defined in Assumption \ref{assum:unique_minimiser}. Here, $N_{\theta, \Sigma}$ denotes the (multivariate) Gaussian distribution with mean vector $\theta$ and variance-covariance matrix $\Sigma$.
            \end{enumerate}
        \end{assum}

        Assumption~\ref{assum:smoothness} is used throughout the paper and provide sufficient regularity in the model and on the data. It can be relaxed to a certain degree (see Remark~\ref{rem:general}), but we do not pursue that generalization in this paper. Assumption~\ref{assum:moment} is used in the transient phase only. In the stationary phase, it is replaced by the weaker moment assumption in Assumption~\ref{assum:unique_minimiser}. In addition, Assumptions~\ref{assum:unique_minimiser} and \ref{assum:mle_consistent} are used in the stationary phase to ensure the posterior converges at a certain rate to a point mass as $n$ goes to infinity. Part (iii) of Assumption \ref{assum:mle_consistent} further implies asymptotic Gaussianity of the posterior to allow for non-trivial limits in the stationary phase. Under regularity conditions the sequence, $\hat{x}_n = \arg\max_{x \in \mathbb{R}^d}\sum_{j=1}^n \log f(Y_j; x)$, of MLE satisfies the conditions in Assumption~\ref{assum:mle_consistent} \citep{kleijn_bernstein-von-mises_2012}. For example, if $Y_j \sim N(0, 1)$ and the prior is flat, then setting $\hat{x}_n = \bar{Y}_n + n^{-1/2-\delta}$ for any $\delta > 0$ satisfies the conditions in Assumption~\ref{assum:mle_consistent}. For our main results in Sections \ref{sec:transient} and \ref{sec:stationary}, we do not assume that $P \in \{F_{x}, x \in \mathbb{R}^d\}$, i.e., we do not assume the model is well-specified. If indeed it is, $x_0$ in Assumption \ref{assum:unique_minimiser} is such that $P = F_{x_0}$.

\section{Fluid limits in the transient phase}
\label{sec:transient}
    This section establishes fluid limits for the Zig-Zag process in Bayesian inference as the dataset size goes to infinity. Throughout this section, we assume the following about the sequence of reference points defined in \eqref{eq:reference_points}.

    \begin{assum}
    \label{assum:reference}
        There exists an $X^*$ defined on $\Omega$ such that $\mathbb{E}\|X^*\|^2 < \infty$, $X^*$ is independent of $\bs{Y}^{(\mathbb{N})}$, and for the sequence of reference points we have that $X^*_n \to X^*$, $\Prob$-almost surely as $n \to \infty$. 
    \end{assum}    

    In the following, all quantities are understood as random elements on $\Omega$. Recall that we are suppressing the dependence on $\omega$ for notational convenience.

    \subsection{Asymptotic switching rates} 
        Consider the subsampling switching rates defined in \eqref{eq:subsampling_rates}. As more data becomes available, the posterior distribution contracts to a point mass. In this case, the gradient of the log-likelihood and hence, the switching rates, grows linearly in $n$. 

        Let $X^*$ be as in Assumption~\ref{assum:reference}. For $i = 1, \dots, d$ and $m \in \mathbb{N}$ define
        \begin{gather}
            \lambda_{\text{can}, i}(x, v) := \left(v_i \cdot \mathbb{E}_P S_i(x, Y)\right)_{+}, \label{eq:asymp_rate_can}\\
            \lambda_{\text{ss}, i}(x, v; m) := \mathbb{E}_{P}\left[\left(v_i \cdot m^{-1}\sum_{j=1}^m S_i(x, Y_j)\right)_{+}\right], \label{eq:asymp_rate_ss}\\
            \lambda_{\text{cv}, i}(x, v; m) := \mathbb{E}_{P}\left[\left(v_i \cdot m^{-1}\sum_{j=1}^m\left( S_i(x, Y_j) - S_i(X^*, Y_j) + \mathbb{E}_{P}S_i(X^*,Y)\right)\right)_{+}\right], \label{eq:asymp_rate_cv}
        \end{gather}
        Here, $\E_P S_i(X^*,Y)$ denotes the function $r\mapsto \E_P S_i(r,Y)$ evaluated at $r=X^*$.

        \begin{remark}
        \label{remark:omega_dependence}
            For fixed $(x, v) \in \mathsf{E}$, both $\lambda_{\text{ss}, i}$ and $\lambda_{\text{can}, i}$ are constant $\mathbb{P}$-almost surely. On the other hand, $\lambda_{\text{cv}, i}$ depends on $\omega$ through $X^*$ where $X^*$ is as in Assumption~\ref{assum:reference}, and hence, depends on the sequence of statistics, $\{T_n\}_{n=1}^{\infty}$. If $T_n$ is a deterministic function of $\bs{Y}^{(n)}$ for each $n$, under Assumption~\ref{assum:reference}, $X^*$ is a finite constant $\Prob-$almost surely. Furthermore, following Remark~\ref{remark:ref_points}, if $T_n$ is a consistent estimator of the minimizer $x_0$, then $X^* = x_0$ $\mathbb{P}$-almost surely. If $T_n$ has randomness for each $\bs{Y}^{(n)} = \bs{y}^{(n)}$, $X^*$ itself may either be a constant or a non-degenerate $\mathbb{R}^d$-valued random variable. In any case, since $X^*$ is the almost-sure limit of the reference points, conditioning on the entire sequence $\{X_n^*\}_{n=1}^{\infty}$ fixes $X^*$ and hence the limiting switching rates. Consequently, the process sequence $\{Z^n_{\text{cv}}, n\in \mathbb{N}\}$ and $X^*$ are conditionally independent given the data and the sequence of reference points.
        \end{remark}
        
        Extend the asymptotic rate definitions by setting $\lambda_{\mathsf{a}, i}(x, v; \infty) := \lambda_{\text{can}, i}(x,v)$ for all $(x,v) \in \mathsf{E}$, $i = 1, \dots, d$, and $\mathsf{a} \in \{\text{ss}, \text{cv}\}$. The following proposition is a consequence of the strong law of large numbers for $U$-statistics (see Appendix \ref{sec:asymp_appendix} for a definition).

        \begin{propo}
        \label{propo:rates}
            Suppose Assumptions~\ref{assum:smoothness}, \ref{assum:moment}, and \ref{assum:reference} hold. 
            Let the subsample size sequence $1 \le m(n) \le n$ satisfy $m(n) \to m \in \mathbb{N} \cup \{\infty\}$ as $n \to \infty$. Then, for each $\mathsf{a} \in \{\text{ss}, \text{cv}\}$, each $i = 1, \dots, d$, and all compact $K\subset \mathbb{R}^d$,
            \[
                \sup_{(x, v) \in K \times \mathcal{V}}\left|\frac{\lambda^n_{\mathsf{a}, i}(x, v; m(n))}{n} - \lambda_{\mathsf{a}, i}(x, v; m)\right| \xrightarrow{n \to \infty} 0, \quad \Prob-\text{almost surely}. 
            \]
        \end{propo}
    
        \begin{proof}
            See Appendix \ref{sec:proof_propo}.
        \end{proof}   

        \begin{remark}[Convergence of canonical rates]
            Since $m(n)$ is integer valued, a finite $m$ means $m(n) = m$ eventually. Hence, Proposition~\ref{propo:rates} covers both fixed and growing subsample sizes. In particular, setting $m(n) = n$ recovers canonical rates and we get that $n^{-1}\lambda^n_{\text{can},i}$ also converges uniformly on compact sets to $\lambda_{\text{can}, i}$, $\mathbb{P}$-almost surely, as $n \to \infty$.
        \end{remark}

        



    %

    \subsection{Fluid limits}
    \label{sec:fluid}
        Proposition \ref{propo:rates} implies that the switching rates in each coordinate grow linearly in $n$. 
        A consequence of this is a multi-scale behaviour among components of Zig-Zag process: each coordinate of the {\it fast} velocity component oscillates at a rate of order $n$, while the {\it slow} position component has unit speed in each coordinate. 
        As $n$ goes to infinity, jumps in the velocity component become more rapid, while the distance traveled by the position component becomes infinitesimally small. On the usual time scale, these rapid oscillations average out, and the position component approaches a smooth deterministic trajectory. This motivates an effective ODE approximation obtained by averaging over the fast velocity dynamics, as we describe below.
            


        Fix $\mathsf{a} \in \{\text{can}, \text{ss}, \text{cv}\}$. 
        In case of $\mathsf{a} \in \{\text{ss}, \text{cv}\}$, fix the choice of subsample size sequence $1 \le m(n) \le n$ and let $m \ \mathbb{N} \cup \{\infty\}$ be such that $m(n) \to m$ as $n \to \infty$. 
        Throughout the following construction, we suppress the arguments $m(n)$ and $m$ in the finite-sample and asymptotic rates, respectively. We restore these arguments when comparing batch sizes. No subsample size parameter is needed for the canonical scheme.

    Define, for each $i = 1, \dots, d$,
    \begin{gather*}
        H_{\mathsf{a}, i} := \left\{x \in \mathbb{R}^d: \lambda_{\mathsf{a}, i}(x, -\bs{1}_d) + \lambda_{\mathsf{a},i}(x, \bs{1}_d) = 0 \right\}, \quad \text{and }
        H_{\mathsf{a}} := \mathbb{R}^d\setminus \cup_{i=1}^d H_{\mathsf{a}, i}. \numberthis \label{eq:setsH}
    \end{gather*}
    Define the asymptotic drift $b_{\mathsf{a}}: H_{\mathsf{a}} \to [-1, 1]^d$ coordinate-wise by,
    \begin{equation}
    \label{eq:limiting_drift}
        b_{\mathsf{a}, i}(x) := \frac{-\mathbb{E}_{P} [S_i(x, Y)]}{\lambda_{\mathsf{a}, i}(x, -\bs{1}_d) + \lambda_{\mathsf{a}, i}(x, \bs{1}_d)}, \quad x \in H_{\mathsf{a}}, \quad i = 1, \dots, d.
    \end{equation}

    Similarly to Remark~\ref{remark:omega_dependence} both, the asymptotic drift $b_{\mathsf{a}}$ and sets $H_{\mathsf{a},i}$, may depend implicitly on $\omega$ through the limiting reference point $X^*$, defined in Assumption~\ref{assum:reference}. While this dependence arises for $\mathsf{a} = \text{cv}$, for $\mathsf{a} \in \{\text{ss}, \text{can}\}$, these quantities do not depend on $X^*$. As usual, we suppress this dependence for notational convenience.
    For all $\mathsf{a}$, $b_{\mathsf{a}}$ is well defined and bounded between $-1$ and $1$ in each coordinate. Under Assumptions~\ref{assum:smoothness} and \ref{assum:moment}, and from \eqref{eq:asymp_rate_can}, \eqref{eq:asymp_rate_ss}, and \eqref{eq:asymp_rate_cv}, $b_{\mathsf{a}}$ is locally Lipschitz on $H_{\mathsf{a}}$. For all $x_0 \in H_{\mathsf{a}}$, there exists therefore a unique maximal solution $(X_{\mathsf{a}, t})_{0 \le t \le \tau_{\mathsf{a}}}$ to the initial value problem,
    \begin{equation}
    \label{eq:zz_ode}
        X_{\mathsf{a},0} = x_0, \quad \quad \frac{dX_{\mathsf{a},t}}{dt} = b_{\mathsf{a}}(X_{\mathsf{a}, t}), \quad \tau_{\mathsf{a}} = \inf\{t \ge 0: X_{\mathsf{a}, t} \notin H_{\mathsf{a}}\}.
    \end{equation}
    Maximal here refers to $\tau_{\mathsf{a}}$ and means that there is no other solution that stays in $H_{\mathsf{a}}$ for a longer time interval. We now state the main result of this section.

    \begin{theorem}
    \label{thm:transient}
        Suppose the assumptions of Proposition~\ref{propo:rates} hold. Suppose, $\mathbb{P}$-almost surely, $X^n_{\mathsf{a}, 0}$ converges weakly to some $x_{0} \in H_{\mathsf{a}}$, as $n \to \infty$. Let $(X_{\mathsf{a}, t})_{0 \le t \le \tau_{\mathsf{a}}}$ be the unique maximal solution to \eqref{eq:zz_ode}.
        Then $\Prob$-almost surely as $n \to \infty$, $({X}^n_{\mathsf{a}, t})_{0 \le t \le T}$ converges weakly (in uniform topology on $C([0, T], \mathbb{R}^d)$) to $(X_{\mathsf{a}, t})_{0 \le t \le T}$ for all $T < \tau_{\mathsf{a}}$.
    \end{theorem}
    \begin{proof}
        See Appendix \ref{sec:proof_transient}.
    \end{proof}

    \begin{remark}[Averaging principle and proof sketch]
    \label{rem:general}
        A general averaging framework for multi-scale systems based on first-order expansions of Markov generators is presented in \cite{pavliotis_multiscale_2008}, and recently applied to PDMPs in \cite{bierkens_scaling_2025, agrawal_transient_2025}. In this paper, we follow the approach in \cite{neal_optimal_2012, luczak_averaging_2013} and construct explicit couplings of the fast process to control the errors while averaging.
        
        We drop $\mathsf{a}$ from the subscript for this remark. Let $dt$ be a short duration of time such that $n^{-1} \ll dt \ll 1$, i.e., $dt \to 0$ and $n \cdot dt \to \infty$ as $n \to \infty$. In our proof, we take $dt = n^{-\delta}$ for some $1/2 < \delta < 1$. Then $dt$ is small enough such that during any time interval $(t, t + dt)$ of length $dt$, the position process can be frozen in the sense that $\lambda^n_i(X^n_{s}, v)ds \approx \lambda^n_i(X^n_t, v)ds + o(1)$ for $s \in (t, t + dt)$ for all $i$. For any $x \in H^n$, let $(V^{n,x}_{t})_{t \ge 0}$ be a pure jump process on $\{-1, +1\}^d$ with jump rates $q(v, \mathsf{F_i(v)}) = \lambda^n_i(x, v)$ for each $i$, and $0$ otherwise. Writing $x = X^n_t$, and starting from the current velocity, we can approximate,
        \[
            X^n_{t + dt} - X^n_{t} = \int_{t}^{t + dt} V^n_{s}\ ds \approx \int_{0}^{dt} V^{n,x}_{s}\ ds = \left(\frac{1}{n\cdot dt} \int_{0}^{n \cdot dt} V^{n,x}_{n^{-1}s}\ ds\right)dt.
        \]

        In the last equality above, the slowed down $V^{n,x}$ process oscillates at rates $n^{-1}\lambda^n_i$, which are $O(1)$ from Proposition~\ref{propo:rates}. Let $\pi^{n,x}$ be the unique invariant distribution of $V^{n, x}$ and $b^n(x) = \pi^{n,x}(V)$ be the expectation of $V$ with respect to $\pi^{n,x}$ (see \eqref{eq:drift_n} for an expression of $b^n$). Since $n \cdot dt \to \infty$, $dt$ is also large enough such that the term on the right-hand side can be approximated by the law of large numbers for $(V^{n,x}_{n^{-1}t})_{t \ge 0}$ giving,
        \[
            X^n_{t + dt} - X^n_t \approx b^n(X^n_t)dt. \numberthis \label{eq:heuristic}
        \]
        These arguments are made rigorous in Appendix~\ref{sec:proof_transient}.
        
        From the proof in Appendix~\ref{sec:proof_transient}, one may observe that the averaging mechanism and the ODE approximation in Theorem~\ref{thm:transient} is not special to the Bayesian setting considered in this paper. In general, if the switching rates $\{\lambda^n_i\}_{i=1}^d$ of Zig-Zag process satisfy $\lambda^n_i(x, v) = n^{\beta}\lambda_i(x, v) + o(1)$ for some $\beta > 0$ and $\{\lambda_i\}_{i=1}^d$, and are sufficiently smooth, i.e., locally Lipschitz with constant that is $O(n^{\beta})$, a similar result as Theorem~\ref{thm:transient} holds (see, e.g., \cite{genadot_averaging_2012, genadot_averaging_2023}). 
    \end{remark}

            Theorem~\ref{thm:transient} establishes smooth approximation for the Zig-Zag trajectories on every time interval $[0, T]$ with $T < \tau_{\mathsf{a}}$. 
    This is not a limitation in approximation due to the strong Markov property of Zig-Zag process; the fluid limits analysis can be re-applied after restarting the process at $\tau_{\mathsf{a}}$ provided the conditions of Theorem~\ref{thm:transient} are met. 

        At a point $x \in H_{\mathsf{a}, i}$, the limiting switching rates for the $i$-th velocity component become $0$ in both directions. Consequently, the frozen velocity process $V^{n,x}$ in Remark~\ref{rem:general} is no longer ergodic with unique invariant distribution, and the averaging argument does not determine how the limiting trajectory should continue. The finite-$n$ position process continues to move, and its paths may cross the vanishing rate sets or remain close to them for extended periods (see Figure \ref{fig:snaps}). Such behaviour can produce an effective process dynamics which stays on a lower dimensional space establishing which requires additional analysis. Consequently, the stability analysis of the sets $H_{\mathsf{a}, i}$ is beyond the scope of this paper; see \cite{agrawal_transient_2025} for related analysis outside the Bayesian context. 

        In the large-$n$ limit, the asymptotic drift $b_{\mathsf{a}}$ in \eqref{eq:limiting_drift} entirely characterizes the dynamics of Zig-Zag trajectories in the set $H_{\mathsf{a}}$. 
        By definition of $H_{\mathsf{a}}$, the denominator in \eqref{eq:limiting_drift} is always positive and the direction of the limiting flow is given by the negative expected score function. When the integration and differentiation can be interchanged, this is same as $-\nabla \mathbb{E}_{P} [-\log f(x; Y)]$ i.e. the limiting flow drifts in the direction of decreasing KL-divergence between $P$ and the model $\{F_x, x \in \mathbb{R}^d\}$. This is consistent with Zig-Zag process moving towards regions of greater posterior concentration. The denominator in \eqref{eq:limiting_drift} is the damping factor that quantifies the loss in speed due to subsampling. To understand the effect of different subsampling schemes, it is sufficient to understand how the denominator in \eqref{eq:limiting_drift} behaves in each situation.

    \subsection{Comparison of asymptotic drifts}
    \label{sec:compare_3}
        Using \eqref{eq:asymp_rate_can}, \eqref{eq:asymp_rate_ss}, and \eqref{eq:asymp_rate_cv}, we can write down the expressions for $b_{\mathsf{a}}$ and for different subsampling strategies.

        For canonical Zig-Zag, from \eqref{eq:asymp_rate_can}, each $H_{\text{can}, i} = \{x \in \mathbb{R}^d: \mathbb{E}_{P}S_i(x, Y) = 0\}$ is a hypersurface in the $d$-dimensional Euclidean space. The corresponding asymptotic drift is given coordinate-wise by,
        \begin{equation}
        \label{eq:can_drift}
            b_{\text{can}, i}(x) = \frac{-\mathbb{E}_{P}[S_i(x, Y)]}{\left|\mathbb{E}_{P}[S_i(x, Y)]\right|}, \quad x \in H_{\text{can}}, \ i = 1, \dots, d,
        \end{equation}
        where $H_{\text{can}} = \mathbb{R}^d \setminus \cup_{i=1}^d H_{\text{can}, i} = \{x \in \mathbb{R}^d: \mathbb{E}_{P}[S_i(x,Y)] \neq 0 \ \forall\ i = 1, \dots, d\}$.

        For $\mathsf{a} \in \{\text{ss}, \text{cv}\}$,  we now display the limiting subsample size in the notation and write $b_{\mathsf{a}}(\cdot \ ; m)$ for the asymptotic drift. Where needed, we similarly write $H_{\mathsf{a}}(m)$ for the domain defined from $\lambda_{\mathsf{a}, i}(\cdot , \cdot ; m)$. Suppose $m < \infty$.  For ZZ-SS, from \eqref{eq:asymp_rate_ss}, each $H_{\text{ss},i}(m) = \{x \in \mathbb{R}^d: m^{-1}\mathbb{E}_{P}|\sum_{j=1}^m S_i(x, Y_j)| = 0\} = \emptyset$ since $P$ is non-degenerate. Consequently, $H_{\text{ss}}(m) = \mathbb{R}^d$ and
        \begin{equation}
        \label{eq:ss_drift}
            b_{\text{ss}, i}(x; m) = \frac{-\mathbb{E}_{P}[S_i(x, Y)]}{\mathbb{E}_P\left[\left|m^{-1}\sum_{j=1}^mS_i(x, Y_j)\right|\right]}, \quad x \in H_{\text{ss}}(m) = \mathbb{R}^d,\ i = 1, \dots, d.
        \end{equation}
        
        For ZZ-CV, the active domain $H_{\text{cv}}(m)$ depends on the limiting reference point $X^*$. It may be as large as the whole space in ZZ-SS case or as small as a union of hypersurfaces in the canonical case, e.g., in the Gaussian location model, in which case the ZZ-CV switching rates are equivalent to the canonical Zig-Zag irrespective of the reference point (see Section~\ref{sec:examples}). From \eqref{eq:asymp_rate_cv},
        \begin{align*}
            b_{\text{cv}, i}(x; m) = \frac{-\mathbb{E}_{P}[S_i(x, Y)]}{\mathbb{E}_{P}\left[\left|m^{-1}\sum_{j=1}^m \left(S_i(x, Y_j) - S_i(X^*, Y_j) + \mathbb{E}_{P}S_i(X^*, Y)\right)\right|\right]}, \numberthis \label{eq:cv_drift}
        \end{align*}
        for all $x \in H_{\text{cv}}(m)$ and $i = 1, \dots, d$, where $H_{\text{cv}}(m)$ is the set of all $x$ for which the denominator is positive.

    When $m = \infty$, by Proposition~\ref{propo:rates}, $H_{\text{ss}}(\infty) = H_{\text{cv}}(\infty) = H_{\text{can}}$, and,
    \[
        b_{\text{ss}, i}(x; \infty) = b_{\text{cv}, i}(x; \infty) = b_{\text{can},i}(x), \quad x \in H_{\text{can}}.
    \]

    We see that because the limiting drift is explicitly confined to $[-1,1]^d$, the maximal speed is obtained only
    when it lies at a vertex in $\{-1,1\}^d$; we will hereafter use this as a characterization of an {\it optimal drift}. 
    For canonical Zig-Zag, \eqref{eq:can_drift} implies that the corresponding ODE has the optimal drift, i.e., $+1$ or $-1$ in each coordinate depending on the relative location of the processs with respect to the direction of decreasing KL-divergence.
    In general, Jensen's inequality implies that for $\mathsf{a} \in \{\text{ss}, \text{cv}\}$, and $m \in \mathbb{N} \cup \{\infty\}$,
    \[
        \lambda_{\mathsf{a},i}(x, -\bs{1}_d; m) + \lambda_{\mathsf{a},i}(x, \bs{1}_d; m) \ge |\mathbb{E}_{P}[S_i(x, Y)]|, \quad x \in \mathbb{R}^d, \ i = 1, \dots, d.
    \]
    Thus, $H_{\text{can}} \subseteq H_{\mathsf{a}}(m)$, and, $|b_{\mathsf{a}, i}(x; m)| \le |b_{\text{can}, i(x)}| = 1$ for all $x \in H_{\text{can}}$, and $i = 1, \dots, d$.
    Equality can occur; in particular, ZZ-CV achieves optimal drift in some models, as illustrated in Section~\ref{sec:examples}.

    \begin{remark}[Optimal subsample size]
    \label{rem:optimal_m}
        Let $\mathsf{a} \in \{\text{ss}, \text{cv}\}$. For any $m' > m$, Jensen's inequality and \eqref{eq:effective_rate} implies that $|b_{\mathsf{a}, i}(x; m')| \ge |b_{\mathsf{a}, i}(x; m)|$ for all $i = 1, \dots, d$, and all $x \in H_{\mathsf{a}}(m) \cap H_{\mathsf{a}}(m')$. 
        In fact, for each $x \in H_{\text{can}}$, $b_{\mathsf{a}}(x; m) \to b_{\text{can}}(x) $ as $m \to \infty$.
        Thus, increasing the subsample size increases the speed in each coordinate of the fluid limit. 
        
        On the other hand, the bound in Proposition~\ref{propo:mad_general} gives, for finite $m$, 
        \[
            |b_{\mathsf{a}, i}(x; m)| \le m|b_{\mathsf{a}, i}(x; 1)|, \quad x \in H_{\mathsf{a}}(m) \cap H_{\mathsf{a}}(1), \ i = 1, \dots, d.
        \]
        Thus, a subsample of size $m$ does not produce an $m$-fold increase in the drift. However, a subsample of size $m$ generally incurs $m$-times the computational cost at each proposed switching time, indicating that a subsample of size $1$ may be computationally optimal in the absence of additional implementation-level tricks (e.g.\ parallelization).   
    \end{remark}

    \begin{figure}[t]
    \centering
    \begin{minipage}{.49\linewidth}
        \centering
        \includegraphics[height = 8.6cm, width = \linewidth]{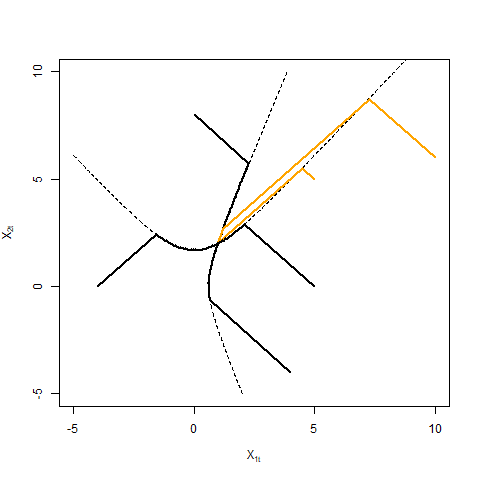}
        \captionof{figure}{Trajectories of the canonical Zig-Zag process in $2$ dimensions with different starting values. The dashed lines represent the hypersurfaces $H_{\text{can},i}, \ i = 1, 2$. Paths indicated in black stick to one of the hypersurfaces as soon as they hit it whereas those in orange don't. Colours were coded after the paths were observed.}
        \label{fig:snaps}
    \end{minipage}\hspace{5pt}
    \begin{minipage}{.49\linewidth}
        \centering
        \begin{subfigure}{\linewidth}
        \centering
            \includegraphics[width = \linewidth, height = 4.35cm]{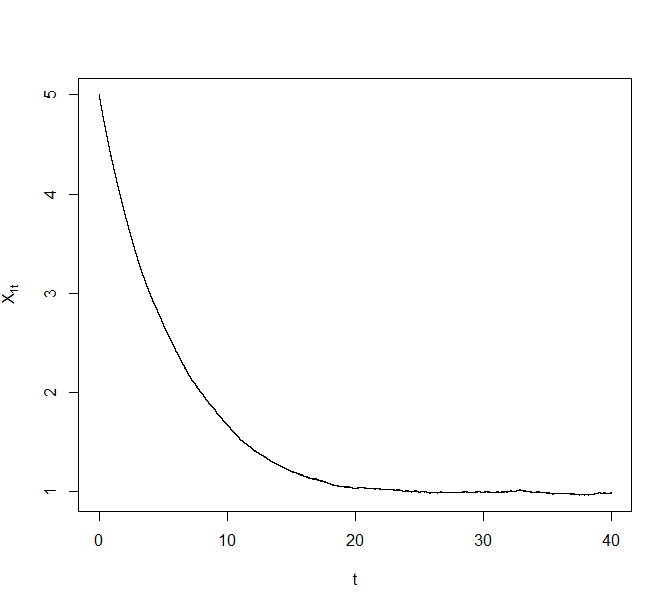}
        \end{subfigure}\\
        \begin{subfigure}{\linewidth}
        \centering
            \includegraphics[width = \linewidth, height = 4.35cm]{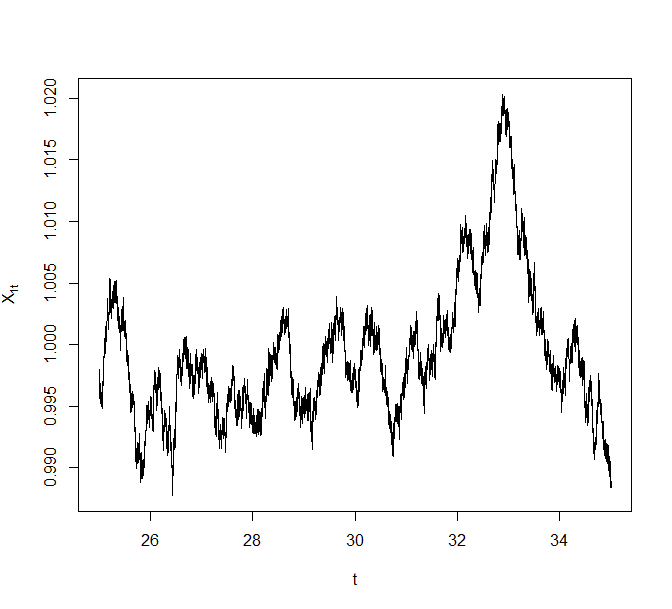}
        \end{subfigure}
        \captionof{figure}{First coordinate of the trajectory of the ZZ-SS algorithm targeting posterior for the Bayesian logistic regression model. The underlying process behaves like a diffusion around the true parameter value $1$.}
        \label{fig:different_phases}
        \begin{tikzpicture}[overlay]
             \draw (1.6,8.3) ellipse (0.8cm and 0.2cm);
             \draw (1.6,8.1) -- (1.0, 6.8);
        \end{tikzpicture}
    \end{minipage}
    \end{figure}


\section{The stationary phase}
\label{sec:stationary}
    As noted in Section~\ref{sec:intro}, the fluid limits only capture deviations in the Zig-Zag trajectories that are big enough on the usual spatial scale. Suppose Assumptions~\ref{assum:unique_minimiser} and \ref{assum:mle_consistent} hold, and let $\{\hat{x}_n\}_{n=1}^{\infty}$ be as in Assumption~\ref{assum:mle_consistent}. As $n$ becomes large, the posterior concentrates around $\hat{x}_n$ in the sense that for any $\epsilon > 0$, the posterior mass inside the ball $\{x : \|x - \hat{x}_n\| < \epsilon\}$ goes to $1$ in $\mathbb{P}$-probability as $n \to \infty$. The following proposition shows that, under stationary initialization, this concentration also confines the process trajectories to regions close to $\hat{x}_n$ over bounded time intervals (see Figure \ref{fig:different_phases} for illustration). The result applies in any fixed dimension to Zig-Zag process targeting $\Pi^{(n)}$. This includes the canonical Zig-Zag, ZZ-SS, and ZZ-CV.
    


    \begin{theorem}
    \label{thm:canonical}
        Suppose Assumptions~\ref{assum:unique_minimiser} and \ref{assum:mle_consistent} hold and let $x_0$ and $\{\hat{x}_n\}_{n=1}^{\infty}$ be as in Assumption~\ref{assum:unique_minimiser} and Assumption~\ref{assum:mle_consistent} respectively. For each $n$, let $(Z^n_t)_{t\ge 0} = (X^n_t, V^n_t)_{t \ge 0}$ be Zig-Zag process with stationary distribution $\Pi^n_{\mathsf{E}} := \Pi^{(n)} \otimes \text{Uniform}(\mathcal{V})$ and initial condition $Z^n_0 \sim \Pi^n_{\mathsf{E}}$. Then, for all $\epsilon > 0$ and $t \ge 0$,
        \[
            \lim_{n \to \infty}\Prob\left(\sup_{s \le t}|X^n_s - \hat{x}_n| > \epsilon\right) = 0.
        \]
        Consequently, $(X_t^n)_{t\ge0}$ converges in $\mathbb{P}$-probability, uniformly on compact time intervals, to the constant path $t\mapsto x_0$ as $n\to\infty$.
    \end{theorem}

    \begin{proof}
        See Appendix~\ref{sec:trivial_proof}.
    \end{proof}

    \begin{remark}
        In dimension $d=1$, suppose there exists $n_0$ such that, $\mathbb{P}$-almost surely, $U_n$ is convex and $\sum_{j=1}^n S(\hat{x}_n,Y_j)=0$ for every $n>n_0$. Under Assumptions~\ref{assum:smoothness}, \ref{assum:unique_minimiser}, and \ref{assum:mle_consistent}, the conclusions of Theorem~\ref{thm:canonical} also hold for the canonical Zig-Zag initialized at $X^n_0=\hat{x}_n$, with either initial velocity; see \cite[Theorem~4.1.1]{agrawal_large-sample_2025}. Convexity ensures that the process crosses $\hat{x}_n$ exactly once between each velocity switch. Successive returns to $(\hat{x}_n,V^n_0)$ define renewal events, and the proof exploits this renewal structure.
    \end{remark}

    
    Theorem~\ref{thm:canonical} establishes a trivial fluid limit in stationarity on the original time-scale.
    To obtain a non-trivial limit in stationarity the process needs to be analyzed on its natural scale, which depends on the rate of the posterior contraction. Consider the reparameterization, 
    \begin{equation}
    \label{eq:rescaling}
        \xi(x) := n^{1/2}(x - \hat{x}_n), \quad \quad x(\xi) = \hat{x}_n + n^{-1/2}\xi.
    \end{equation}
    Assumption~\ref{assum:mle_consistent} implies that the posterior distribution $\Pi^{(n)}$ in terms of $\xi$-parameter converges to a centered multivariate Gaussian distribution with covariance matrix $I(x_0)^{-1}$ as $n \to \infty$, where $I(x_0)$ is the Fisher-information matrix as defined in Assumption~\ref{assum:unique_minimiser}. Under this reparameterization, it is natural to expect that the Zig-Zag process converges to a non-trivial stochastic process which converges ergodically towards the same Gaussian distribution.

    Recall that, $Z^n_{\mathsf{a},t} = (X^n_{\mathsf{a}, t}, V^n_{\mathsf{a}, t})_{t \ge 0}$ is Zig-Zag process with switching rates $\{\lambda^n_{\mathsf{a}, i}\}_{i=1}^d$. 
    For each $n$ and $\mathsf{a} \in \{\text{ss}, \text{cv}\}$, define the rescaled position component $(\xi^n_{\mathsf{a}, t})_{t \ge 0}$ by
    \begin{equation}
    \label{eq:xint}
        \xi^n_{\mathsf{a},t} = n^{1/2}(X^n_{\mathsf{a}, t} - \hat{x}_n); \quad t \ge 0.
    \end{equation}
    Also, define
    \begin{equation}
    \label{eq:Unt}
        U^n_{\mathsf{a}, t} := (\xi^n_{\mathsf{a}, t}, V^n_{\mathsf{a},t}) = \left(n^{1/2}(X^n_{\mathsf{a},t} - \hat{x}_n), V^n_{\mathsf{a},t}\right); \quad t \ge 0.
    \end{equation}
    Then, $(U^n_{\mathsf{a},t})_{t \ge 0}$ is a Zig-Zag process with $i$-th switching rate $\lambda^n_{\mathsf{a}, i}(\xi, v) = \lambda^n_{\mathsf{a},i}(\hat{x}_n + n^{-1/2}\xi, v)$. In a slight abuse of notation, we will use $\lambda^n_{\mathsf{a},i}$ to denote switching rates in both $x$- and $\xi$-coordinate; the interpretation being clear from the passed argument. From \eqref{eq:Unt}, observe that in the $\xi$-coordinate, the Zig-Zag process has a speed $n^{1/2}$ in the sense that that the position component satisfies, $\xi^n_{\mathsf{a},t} - \xi^n_{\mathsf{a},0} = n^{1/2}\int_0^t V^n_{\mathsf{a},s}\ ds$ at all times. However, the switching rates for ZZ-SS and ZZ-CV in terms of $\xi$ may differ in magnitudes and hence require separate analysis. 

    \subsection{Zig-Zag with subsampling (without control-variates)}
    \label{sec:zzss_stat}
        We first consider ZZ-SS. From \eqref{eq:subsampling_rates} and \eqref{eq:zzss_ej1}, the $i$-th switching rate in the rescaled parameter is written as,
	    \begin{equation}
        \label{eq:rates_ss_stat}
	        \lambda^n_{\text{ss}, i}(\xi, v) = \frac{n}{m|\mathcal{S}_{(n, m)}|}\sum_{S \in \mathcal{S}_{(m, n)}} \left(v_i \cdot \sum_{j \in S} S_i(\hat{x}_n + n^{-1/2}\xi, Y_j)\right)_{+} .
	    \end{equation}
        From Assumption~\ref{assum:mle_consistent}, $\hat{x}_n + n^{-1/2}\xi \to x_0$  in $\mathbb{P}$-probability for all $\xi$. Then, using similar arguments as in the proof of Proposition~\ref{propo:rates}, it can be shown that for large $n$, the above quantity is $O_p(n)$. Thus, in $(\xi, v)$-coordinate, even though the position component moves with speed $n^{1/2}$, the velocity component still oscillates at rate which is $O(n)$. This is, once again, a situation of multi-scale dynamics as in the last section. It is tempting to scale down the time by a factor $n^{1/2}$ and and carry out the averaging as in Remark~\ref{rem:general}, but this will lead to a trivial fluid limit as we will now argue.

        After slowing down the process by a factor $n^{1/2}$, the fast component oscillates at a rate $O(n^{1/2})$. Let $dt$ be the short time duration as in Remark~\ref{rem:general} with $n^{-1/2} \ll dt \ll 1$. Then, from \eqref{eq:heuristic} and \eqref{eq:Unt}, dropping the subscript,
        \[
            \xi^n_{n^{-1/2}(t + dt)} - \xi^n_{n^{-1/2}t} = \int_t^{t + dt} V^n_{n^{-1/2}s}\ ds
            \approx b^n(X^n_{n^{-1/2}t}) = b^n(\hat{x}_n + n^{-1/2}\xi^n_{n^{-1/2}t})dt.
        \]
        
        If $\xi^n_{n^{-1/2}t} = O_p(1)$, then by Lemma~\ref{lemma:zz_ss_1} the right-hand side is $O_p(n^{-1/2})$ and converges to $0$. 
        This happens precisely because there is more symmetry in the target distribution around $\hat{x}_n$ which makes all directions equally likely for the process to move in. Indeed, consider how the event rates change under flipping the $i$-th coordinate of the velocity,
	    \begin{align*}
	        \lambda^n_{\text{ss}, i}(\xi, v) - \lambda^n_{\text{ss}, i}(\xi, -v) &= v_i \left(\sum_{j=1}^n S_i(\hat{x}_n + n^{-1/2}\xi, Y_j)\right) \\
	        &\approx o_p(n^{1/2}) + n^{1/2}v_i\left(\xi \cdot n^{-1}\sum_{j=1}^n \nabla S_i(\hat{x}_n, Y_j)\right),
	    \end{align*}
        due to Assumption~\ref{assum:mle_consistent}. By Lemma~\ref{lemma:In_convergence}, the average in the second term converges to a finite positive quantity giving us that $\lambda^n_{\text{ss}, i}(\xi, v) - \lambda^n_{\text{ss}, i}(\xi, -v) = O_{p}(n^{1/2})$. On the other hand, $\lambda^n_{\text{ss}, i}(\xi, v) + \lambda^n_{\text{ss}, i}(\xi, -v) = O_{p}(n)$. Combining the two gives $b^n(\hat{x}_n + n^{-1/2}\xi) = O_p(n^{-1/2})$.\\
        
        Now, consider the following heuristics. Let $dt$ be a small time duration as in Remark~\ref{rem:general}. The average displacement in $\xi^n_t$ during $(t, t+dt)$ is, from \eqref{eq:heuristic}  and \eqref{eq:Unt},
        \[
            \mathbb{E}\left(\xi^n_{t + dt} - \xi^n_{t}\right) = n^{1/2}\mathbb{E}(X^n_{t + dt} - X^n_t) \approx n^{1/2}\mathbb{E}b^n(\hat{x}_n + n^{-1/2}\xi^n_t)dt = O_p(1)dt, 
        \]
        by Lemma~\ref{lemma:zz_ss_1}. Let $(V^{n,x}_t)_{t \ge 0}$ be the replacement velocity process as in Remark~\ref{rem:general}. For each $i$, from \eqref{eq:heuristic} and \eqref{eq:Unt},
        \[
            \mathbb{E}(\xi^n_{t + dt, i} - \xi^n_{t, i})^2 \approx \mathbb{E} \left(\frac{1}{(n\cdot dt)^{1/2}}\int_{0}^{n \cdot dt} {V}^{n, X^n_t}_{n^{-1/2}s, i} \ ds\right)^2dt \approx O_{p}(1)dt,
        \]
        by the $L^2$-convergence of exponentially ergodic Markov processes. Thus, in the $\xi$-parameter, the above observations imply that the infinitesimal drift and the infinitesimal variance in the position component are both $O_p(1)$. This suggests that the process moves on a diffusive scale. \\
        
        More precisely, we have the following result.

        \begin{theorem}
        \label{theo:zzss_stat_fixed}
            Suppose Assumptions \ref{assum:smoothness}, \ref{assum:unique_minimiser}, and \ref{assum:mle_consistent} hold. Suppose the subsample size $m(n) \to m$ for some fixed $m < \infty$ as $n \to \infty$. Let $(U^n_{\text{ss}, t})_{t \ge 0} = (\xi^n_{\text{ss}, t}, V^n_{\text{ss}, t})$ be the ZZ-SS process defined in \eqref{eq:Unt}. Let $I(x_0)$ be as in the Assumption~\ref{assum:unique_minimiser} and let $A$ be a $d \times d$ diagonal matrix with
            \[
                A_{ii} = \frac{2}{\mathbb{E}_{P}\left[\left|m^{-1}\sum_{j=1}^m S_i(x_0, Y_j)\right|\right]}, \ Y_1, \dots, Y_m \overset{\text{iid}}{\sim} P, \quad i = 1, \dots, d. \numberthis \label{eq:A_def}
            \]
            Suppose $\text{Law}(\xi^n_{\text{ss}, 0}) \to \nu$ weakly and that the Ornstein-Uhlenbeck SDE, 
            \begin{equation}
            \label{eq:limiting_sde}
                d\xi_t = -\frac{A I(x_0)}{2}\ \xi_t dt + A^{1/2}\ dB_t,
            \end{equation}
            with initial condition $\xi_0 \sim \nu$ has a unique weak solution for $t \ge 0$. 
            Then, as $n \to \infty$, $(\xi^n_{\text{ss}, t})_{t \ge 0}$ converges weakly (in Skorohod topology) in $\Prob$-probability to the solution $(\xi_t)_{t \ge 0}$ of \eqref{eq:limiting_sde}.
        \end{theorem}

        \begin{proof}
            See Appendix~\ref{sec:proof_zzss}.
        \end{proof}
	
      \begin{remark}
        \label{rem:zzss_complexity}
            The limiting diffusion is reversible with respect to the limiting Gaussian distribution $\Pi = N(0,I(x_0)^{-1})$ in the Bernstein--von Mises theorem (Assumption~\ref{assum:mle_consistent}). Its $L^2(\Pi)$-spectral gap determines the optimal exponential rate of convergence to equilibrium; a larger gap therefore implies faster $L^2$-convergence, see, e.g., \cite{pavliotis_stochastic_2014}.
            
            Write $A(m)$ for the matrix in \eqref{eq:A_def} to make its dependence on the subsample size explicit. The spectral gap of the Markov generator associated with \eqref{eq:limiting_sde} is
            \[
                g_m
                = \frac12 \lambda_{\min}\!\left(
                    I(x_0)^{1/2}A(m)I(x_0)^{1/2}
                  \right)
                = \frac12 \lambda_{\min}\!\left(A(m)I(x_0)\right),
            \]
            where $\lambda_{\min}$ denotes the smallest eigenvalue; see \cite[Theorem~3.1]{metafune_spectrum_2002}. The second equality follows because the two matrices are similar. From \eqref{eq:A_def},
            \[
                \min_{i}\left(\frac{m}{\mathbb{E}_{P}\left[\left|\sum_{j=1}^m S_i(x_0,Y_j)\right|\right]}\right) \le \frac{g_m}{\lambda_{\min}(I(x_0))} \le \max_i\left(\frac{m}{\mathbb{E}_{P}\left[\left|\sum_{j=1}^m S_i(x_0,Y_j)\right|\right]}\right).
            \]
            Thus, $A(m)$ is a damping factor and quantifies loss in mixing due to subsampling.

            Proposition~\ref{propo:mad_general} and the triangle inequality give
            \[
                \mathbb{E}_P\left[|S_i(x_0,Y_1)|\right] \le \mathbb{E}_P\left[\left|\sum_{j=1}^m S_i(x_0,Y_j)\right|\right] \le m\mathbb{E}_P\left[|S_i(x_0,Y_1)|\right], \quad i=1,\ldots,d.
            \]
            Combining these observations yields,
            \[
                g_1\le g_m\le mg_1.
            \]
            Hence, increasing the batch size from $1$ to $m$ improves the spectral gap by at most a factor of $m$. When computational cost per unit of rescaled time is proportional to $m$, the spectral gap per unit cost satisfies $g_m/m \le g_1$. Consequently, $m=1$ is optimal under this criterion in agreement with Remark~\ref{rem:optimal_m}.
	\end{remark}

       The weak convergence result in Theorem~\ref{theo:zzss_stat_fixed} complexity bounds for the finite-$n$ Zig-Zag paths. For any two signed measures $\mu$, $\nu$ on $\mathbb{R}^d$ define the $\|\dots\|_{\text{KR}}$ distance \citep{roberts_complexity_2016} as,
        \begin{equation*}
            \|\mu - \nu \|_{\text{KR}} = \sup_{f \in \text{Lip}_1^1}|\mu(f) - \nu(f)|, \numberthis \label{eq:kr_distance}
        \end{equation*}
        where,
        \[
            \text{Lip}_1^1 = \{f: \mathbb{R}^d \to \mathbb{R}, |f(x) - f(y)| \le \|x - y\| \text{ for all } x, y \in \mathbb{R}^d, \ |f| \le 1 \}.
        \]
        The distance $\|\dots\|_{\text{KR}}$ is bounded by $2$, and metrises weak convergence of probability measures on separable metric spaces endowed with their Borel $\sigma$-algebras \cite{roberts_complexity_2016}. It is the $1$-Wasserstein distance corresponding to the truncated ground metric $\min\{2, \|x - y\|\}$ (see \cite{roberts_complexity_2016}). Let $\Pi^{(n)}_{\sharp \xi}$ denote the push-forward of $\Pi^{(n)}$ under the rescaling \eqref{eq:rescaling} and let 
        \[
            \Pi^{(n)}_{\sharp \mathsf{E}} = \Pi^{(n)}_{\sharp\xi} \otimes \text{Uniform}(\mathcal{V}). 
        \]
        The following result is obtained from a generalization of \cite[][Theorem~1]{roberts_complexity_2016}, see Appendix~\ref{sec:complexity} for details.
        
        \begin{corollary}
        \label{corr:complexity_ss}
            Under the assumptions of Theorem~\ref{theo:zzss_stat_fixed}, for any $\epsilon > 0$ and $\delta > 0$, there exist $T < \infty$ and $N < \infty$, such that,
            \[
                \mathbb{P}\left(\mathbb{E}_{U^n_{\text{ss}, 0} \sim \Pi^{(n)}_{\sharp \mathsf{E}}}\left\|\mathrm{Law}(\xi^n_{\text{ss}, t} \mid U^n_{\text{ss}, 0}) - \Pi^{(n)}_{\sharp \xi}\right\|_{\text{KR}} < \epsilon \ \forall \  t \ge T\right) > 1- \delta, \quad n \ge N.
            \]
        \end{corollary}

        \begin{proof}
            See Appendix~\ref{sec:proof_complexity_stat}.
        \end{proof}

    \subsection{Zig-Zag with Control Variates}
    \label{sec:zzcv_stat}
        In this section, we consider the ZZ-CV process in the $(\xi, v)$-coordinate. From \eqref{eq:subsampling_rates} and \eqref{eq:zzcv_ej1}, the $i$-th switching rate in $\xi$-parameter is,
        \begin{equation}
        \label{eq:cv_rate_xi}
            \lambda^n_{\text{cv}, i}(\xi, v) = \frac{n}{m|\mathcal{S}_{(n,m)}|}\sum_{S \in \mathcal{S}_{(n,m)}}\left(v_i \cdot \sum_{j \in S}E_{\text{cv}, i}^j(\xi)\right)_{+},
        \end{equation}
        where
        \begin{align*}
            E_{\text{cv}, i}^j(\xi) &= S_i(\hat{x}_n + n^{-1/2}\xi, Y_j) - S_i(X_n^*, Y_j) + n^{-1}\nabla U_n(X_n^{*}) \\
            &\approx (n^{-1/2}\xi + \hat{x}_n - X^*_n) \cdot \nabla S_i(\hat{x}_n, Y_j) + n^{-1}\nabla U_n(X_n^{*}),
        \end{align*}
        and $(X_n^*)_{n=1}$ is a random sequence of reference points defined by \eqref{eq:reference_points}.  
        
        The magnitude of the switching rates depends explicitly on the sequence of reference points and their distance to $\hat{x}_n$. As mentioned in Remark~\ref{remark:ref_points}, a reasonable choice for $T_n$ is to ensure $X^*_n$ is sufficiently close to $\hat{x}_n$. Suppose that the reference points are chosen such that $\|X^*_n - \hat{x}_n\| = O_{p}(n^{-1/2})$. It would then imply that the ZZ-CV switching rates in $\xi$-coordinate are of magnitude $O(n^{1/2})$. This would mean that, in the $(\xi, v)$-coordinate, the position and the velocity process mix at the same rate which is $O(n^{1/2})$. Slowing down time by a factor $n^{-1/2}$ brings both the components back to the unit time scale. Consequently, the limiting process is a Zig-Zag process with an appropriate switching rate. As in the transient phase, we make the following assumption about the sequence of reference points.
        \begin{assum}
        \label{assum:reference_stat}
            There exists a $\xi^*$ defined on $\Omega$ such that $\mathbb{E}\|\xi^*\|^2 < \infty$, $\xi^*$ is independent of $\bs{Y}^{(\mathbb{N})}$, and $n^{1/2}(X^*_n - \hat{x}_n) \to \xi^*$, in $\Prob$-probability. 
        \end{assum}
        
        \begin{figure}
        \centering
            \begin{subfigure}{0.325\linewidth}
                \includegraphics[scale = 0.25]{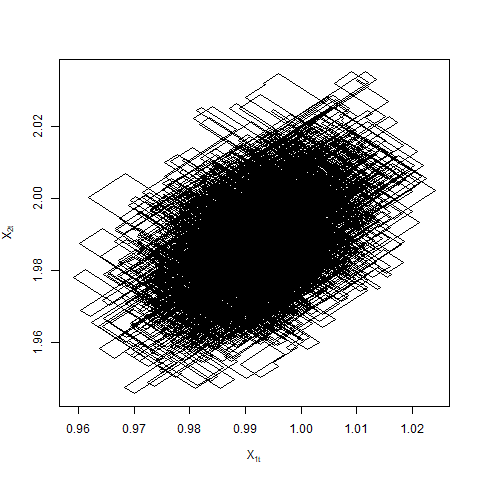}
            \subcaption{$X^* = (1, 2)$}
            \end{subfigure}
            \begin{subfigure}{0.325\linewidth}
                \includegraphics[scale = 0.25]{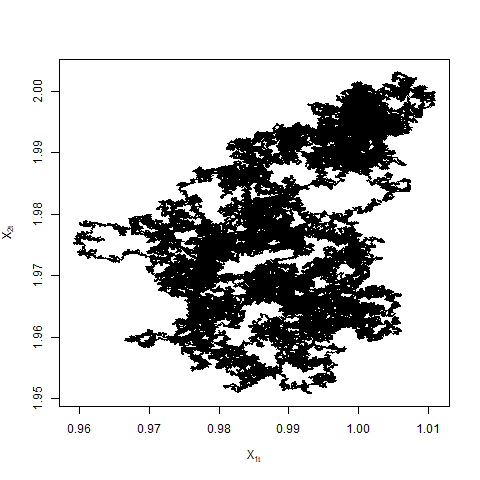}
            \subcaption{$X^* = (-10, 1)$}
            \end{subfigure}
            \begin{subfigure}{0.325\linewidth}
                \includegraphics[scale = 0.25]{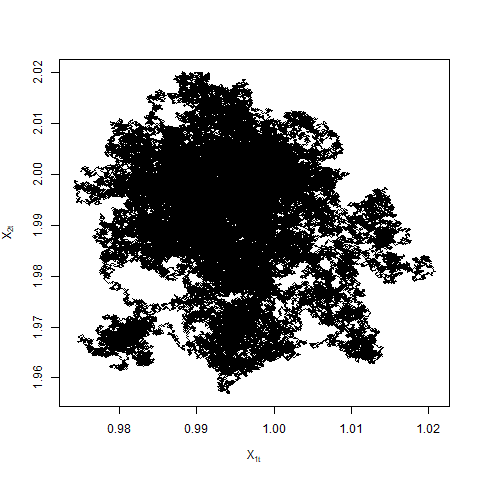}
            \subcaption{$X^* = (5, 5)$}
            \end{subfigure}
            \caption{Stationary phase trajectories of ZZ-CV for Bayesian logistic regression example with true parameter $(1, 2)$ with different reference points.}
        \label{fig:different_refs}
        \end{figure}
	    

        We have the following result.
        
        \begin{theorem}
        \label{theo:zzcv_stat_fixed}
            Suppose Assumptions~\ref{assum:smoothness}, \ref{assum:unique_minimiser}, \ref{assum:mle_consistent}, and \ref{assum:reference_stat}, all hold. Suppose the subsample size $1 \le m(n) \le n$ and $m(n) \to m$ for some fixed $m < \infty$ as $n \to \infty$. Let $(U^n_{\text{cv}, t})_{t \ge 0} = (\xi^n_{\text{cv}, t}, V^n_{\text{cv}, t})$ be the ZZ-CV process defined in \eqref{eq:Unt} and $\xi^*$ be as in Assumption~\ref{assum:reference_stat}. Suppose $U^n_{\text{cv}, 0} \sim \mu^n$ for all $n$ and $\mu^n \to \mu$ weakly. Then $(U^n_{\text{cv}, n^{-1/2}t})_{t \ge 0}$ converges weakly (in Skorohod topology) in $\Prob$-probability to a $d$-dimensional Zig-Zag process $(U_t)_{t \ge 0}$ with $U_0 \sim \mu$, and $i$-th switching rate,
            \begin{align}
                \lambda_{\text{cv}, i}(\xi, v; m) = m^{-1}\mathbb{E}_{P}\left[\left(v_i \sum_{j=1}^m\left((\xi - \xi^*)\cdot \nabla S_i(x_0, Y_j) + \xi^* \cdot I_{\cdot i}(x_0)\right)\right)_{+}\right], \numberthis \label{eq:limrate}
            \end{align}
            where $I_{\cdot i}(x_0)$ denotes the $i$-th column of $I(x_0)$ and $I(x_0)$ is as in Assumption~\ref{assum:unique_minimiser}.
        \end{theorem}

        \begin{proof}
            See Appendix \ref{sec:proof_zzcv}.
        \end{proof}

        Once again with a slight abuse of notation, we use $\{\lambda_{\text{cv},i}\}_{i=1}^d$ to denote asymptotic switching rates in both the $x$- and $\xi$-parameters. The asymptotic rate in \eqref{eq:limrate} is not to be confused with the definition in \eqref{eq:asymp_rate_cv}. 

        \begin{remark}
        \label{rem:bad_reference}
            The difference in the switching rates in opposite directions for ZZ-CV is the same as ZZ-SS i.e. $O(n^{1/2})$. The limiting dynamics are then governed only by the magnitude of the sum of the switching rates. The Assumption~\ref{assum:reference_stat} cannot be dropped from the statement of Theorem~\ref{theo:zzcv_stat_fixed}. In general, suppose $\|X^*_n - \hat{x}_n\| = O_{p}(n^{-\alpha})$ for some $\alpha \ge 0$, then the switching rates or the sum of the switching rates are of the magnitude $\max\{1/2, 1- \alpha\}$. In particular when $\alpha = 0$, the sum of the rates is $O(n)$ and the process is once again in the diffusive scale as in the case of ZZ-SS. We do not pursue this any further in this paper but see Figure~\ref{fig:different_refs} for examples of Zig-Zag trajectories when $X^*_n \not\to \hat{x}_n$.
        \end{remark}  
        
        As in the case of fluid limits, the asymptotic switching rate is a random element on $\Omega$ and depends on $\omega$ through $\xi^*$. As a consequence, the limiting process is a random Zig-Zag process depending on the statistic $T_n$ used to generate reference points. If $T_n$ is determinisitic and such that $\|X^*_n - \hat{x}_n\| = O_{p}(n^{-\alpha})$ for $\alpha > 1/2$, then $\xi^* = 0$ almost surely. The limiting switching rate reduces to,
        \[
            \lambda_{\text{cv}, i}(\xi, v; m \mid \xi^* = 0) = m^{-1}\mathbb{E}_{P}\left(v_i \cdot \sum_{j=1}^m\xi \cdot \nabla S_i(x_0, Y_j)\right)_{+}.
        \]
        In general, the limiting Zig-Zag process is not canonical. We have the following remark similar to Remark~\ref{rem:zzss_complexity}.
        
       \begin{remark}
        \label{rem:zzcv_complexity}
            From \eqref{eq:limrate}, $\lambda_{\text{cv}, i}(\xi, v; m) - \lambda_{\text{cv}, i}(\xi, -v; m) = v_i\left(\xi \cdot I_{\cdot i}(x_0)\right)$, $\mathbb{P}$-almost surely for all $\xi$ and $i$.
            Thus, the position component of the limiting Zig-Zag process is invariant with respect to the limiting Gaussian distribution $\Pi = N(0,I(x_0)^{-1})$ in the Bernstein--von Mises theorem (Assumption~\ref{assum:mle_consistent}).

            Using Jensen's inequality in \eqref{eq:limrate}, for all fixed $m$,
            \[
                \lambda_{\text{cv}, i}(\xi, v; m) \ge \left(v_i (\xi\cdot I_{\cdot i}(x_0))\right)_{+}
            \]
            for all $(\xi, v)$ and $i$. The right-hand side is the $i$-th canonical rate corresponding to the target distribution $\Pi$. Let $\gamma_{m, i}(\xi) = \lambda_{\text{cv}, i}(\xi, v; m) - \left(v_i (\xi\cdot I_{\cdot i}(x_0))\right)_{+}$. The above implies $\gamma_{m, i}(\xi) \ge 0$, and using $(a)_+ = (|a| +a)/2$ and Jensen's inequality, $\gamma_{m+1, i} \le \gamma_{m, i}$ for all $m$. Increasing the subsample size therefore reduces excess switching in the Zig-Zag process. However, also by Proposition~\ref{propo:mad_general}, 
            \[
                \frac{\gamma_{1, i}(\xi) + |\xi\cdot I_{\cdot i}(x_0)|}{m} \le \gamma_{m, i}(\xi) + |\xi\cdot I_{\cdot i}(x_0)| \le \gamma_{1, i}(\xi) + |\xi\cdot I_{\cdot i}(x_0)|, \numberthis \label{eq:ordering_switching_rate}
            \]
            for all $\xi$ and $i$. Thus, an $m$-fold increase in the subsample size does not produce an $m$-fold reduction in the excess switching rate. 

            When $d = 1$, the ordering in \eqref{eq:ordering_switching_rate} implies an ordering in the asymptotic variances for the corresponding Zig-Zag processes, with lower excess switching rate implying smaller asymptotic variance \cite{andrieu_peskuntierney_2021, bierkens_limit_2017}. As in Remark~\ref{rem:zzss_complexity}, it can then be argued that $m = 1$ is computationally optimal subsample size. We conjecture that the same holds for $d \ge 2$, however an ordering akin to $d=1$ is still an open problem.
        \end{remark}

        The following result is for completeness. It shows that if $m(n) \to \infty$ as $n \to \infty$, the limiting process does not depend on the reference points in ZZ-CV and is equivalent to the canonical Zig-Zag. Once again, the limiting process is invariant with respect to the limiting Gaussian distribution.

	\begin{theorem}
	\label{theo:zzcv_stat_varying}
	        In Theorem~\ref{theo:zzcv_stat_fixed}, suppose that the subsample size $m(n) \to \infty$ as $n \to \infty$. Then the limiting process $(U_t)_{t\ge 0}$ is a Zig-Zag process with switching rates
	        \[
	            \lambda_{\text{cv}, i}(\xi, v; \infty) = (v_i (\xi \cdot I_{\cdot i}(x_0)))_{+}. \numberthis \label{eq:limrate_can}
	        \]
	\end{theorem} 

        \begin{proof}
            See Appendix~\ref{sec:proof_zzcv}.    
        \end{proof}

        Finally, we have the analogue of Corollary \ref{corr:complexity_ss} for ZZ-CV.

        \begin{corollary}
        \label{corr:complexity_cv}
            Let $(U^n_{\text{cv}, t})_{t \ge 0}$ be as in Theorem \ref{theo:zzcv_stat_fixed} or \ref{theo:zzcv_stat_varying} and suppose the technical assumptions are satisfied. Then, for any $\epsilon > 0$ and $\delta > 0$, there exist $T < \infty$ and $N < \infty$, such that,
            \[
                \mathbb{P}\left(\mathbb{E}_{U^n_{\text{cv}, 0} \sim \Pi^{(n)}_{\sharp \mathsf{E}}}\left\|\mathrm{Law}(\xi^n_{\text{cv}, n^{-1/2}t} \mid U^n_{\text{cv}, 0}) - \Pi^{(n)}_{\sharp \xi}\right\|_{\text{KR}} < \epsilon \ \forall \  t \ge T\right) > 1- \delta, \quad n \ge N.
            \]
        \end{corollary}

        \begin{proof}
           See Appendix~\ref{sec:proof_complexity_stat}.  
        \end{proof}

    \subsection{Computational cost}
    \label{sec:cost}
        Corollary \ref{corr:complexity_ss} implies that in the stationary phase the ZZ-SS process is able to decorrelate on time scales of constant duration, even as the sample size $n$ grows. This means, as $n \to \infty$, the corresponding continuous time algorithm requires $O(1)$ time interval to obtain an essentially independent sample. On the other hand, Corollary \ref{corr:complexity_cv} gives time intervals of $O(n^{-1/2})$ for canonical Zig-Zag and ZZ-CV to obtain an essentially independent sample in the stationary phase.
        
        For fixed subsample size, say $m(n) = 1$, the switching rates in the rescaled coordinates are $O_p(n)$ for ZZ-SS and $O_p(n^{1/2})$ for canonical Zig-Zag and ZZ-CV. These orders follow from the calculations in Sections~\ref{sec:technical_ss} and \ref{sec:technical_cv} respectively. The overall computational cost of a Zig-Zag algorithm depends on the choice of computational bounds $(M_i^n)_{i=1}^d$ used for Poisson thinning (see Section \ref{sec:zz_bayes_1}). If the computational bounds are of the same respective order as the switching rates, then the cost estimates in Table~\ref{tab:summary} follow after the initialization and setup costs. Indeed from the estimates above, the ZZ-SS algorithm proposes $O(n)$ number of events, each costing $O(1)$, per an essentially independent sample. For canonical Zig-Zag and ZZ-CV on the other hand, the same is achieved using $O(1)$ number of proposals, each costing $O(n)$ for canonical Zig-Zag and $O(1)$ for ZZ-CV. This leads to an overall cost of $O(n)$ for ZZ-SS and canonical Zig-Zag and $O(1)$ for ZZ-CV in the stationary phase after initialization costs. As noted in \cite{bierkens_zig-zag_2019}, $O(n)$ is the best we can expect for any standard Monte Carlo algorithm. Thus, ZZ-CV provides a factor $n$ improvement in efficiency compared to standard MCMC methods. This observation is empirically supported by experiments in \cite{bierkens_zig-zag_2019}.

        In practice, however, the computational bounds will not necessarily be of the same order as the switching rates. For example, when the gradient of the model log-density is globally bounded, the computational bounds scale as $O(n)$ for each of the three algorithms. The total cost in stationarity then scales as $O(n^{3/2})$ for canonical Zig-Zag, $O(n)$ for ZZ-SS, and $O(n^{1/2})$ for ZZ-CV. On the other hand, consider the logistic regression example presented Section~\ref{sec:blr}. When the weight distribution is (sub-)Gaussian, the computational bounds scale as $O(n\sqrt{\log n})$ for ZZ-SS and $O(n\log n)$ for canonical Zig-Zag and ZZ-CV \cite{bierkens_supplement_2019}. The total complexity in stationarity then scales as $O(n^{3/2}\log n)$ for canonical Zig-Zag, $O(n\sqrt{\log n})$ for ZZ-SS, and $O(n^{1/2}\log n)$ for ZZ-CV. In general, the magnitude of the tightest computational bounds will depend, in addition to the subsampling scheme, on the true data-generating distribution and the chosen model.

        Finally, as argued in Remark~\ref{rem:bad_reference}, the realization of $O(n)$ gain in efficiency due to ZZ-CV is conditional on the reference point mechanism satisfying $\|X^*_n - \hat{x}_n\| = O_{p}(n^{-1/2})$. Computing such a reference point is an expensive procedure relying on numerical optimization routines, for example using stochastic gradient methods, the complexity of which is typically $O(n)$. 
        Additionally, there will be a one-off upfront cost of calculating $\nabla U_{n}(X_n^*)$ leading to a total initialization cost of $O(n)$. Once a good reference point has been found, a natural initialization for the algorithm is $X^n_0 = X^*_n + n^{-1/2}Z$, for $Z \sim N(0, I_d)$. This keeps the ZZ-CV in the stationary phase and will ensure the promised $O(n)$ improvement per an independent sample.


\section{More on fluid limits}
\label{sec:examples}
    In this section, we go back to the transient phase and consider the fluid limits in more detail. 

    \subsection{Toy examples}

        We illustrate fluid limits in some toy models in one dimension. Let $X^*$ be as in the Assumption~\ref{assum:reference}.

        
        \begin{example}[Gaussian location model]
            Suppose $f(y; x) \propto \exp(-(y - x)^2/2) $ and $P = N(x_0, 1)$. The score function $S(x, y) = (x - y)$ is linear in both $x$ and $y$. As a result,
            \[
                S(x, Y) - S(X^*, Y) + \mathbb{E}_PS(X^*, Y) = x - x_0, \quad \mathbb{P}\text{-almost surely,}
            \]
            for all $x \in \mathbb{R}$. From \eqref{eq:can_drift} and \eqref{eq:cv_drift}, we get,
            \[
                b_{\text{cv}}(x; m) = b_{\text{can}}(x) = \frac{-(x - x_0)}{|x - x_0|},  \quad x \neq x_0,
            \]
            for all $m$. Moreover, from \eqref{eq:ss_drift}, 
            \[
                b_{\text{ss}}(x; m) = \frac{-(x - x_0)}{\mathbb{E}_P\left|x - m^{-1}\sum_{j=1}^m Y_j\right|}, \quad x \in \mathbb{R},
            \]
            where $Y_1, \dots Y_m \sim N(x_0, 1)$.
        \end{example} 

        In the Gaussian location model example above, the linearity of the score function yields equivalence of ZZ-CV and canonical Zig-Zag in finite samples as well, i.e, for all $n$ and $1 \le m(n) \le n$, $\lambda^n_{\text{cv}}(x, v; m(n)) \equiv \lambda^n_{\text{can}}(x,v)$, irrespective of the choice of the reference points. For general one-dimensional models with log-concave densities, ZZ-CV will always be equivalent to the canonical Zig-Zag provided the reference point solves the likelihood equation; see \cite[][Proposition~3.4.1]{agrawal_large-sample_2025}.




        \begin{example}[Laplace model]
            Suppose $P = \text{Laplace}(x_0, 1)$ and $f(y; x) \propto \exp(-|y - x|)$. The log-density is not continuously differentiable on a null set, and hence violates the smoothness assumption needed for fluid limits. However, defining $S(x, y)$ to be the weak derivative with respect to $x$ of $|x - y|$, we get,
            \[
                S(x, y) = \begin{cases}
                    \mathrm{sgn}(x - y); & y \neq x, \\
                    0; & y = x.
                \end{cases}
            \]
            Let $p(x) = P(Y \le x) = F_{x_0}(x)$. From \eqref{eq:ss_drift},
            \[
                b_{\text{ss}}(x; m) = \frac{-(2p(x) - 1)}{\mathbb{E}|2m^{-1}Y - 1|}; \quad Y \sim \text{Binomial}(m, p(x)), \quad x \in \mathbb{R}.
            \] 
            Given the form $S$, note that for any $x, x'$, $S(x, y) - S(x', y)$ only depends on the direction of $x$ relative to $x'$. From \eqref{eq:cv_drift}, the denominator term in $b_{\text{cv}}$ only depends on the location of $X^*$ and the direction of $x$ relative to $X^*$. Let $m =1$. Let $x^*$ be the observed value of $X^*$. When $p(x^*) < 1/2$, i.e. $x^* < x_0$, the drift is given by,
            \[
                b_{\text{cv}}(x; 1) = \begin{cases}
                    1 & x \le x^* \\
                    \dfrac{1 - 2p(x)}{4p(x^*)(p(x) - p(x^*)) - 2p(x^*) + 1} & x > x^*.
                \end{cases}
            \]
            On the other hand when $x^* > x_0$, $p(x^*) > 1/2$, and we have,
            \[
                b_{\text{cv}}(x; 1) = \begin{cases}
                    \dfrac{1 - 2p(x)}{4(1 - p(x^*))(p(x^*) - p(x)) + 2p(x^*) - 1} & x < x^* \\
                    -1 & x \ge x^*.
                \end{cases}
            \]
            Hence, the fluid limit travels with optimal speed till $x^*$, after which it slows down until it reaches $x_0$, where it settles. Appropriate choice of control variates, however, can make the process more efficient. Suppose the sequence of reference points is chosen such that $X^* = x_0$, $\mathbb{P}$-almost surely. Then, $p(X^*) = 1/2$, $\mathbb{P}$-almost surely, and the asymptotic ZZ-CV drift is equivalent to the canonical drift,
            \[
                b_{\text{cv}}(x; 1) = b_{\text{can}}(x) = \frac{-(2p(x) - 1)}{|2p(x) - 1|},  \quad x \neq x_0.
            \]
        \end{example}

        \begin{example}[Cauchy model]
            Suppose $P = \text{Cauchy}(x_0, 1)$ and $f(y; x) \propto (1 + (y - x)^2)^{-1}$. Then, $S(x, y) = 2(x - y)/(1 + (x - y)^2)$. The explicit expression is hard to calculate in this case but let $m=1$. Then,
            \[
                b_{\text{ss}}(x; 1) = - \frac{2(x - x_0)}{(x - x_0)^2 + 4} \left(\mathbb{E}_{P}\left[\left|\frac{2(x - Y)}{1 + (x - Y)^2} \right|\right]\right)^{-1}, \quad x \in \mathbb{R}.
            \]
            Moreover, suppose $X^* = x_0$, $\mathbb{P}$-almost surely. Then, $\mathbb{E}_{P}S(x_0, Y) = 0$ and for all $x \neq x_0$,
            \begin{align*}
                b_{\text{cv}}(x; 1) &= - \frac{2(x - x_0)}{(x - x_0)^2 + 4} \left(\mathbb{E}_{P}\left[\left|\frac{2(x - Y)}{1 + (x - Y)^2} - \frac{2(x_0 - Y)}{1 + (x_0 - Y)^2}\right|\right]\right)^{-1}. \\
                b_{\text{can}}(x) &= \frac{-(x - x_0)}{|x - x_0|}.
            \end{align*}
        \end{example}    

        \begin{figure}[t]
        \centering
            \begin{subfigure}{0.325\linewidth}
                \includegraphics[scale = 0.25]{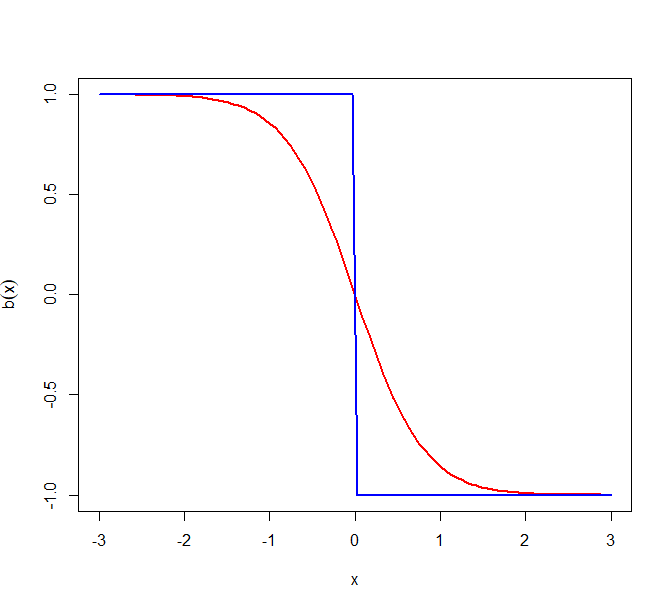}
            \subcaption{Gaussian}
            \end{subfigure}
            \begin{subfigure}{0.325\linewidth}
                \includegraphics[scale = 0.25]{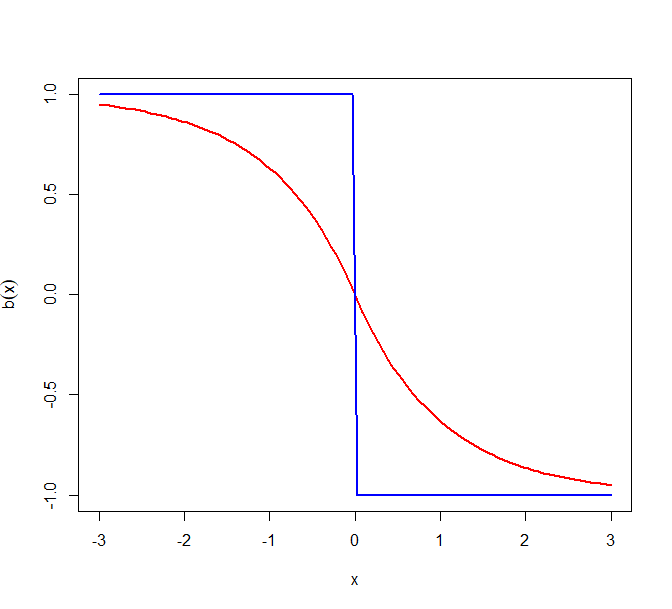}
            \subcaption{Laplace}
            \end{subfigure}
            \begin{subfigure}{0.325\linewidth}
                \includegraphics[scale = 0.25]{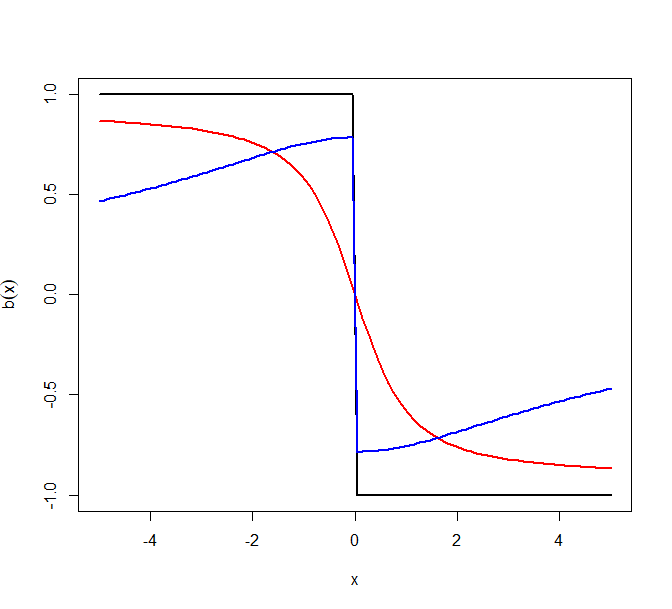}
            \subcaption{Cauchy}
            \end{subfigure}
        \caption{Asymptotic drift for different models and subsampling schemes in one dimension with $P = F_0$ and $X^* = 0$ and $m = 1$. The black curve denotes $b_{\text{can}}$, the red curve is $b_{\text{ss}}$ and the blue curve is $b_{\text{cv}}$. Note that in the first two plots, the black and the blue curves overlap.}
        \label{fig:drifts}
        \end{figure}

        Figure \ref{fig:drifts} plots the asymptotic drifts for the above models for canonical Zig-Zag, ZZ-SS, and ZZ-CV (when $X^* = x_0$, $\mathbb{P}$-almost surely). The subsample size $m = 1$ and $x_0 = 0$. Under this setting, as the calculations above show, when the model is Gaussian or Laplace, ZZ-CV is equivalent to canonical Zig-Zag. However, we see that for the Cauchy model, ZZ-CV is sub-optimal near $x_0$. Indeed, for the Cauchy model, $\lim_{x \to x_0}|b_{\text{cv}}(x;1)| = \pi/4 < 1$. The more interesting feature of these plots is the Cauchy model where, contrary to expectations, ZZ-CV performs worse than ZZ-SS. Later in Section \ref{sec:comparison}, we will see that this is a consequence of the heavy-tails of the likelihood (more
        specifically, asymptotic vanishing of the score function for large parameter values) and holds in some greater
        generality. 

    \subsection{Bayesian logistic regression}
    \label{sec:blr}
        Consider a set of observations $(w^j, y^j)_{j=1}^n$ with $w^j \in \mathbb{R}^d$ and $y^j \in \{0, 1\}$, $j = 1, \dots, n$. We assume a Bayesian logistic regression model on the observations such that given a parameter $x \in \mathbb{R}^d$ and covariate $w^j$, the binary variable $y^j$ has distribution,
        \[
            \Prob(y^j = 1| w^j, x) = \frac{1}{1 + \exp(-\sum_{i=1}^d x_iw^j_i)} = 1 - \Prob(y = 0| w^j, x), \quad j = 1, \dots, n.
        \]
        Additionally, since we are interested in large sample asymptotics, we suppose that the covariates $w^j \overset{\text{iid}}{\sim} Q, j = 1, \dots, n$, for some probability measure $Q$ on $\mathbb{R}^d$ with Lebesgue density $q$. Using a flat prior for $x$, we compute the posterior density $\pi^{(n)}$ as,
        \[
            \pi^{(n)}(x) \propto \prod_{j=1}^n q(w^j)\Prob(y = y^j| w^j, x) = \prod_{j=1}^n q(w^j)\frac{\exp(y_j \cdot \sum_{i=1}^d x_iw^j_i)}{1 + \exp(\sum_{i=1}^d x_iw^j_i)}, \quad x \in \mathbb{R}^d.
        \]  
        
        Let $p(x; w)$ denote the probability $\Prob(Y = 1; w, x)$ and suppose the model is well-specified with true parameter $x_0$. Then, simple calculations yield that the canonical drift in the $i$-th coordinate is,
        \[
            b_{\text{can}, i}(x) = \frac{-\mathbb{E}_{W \sim Q}\left[W_ip(x; W) - W_ip(x_0; W)\right]}{|\mathbb{E}_{W \sim Q}\left[W_ip(x; W) - W_ip(x_0; W)\right]},
        \]
        for all $x$ for which the denominator is non-zero. Further suppose that $m = 1$ and that $X^* = x_0$ almost surely. Then $\mathbb{E}_{P}S(x_0, Y) = 0$ and, 
        \begin{equation*}
        \begin{gathered}
            b_{\text{ss}, i}(x; 1) = \frac{-\mathbb{E}_{W \sim Q}\left[W_ip(x; W) - W_ip(x_0; W)\right]}{\mathbb{E}_{W \sim Q}\left[|W_i|\cdot(p(x_0, W) + p(x; W) - 2p(x_0, W)p(x; W)) \right]}, \quad x \in \mathbb{R}^d;\\
            b_{\text{cv}; i}(x, 1) = \frac{-\mathbb{E}_{W \sim Q}\left[W_ip(x; W) - W_ip(x_0; W)\right]}{\mathbb{E}_{W \sim Q}\left[\left|W_ip(x; W) - W_ip(x_0; W)\right|\right]}, \quad x \neq x_0.
        \end{gathered}
        \end{equation*}

        As an illustrative example, we generated a synthetic data set of size $10000$ from the above model with $Q = \delta_1 \times N(0, 1)$ and $x_0 = (1, 2)$. Figure \ref{fig:blr} plots trajectories of the Zig-Zag process under different subsampling schemes targeting the corresponding Bayesian posterior. The reference point for ZZ-CV was chosen to be a numerical estimate of the MLE. The black curve denotes the actual Zig-Zag process trajectories and the superimposing red curve denotes the solution to the asymptotic ODE.

        \begin{figure}
        \centering
            \begin{subfigure}[t]{0.325\linewidth}
            \centering
                \begin{subfigure}[b]{\linewidth}
                \centering
                    \includegraphics[scale = 0.25]{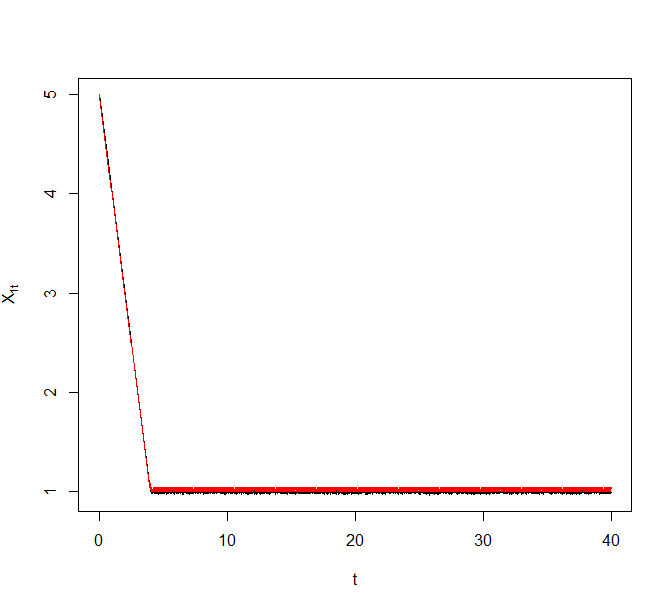}
                \end{subfigure}\\
                \begin{subfigure}[b]{\linewidth}
                \centering
                    \includegraphics[scale = 0.25]{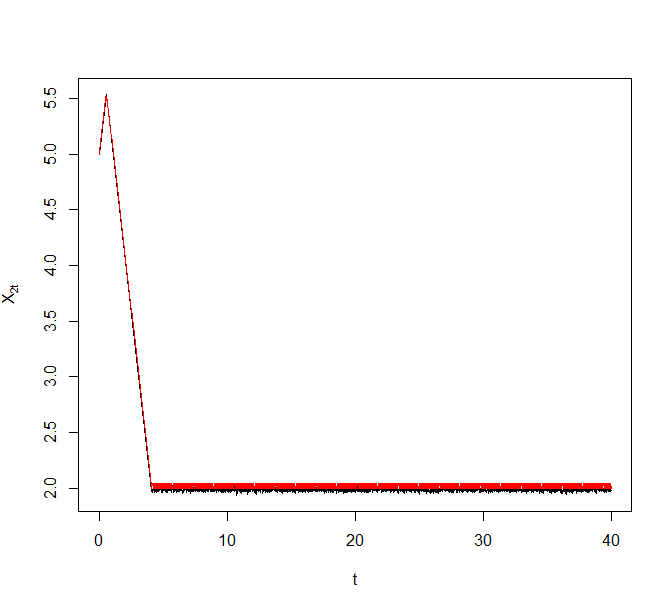}
                \end{subfigure}
            \subcaption{Canonical Zig-Zag}
            \end{subfigure}
            \begin{subfigure}[t]{0.325\linewidth}
            \centering
                \begin{subfigure}[b]{\linewidth}
                \centering
                    \includegraphics[scale = 0.25]{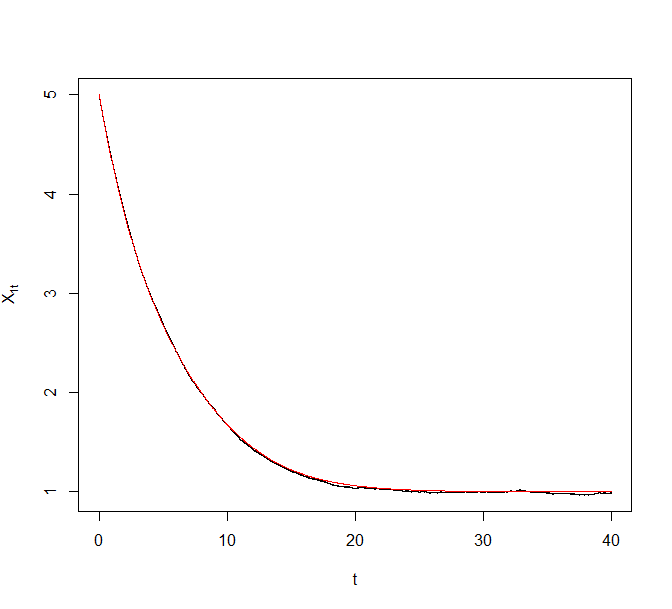}
                \end{subfigure}\\
                \begin{subfigure}[b]{\linewidth}
                \centering
                    \includegraphics[scale = 0.25]{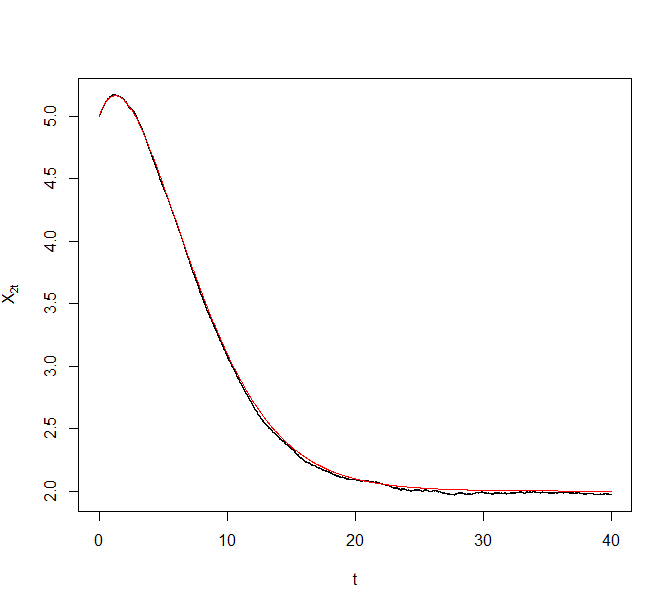}
                \end{subfigure}
            \subcaption{ZZ-SS}
            \end{subfigure}
            \begin{subfigure}[t]{0.325\linewidth}
            \centering
                \begin{subfigure}[b]{\linewidth}
                \centering
                    \includegraphics[scale = 0.25]{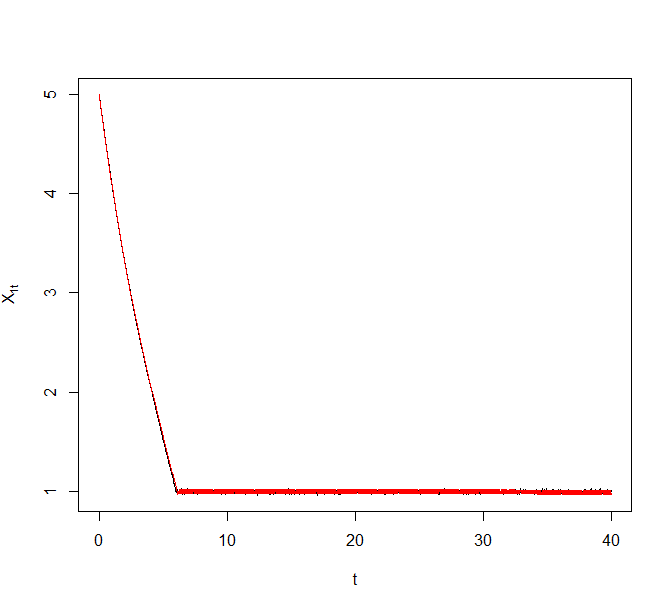}
                \end{subfigure}\\
                \begin{subfigure}[b]{\linewidth}
                \centering
                    \includegraphics[scale = 0.25]{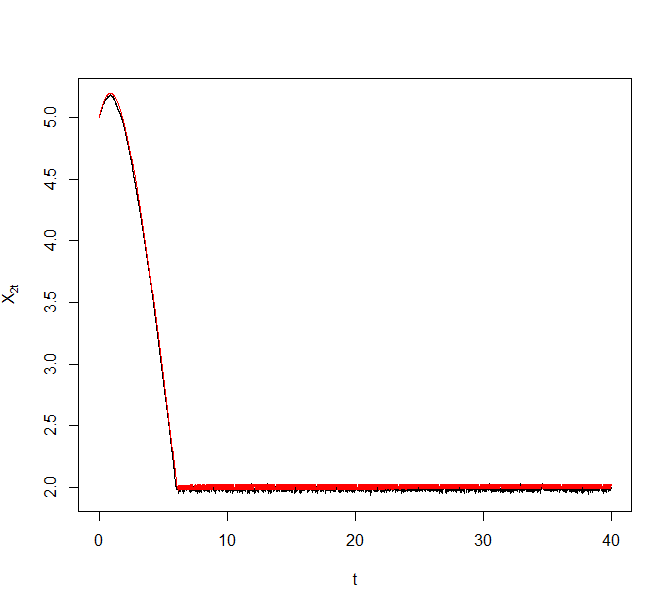}
                \end{subfigure}
            \subcaption{ZZ-CV}
            \end{subfigure}   
            \caption{Bayesian Logistic regression example with 10000 data points. The black curves are the Zig-Zag process trajectories and the red curves denote the theoretical asymptotic limit. The top and the bottom panels correspond to the first and second coordinates respectively.}
            \label{fig:blr}
        \end{figure}

    \subsection{Comparison of ZZ-SS and ZZ-CV}
    \label{sec:comparison}
        As noted in Section~\ref{sec:compare_3}, comparing different subsampling schemes is equivalent to comparing the corresponding denominators in the asymptotic drift \eqref{eq:limiting_drift}. Looking at the expressions for the denominator in \eqref{eq:ss_drift} and \eqref{eq:cv_drift}, it is not obvious when one is faster than the other. Figure \ref{fig:drifts} already illustrates an example where, contrary to expectations, ZZ-CV performs worse than ZZ-SS.         Throughout this section, we suppose that $P = F_{x_0}$ for some $x_0 \in \mathbb{R}^d$ and $\mathbb{E}_{P}S(x_0, Y) = 0$.  

        For the ease of exposition, set $m =1$. The asymptotic drift for ZZ-CV,
        \begin{align*}
            b_{\text{cv}, i}(x; 1) &= \frac{-\mathbb{E}_{P}[S_i(x, Y)]}{\mathbb{E}_{P}\left[\left|S_i(x, Y) - S_i(X^*, Y) + \mathbb{E}_{P}S_i(X^*, Y)\right|\right]} \\
            &= \frac{\mathbb{E}_{P}[|S_i(x, Y)|]}{\mathbb{E}_{P}\left[\left|S_i(x, Y) - S_i(X^*, Y) + \mathbb{E}_{P}S_i(X^*, Y)\right|\right]}\ b_{\text{ss}, i}(x; 1)
        \end{align*}

        Suppose the model is heavy-tailed in the sense that for all $y$, $\lim_{\|x\| \to \infty} \|S(x, y)\| = 0$. The dominated convergence theorem implies that for any fixed value of $X^*$, the multiplicative factor in the above expression goes to $0$ as $\|x\|$ goes to infinity for each fixed value of $X^*$. Since $b_{ss}(x; 1)$ is bounded, this implies that all fixed $X^*$, $\lim_{\|x\| \to \infty} b_{\text{cv}, i}(x; 1) = 0$ for all $i = 1, \dots, d$. In particular, there exists an $M > 0$ such that $|b_{\text{cv}, i}(x; 1)| < |b_{\text{ss}, i}(x; 1)|$ for all $\|x\| > M$, and $i =  1, \dots, d$.


        Hence, the ZZ-SS will outperform ZZ-CV in tails in terms of the speed of convergence when the model is heavy-tailed. This happens because control variates in general require correlation between terms to offer variance reduction. For heavy-tailed models, as the $S(\cdot, y)$ function becomes flatter in the tails, the terms $S(x, Y)$ and $S(X^*, Y)$ become increasingly uncorrelated. Consequently, the ZZ-CV drift becomes extremely slow. The simplest thing to do in such situation is to run ZZ-SS in the tails and ZZ-CV when the process is closer to the centre. More explicitly for some $M > 0$ consider the {\it mixed subsampling} scheme where $E^j_i$ is defined by,
        \begin{equation}
        \label{eq:mixed}
            E_{\text{mix}}^j(x) = E^j_{\text{ss}}(x) - a(x)\left(S(X^*_n, Y_j) - n^{-1}\nabla U(X^*_n)\right), \quad a(x) = \begin{cases}
                0, & \|x - x_0\| > M, \\
                1, & \|x - x_0\| \le M.
            \end{cases}
        \end{equation}
        The resulting algorithm is still exact since $E^j_{\text{mix}}$ satisfies the decomposition \eqref{eq:additive_decomp} for all $x$. In fact, $E^j_{\text{mix}}$ is a control variate Zig-Zag with different control variate then ZZ-CV. 
        More importantly, the mixed subsampling scheme exhibits much faster convergence to ZZ-SS and ZZ-CV. As an illustration, consider the Cauchy example from the earlier section. The point of intersection of $b_{\text{ss}}$ and $b_{\text{cv}}$ in the last panel of Figure \ref{fig:drifts} was obtained numerically to be $\approx 1.605$. For comparison purposes, we apply the mixed scheme in \eqref{eq:mixed} with $M = 1.605$. Figure \ref{fig:cauchy_comparison} plots trajectories of the Zig-Zag process under different subsampling schemes for the Cauchy model. The solid black line corresponds to the mixed algorithm described above. The proposed method provides an improvement over both ZZ-SS and ZZ-CV. 
        \begin{figure}[t]
            \centering
            \includegraphics[scale = 0.5]{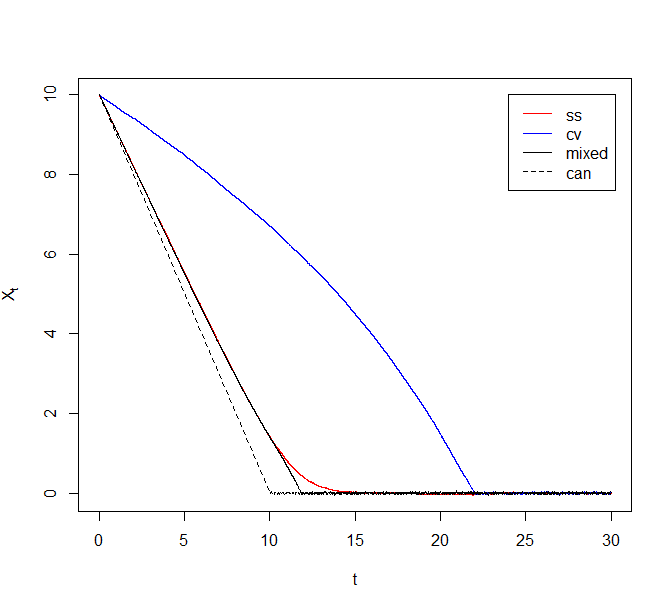}
            \caption{Trajectories of Zig-Zag process for different subsampling schemes targeting posterior from a Cauchy model.}
            \label{fig:cauchy_comparison}
        \end{figure}
        More sophisticated choices of control variates might achieve greater variance reduction (see, e.g., \cite{baker_control_2019} and references therein). 

        On the other hand, close to $x_0$, we argued in Section~\ref{sec:zzss_stat} that the limiting drift for ZZ-SS is $0$. Indeed by dominated convergence, we have that $\lim_{x \to x_0} b_{\text{ss}, i}(x) = 0$ for all $i = 1, \dots, d$. Further, suppose the reference points are chosen such that $X^* = x_0$ almost surely. When $d = 1$, since $\mathbb{E}_{P}S(x_0, Y) = 0$,
        \begin{align*}
            |b_{\text{cv}}|(x; 1) = \left|\frac{\mathbb{E}_{P}[S(x, Y)]}{\mathbb{E}_{P}\left[\left|S(x, Y) - S(x_0, Y)\right|\right]}\right|  = \left|\frac{\mathbb{E}_{P}[\partial_xS(x(Y), Y)}{\mathbb{E}_{P}\left[\left|\partial_xS(x(Y), Y)\right|\right]}\right|,
        \end{align*}
        where $x(Y)$ lies between $x$ and $x_0$. When the model is log-concave, $\partial_xS(x(Y), Y) > 0$ almost surely for all $x$. We get $|b_{\text{cv}}(x;1)| = 1$ for all $x \neq x_0$ and the ZZ-CV is equivalent to the canonical Zig-Zag. The same is true for any value of $m$. In general, we get near stationarity,
        \[
            \lim_{x \to x_0} |b_{\text{cv}}|(x;1) = \left|\frac{\mathbb{E}_{P}[\partial_xS(x_0, Y)}{\mathbb{E}_{P}\left[\left|\partial_x S(x_0, Y)\right|\right]}\right| > 0.
        \]
        Hence, close to $x_0$, the ZZ-CV process performs better than ZZ-SS. But when the model is not log-concave, the limit in the above expression is less than $1$ and the ZZ-CV process is sub-optimal. Observe also that the right-hand side is equal to the absolute value of the drift in \eqref{eq:limiting_drift} defined in terms of the asymptotic switching rate \eqref{eq:limrate} for ZZ-CV when $X^* = x_0$.

    \subsection{Cost to convergence}
        The fluid limit captures the process in its transient phase and can be used to inform on the time taken to converge to equilibrium. It is, in general, difficult to say exactly when the process has reached equilibrium. Nevertheless, following our analysis so far, we might say that the process starts behaving `as in stationarity' once it enters an $O(n^{-1/2})$ neighbourhood of $x_0$; similar characterizations have been used in \cite{christensen_scaling_2005} to address analogous questions for random walk Metropolis and the MALA algorithm in high-dimensional models. 

        For an illustration in $1$-dimension, suppose the starting points $X^n_0$ satisfy the assumptions of Theorem~\ref{thm:transient} and, in particular, $X^n_0 \not\to x_0$ as $n \to \infty$. If the magnitude of the limiting drift around $x_0$ is bounded below, as is the case with canonical Zig-Zag, the process takes $O(1)$ time interval to reach $n^{-1/2}$-neighbourhood of $x_0$. On the other hand, a locally linear attracting drift around $x_0$, i.e., $b(x) \sim -c(x - x_0)$ for some $c > 0$, as is the case with ZZ-SS, yields a travel time of order $O(\log n)$. The drift, and hence the time to stationarity, for ZZ-CV depends on the model and limiting reference point. As seen previously in this section, it can either be bounded away from $0$ in magnitude in which case ZZ-CV also reaches stationarity in $O(1)$ time or it can be locally linear and attracting in which case ZZ-CV is as slow as ZZ-SS. Under a thinning bound of order $n$ during the entire transient-phase flight, the above predictions estimate the worst total cost at $O(n^2)$ for canonical Zig-Zag and $O(n\log n)$ for ZZ-SS and ZZ-CV. As seen in Figure~\ref{fig:cauchy_comparison}, the mixed scheme in \eqref{eq:mixed} combines the best of both scenarios to get a flight time of $O(1)$ and the overall cost of $O(n)$ in the transient phase.
        
        Comparing with the stationary phase estimates in Section~\ref{sec:cost}, in all cases, both the time to burn-in (or to reach stationarity) and the cost in the transient phase dominates the time scale on which the process decorrelates and the overall cost in the stationary phase. In practice, if a good reference point is available in the sense of Remark~\ref{rem:bad_reference}, the warm-start initialization in Section~\ref{sec:cost} can avoid the transient-phase flight with additional initialization cost.

\section{Discussion}
\label{sec:conclusion}
    In this paper, we rigorously study the effect of subsampling on mixing properties of Zig-Zag algorithm in settings of Bayesian inference for i.i.d. observations. Drawing motivation as well as theoretical guarantees from Bayesian large-sample theory, we prove weak convergence results for different variants of Zig-Zag process as dataset size $n \to \infty$. Our objective in undertaking this study has been to comment on the overall complexity of the algorithm as a function of $n$. We estimate that ZZ-CV has a total algorithmic complexity of $O(1)$ in stationarity while both canonical Zig-Zag and ZZ-SS have $O(n)$. This reveals that the Zig-Zag algorithm is no worse than other reversible MCMC algorithms such as the random walk Metropolis which also has a complexity of $O(n)$ in stationarity \citep{schmon_large-sample_2021}. Moreover, the control variates can be utilized to further gain an order of magnitude in total computational cost. Our results also provide insights into convergence time of Zig-Zag samplers initiated out of stationarity. 
    The canonical Zig-Zag converges to stationarity with optimal drift while the subsampling versions remain sub-optimal in general. Although, ZZ-CV can attain optimal drift for log-concave model densities in one dimension for suitable reference points. More generally, we observe that both, the time to reach stationarity and the cost in the transient phase, dominates the time scale on which the process decorrelates and the overall cost in the stationary phase. 

    Based on our observations, we infer strong support for the use of Zig-Zag for sampling from Bayesian posteriors for big data. However, we also caution that the superefficient performance of ZZ-CV heavily relies on good reference points. Inappropriately chosen reference points can lead to undesirable performance as discussed in Sections~\ref{sec:stationary} and \ref{sec:examples}. To mitigate the effect of bad reference points in transient phase, we propose a mixed scheme whereby vanilla subsampling is done in the initial phases of the algorithm and control variates are introduced later when the sampler is close to stationarity. We empirically show that such a scheme can achieve faster convergence than either ZZ-SS or ZZ-CV alone. In practice, although, if a good reference point is available in the sense of Remark~\ref{rem:bad_reference}, it is natural to initialize ZZ-CV close to it, making the transient fluid-limits for ZZ-CV less relevant in practice. Nevertheless, the fluid limit provides insight into the behaviour of ZZ-CV under initializations away from the posterior mode. It is useful for diagnosing poor reference points and prospective adaptive algorithms. The analysis of more sophisticated control variates is an important direction of future research.

    The asymptotic behaviour of both ZZ-SS and ZZ-CV depends on the choice of the subsample size $m$. Theoretically, larger subsample size implies faster convergence and better mixing. In practice, larger $m$ also means additional costs. We argue (in Remarks~\ref{rem:optimal_m}, \ref{rem:zzss_complexity}, and~\ref{rem:zzcv_complexity}) that it is optimal to choose subsample size $m = 1$ in terms of total algorithmic complexity, under a natural model of computation. We assume that the subsamples are drawn randomly without replacement. As a result, we can invoke the law of large numbers for $U$-statistics (Proposition~\ref{propo:slln_ustatistics}) to obtain our limit results. There is no added generality in the results in allowing sampling with replacement. The effective switching rates 
    can still be expressed explicitly as a combination of $U$-statistics and the strong law of large numbers will hold \citep{lee_u-statistics_2019}. In practice, however, it is better to do sampling without replacement as it leads to a smaller variance during implementation.  
    
    In this paper, we have limited ourselves to the analysis of underlying PDMP to study the overall complexity of Zig-Zag algorithms. In practice, however, the total cost of the algorithm also depends on the computational bounds used for Poisson thinning. These computational bounds determine the number of proposed switches in the algorithm and directly affect the total complexity of the sampler. In general, subsampling will lead to increased computational bounds as exemplified in Section \ref{sec:cost}. We see the identification of tight and feasible computational bounds as an important direction of further research.  

    The dynamics of Zig-Zag process is completely characterized by its switching rates. The weak convergence results in this paper rely on properties of these rates such as smoothness, linear in $n$ growth, and asymptotic symmetry (for ZZ-SS in stationarity). Under the model considered in this paper, the i.i.d. structure provides the linear in $n$ growth for switching rates. The Bernstein--von Mises theorem provides appropriate posterior rescaling and Gaussian approximation, which in turn yields asymptotic symmetry for ZZ-SS in stationary phase.
    
    However, neither the subsampling method nor the Bernstein--von Mises theorem is restricted to the i.i.d. setting. The former applies whenever the potential $U_n$ admits the decomposition in \eqref{eq:additive_decomp}. The latter has been established beyond the i.i.d. setting, for example, for hierarchical models \citep{koers_misspecified_2026}, and for generalized posteriors with temporally dependent data \citep{sinha_roy_prequential_2026}. More generally, suitable local quadratic approximations and posterior concentration conditions can yield Bernstein--von Mises results \citep{kleijn_bernstein-von-mises_2012,miller_asymptotic_2021}. Although, the $U$-statistics arguments used in this paper may not directly extend to non i.i.d. data. Nevertheless, if the dataset and model possess sufficient structure to establish the required uniform convergence of the switching rates, as in Proposition~\ref{propo:rates}, together with the relevant regularity properties, analogous weak convergence results may be obtained.

    Similar results may also be established for more general PDMPs with growing event rates, although their analysis may require different proof techniques. Such results would facilitate comparisons between PDMP samplers, including the Bouncy Particle Sampler \citep{bouchard-cote_bouncy_2018} and the Coordinate Sampler \citep{wu_coordinate_2020}, and provide further guidance to practitioners. We envision the extension of our results to other PDMP methods as important future work.

    It could be argued that the large $n$ limit, as studied in this paper, is of minor interest since in this case the posterior distribution is already well approximated by a Gaussian distribution; see also, e.g., \cite{johndrow_no_2024} which carefully considers the validity of subsampling MCMC in the large $n$ regime. However in practical situations of interest it may be unclear whether $n$ is sufficiently large for the Gaussian approximation to be satisfactory. The use of MCMC provides robust insurance against overconfidently assuming approximate Gaussianity, and the use of subsampling strategy as discussed and analyzed in this paper then provides means for carrying out MCMC in a computationally efficient way, competitive with the use of a Gaussian approximation. In view of this it could also be of interest to provide non-asymptotic guarantees on the moderate $n$ regime or regimes in which both $d$ and $n$ tend to infinity, which remain prominent directions for further research.
    

    We regard this work as contributing to the understanding of how the properties of PDMP algorithms vary with data. While other work has made significant progress on studying dimension dependence of these methods \cite{holderrieth_cores_2021, deligiannidis_randomized_2021, andrieu_hypocoercivity_2021,bierkens_high-dimensional_2022,bierkens_large_2021} with no reference to sample size dependence, important challenges remain in understanding how algorithms perform as a function of both dimension and sample size.


\begin{appendix}
\section{Asymptotic switching rates and the drift}
\label{sec:asymp_appendix}
    In this section, we will prove a law of large numbers for switching rates defined in \eqref{eq:subsampling_rates}. 
    
    \subsection{Technical definitions and results}
        First, we recall the definition of a $U$-statistic and the corresponding strong law of large numbers (SLLN). For details, we refer to \cite{lee_u-statistics_2019}.

    \begin{defi}[$U$-statistics \cite{lee_u-statistics_2019}]
    \label{def:ustatistic}
        Let $\mathcal{X}$ be an observation space and $\Theta$ be a parameter space. For some $m \in \mathbb{N}$, let $k: \mathcal{X}^m \to \Theta$ be a functional such that $k$ is symmetric in its arguments. Let $X_1, \dots, X_n$ be a collection of $n \ge m$ i.i.d. distributed random variables. The $U$-statistic of degree $m$ with kernel $k$ is the estimator,
        \[
            U_n(X_1, \dots, X_n) := \binom{n}{r}^{-1} \sum_{(n, m)}k(X_{j_1}, \dots, X_{j_m}), \numberthis \label{eq:ustatistic}
        \]
        where the sum $\sum_{(n, m)}$ is taken over all subsets $1 \le j_1 < \dots j_m \le n$ of $\{1, \dots, n\}$.
    \end{defi}

    Now, the strong law of large numbers for $U$-statistic.

    \begin{propo}[SLLN for $U$-statistics \cite{lee_u-statistics_2019}]
    \label{propo:slln_ustatistics}
        For $n \ge m$, let $X_1, \dots, X_n \sim F$ be i.i.d random variables from some distribution $F$. Then,
        \[
            \mathbb{E}U_n = \int_{\mathcal{X}}\dots\int_{\mathcal{X}}k(x_1, \dots, x_m)F(dx_1)\dots F(dx_m) =: k(F).
        \]
        Suppose $\mathbb{E}|k(X_1, \dots, X_m)| < \infty$. Then $U_n$ converges to $k(F)$ almost surely.
    \end{propo}

    In addition to the result on $U$-statistics, we will also use the following result on variance of the population mean estimator using simple random sampling without replacement (SRSWOR). For details, we refer to \cite{cochran_sampling_1977}. For $m \le n$, recall that $\mathcal{J}_{(n, m)}$ denotes the collection of all possible subsets of size $m$ with distinct elements from $\{1, \dots, n\}$.

    \begin{propo}[Variance of SRSWOR estimator]
    \label{propo:srswor}
        Given a population $\{y_1, \dots, y_n\}$ of size $n \ge 2$, define the population mean, $\bar{y}$, and the SRSWOR estimator, $\bar{y}_{J}$, of size $m \le n$ of the population mean by,
        \[
            \bar{y} := \frac{1}{n} \sum_{j=1}^n y_j, \quad \bar{y}_J := \frac{1}{m}\sum_{j \in J} y_j, \quad J \sim \text{Uniform}(\mathcal{J}_{(n, m)}). \numberthis \label{eq:srswor_est}
        \]
        Then,
        \[
            \mathbb{E} \bar{y}_J = \bar{y}, \quad \text{Var}(\bar{y}_J) = \left(\frac{1}{m} - \frac{1}{n}\right)\frac{1}{n-1}\sum_{j=1}^n (y_j - \bar{y})^2. \numberthis \label{eq:srswor_var}
        \]
    \end{propo}
    \begin{proof}
        See \cite[][Theorem 2.2]{cochran_sampling_1977}.
    \end{proof} 
    
    We can now prove Proposition \ref{propo:rates}.

\subsection{Asymptotic switching rates}
\label{sec:proof_propo}

    The idea for the proof of Proposition \ref{propo:rates} is that the switching rates defined by \eqref{eq:subsampling_rates} can be expressed as $U$-statistics. Indeed, from \eqref{eq:subsampling_rates}, observe that for fixed $(x, v) \in \mathsf{E}$ and $i \in \{1, \dots, d\}$, we have,
    \[
        \frac{\lambda^n_i(x,v; m)}{n} = \binom{n}{m}^{-1}\sum_{(n, m)}k(E^{j_1}_{i}, \dots, E^{j_m}_i), \numberthis \label{eq:ratesbyn}
    \]
    for some choice of functions $E^j$ and $k(y_1, \dots, y_m) = (v_i\cdot m^{-1}\sum y_j)_{+}$. If $E^j$s are chosen to be i.i.d. and do not depend on $n$, then for each $(x, v)$, $\lambda^n_i(x, v)/n$ satisfies Definition \ref{def:ustatistic} and, thus, is a $U$-statistic of degree $m$ with kernel $k$. If, in addition, $m(n) = m$ is fixed for some $m$, Proposition \ref{propo:slln_ustatistics} can be directly applied. The functions $(E^j)$ corresponding to ZZ-SS and ZZ-CV, in \eqref{eq:zzss_ej1} and \eqref{eq:zzcv_ej1} respectively, will not always satisfy this condition. Nevertheless, the strong law of large numbers still holds as we now prove.


    \begin{proof}[Proof of Proposition \ref{propo:rates}]
        First, consider the ZZ-SS case when 
            \[
                E^j(x) = E^j_{\text{ss}}(x) = S(x, Y_j) = -\nabla_x\log f(Y_j; x).
            \]
            Then, since $Y_1, \dots, Y_n$ are i.i.d., for each $i$, so are the terms $\{E^j_{\text{ss}, i}\}_{j = 1}^{n}$. 
            
            Suppose $m(n) \to m < \infty$. Since $m(n)$ only takes integer values, there exists $N$ such that $m(n) = m$ for all $n \ge N$. For all $n \ge N$, and all $i = 1, \dots, d$, $n^{-1}\lambda^n_{\text{ss}, i}(x, v; m)$ is a $U$-statistic of the form \eqref{eq:ratesbyn} indexed by $(x, v)$. By Proposition \ref{propo:slln_ustatistics}, for all $(x, v) \in \mathsf{E}$ and $i = 1, \dots, d$, 
            \begin{equation}
            \label{eq:limitrate1}
                \frac{\lambda^n_{\text{ss},i}(x, v; m(n))}{n} \xrightarrow{n \to \infty} \mathbb{E}\left[\left(v_i \cdot m^{-1}\sum_{j=1}^m S_i(x, Y_j) \right)_{+}\right], \quad \Prob-\text{almost surely},
            \end{equation}
            where $Y_1, \dots, Y_m \overset{iid}{\sim} P$. Furthermore, given the smoothness assumption (\ref{assum:smoothness}), the convergence is uniform on all compact sets. 
        
            Suppose $m(n) \to \infty$ as $n \to \infty$. Let $(y_j)_{i=1}^n$ be the observed dataset. Given $(y_j)_{j=1}^n$, define $\bar{E}(x) := \sum_{j = 1}^n E^j(x)/n$. Note that $\zeta(J, x)/n = \sum_{j \in J} E^j(x)/m(n)$ is the observed value of the SRSWOR estimator (as defined in \eqref{eq:srswor_est}) of the population mean $\bar{E}(x)$ when the sample $J \in \mathcal{J}_{(n, m(n))}$ is observed. Consider,
            \begin{align*}
                &\left|\frac{\lambda^n_{\text{ss},i}(x, v; m(n))}{n} - \frac{(v_i\cdot\sum_{j=1}^n E^j_{i}(x))_{+}}{n}\right| \\
                &\quad = \left|\frac{1}{|\mathcal{J}_{(n,m(n))}|}\sum_{J \in \mathcal{J}_{(n,m(n))}}\left(v_i \cdot \frac{\zeta_i(J, x)}{n}\right)_{+} - \left(v_i \cdot \bar{E}_{i}(x)\right)_{+}\right| \\
                &\quad \le \frac{1}{|\mathcal{J}_{(n,m(n))}|}\sum_{J \in \mathcal{J}_{(n,m(n))}} \left| \left(v_i \cdot \frac{\zeta_i(J, x)}{n}\right)_{+} - \left(v_i \cdot \bar{E}_i(x)\right)_{+} \right|\\
                &\quad \le \frac{1}{|\mathcal{J}_{(n,m(n))}|}\sum_{J \in \mathcal{J}_{(n,m(n))}} \left| \frac{\zeta_i(J, x)}{n} - \bar{E}_i(x) \right|.
            \end{align*}
            The last term on the right in the above equation is the mean absolute deviation about the mean of the SRSWOR estimator $\zeta(J, x)/n$. This can be bounded above by the standard error of this estimator. From Proposition \ref{propo:srswor} we get, $\Prob-$almost surely,
            \begin{align*}
                &\left|\frac{\lambda^n_{\text{ss},i}(x, v; m(n))}{n} - \frac{(v_i\cdot\sum_{j=1}^n E^j_{i}(x))_{+}}{n}\right| \\
                &\quad \quad \le \sqrt{\left(\frac{n - m(n)}{n\cdot m(n)}\right) \frac{n}{n-1}}\sqrt{\frac{1}{n}\sum_{j=1}^n (E^j_{i}(x) - \bar{E}_i(x))^2}. \numberthis \label{eq:srswor_bound}
            \end{align*}
            Since, $E^j(x) = E^j_{\text{ss}}(x) = S(x, Y_j)$,
            \[
                \frac{1}{n}\sum_{j=1}^n (E^j_{i}(x) - \bar{E}_i(x))^2 \xrightarrow{n \to \infty} \text{Var}_{Y \sim P}(S_i(x, Y)) < \infty
            \]
            uniformly on compact sets almost surely by the law of large numbers. And thus, if $m(n) \to \infty$ as $n \to \infty$, the right-hand side in \eqref{eq:srswor_bound} goes to $0$ uniformly in $x$ almost surely. Consequently, we have,
            \begin{equation}
            \label{eq:limitrate2}
                \frac{\lambda^n_{\text{ss},i}(x, v; m(n))}{n} \xrightarrow{n \to \infty} \left(v_i\cdot \mathbb{E}_{P}S_i(x, Y)\right)_{+}, \quad \mathbb{P}\text{-almost surely}.
            \end{equation}
        
            Now consider the ZZ-CV case, i.e.,
            \[
                E^j(x) = E^j_{\text{cv}} = S(x, Y_j) - S(X^*_n, Y_j) + n^{-1}\nabla U_n(X_n^*),
            \]
            where $X^*_n = T_n(\bs{Y}^{(n)})$. The sequence of $E^j$ is no longer independent, and so the switching rate does not resemble a $U$-statistic. Let $X^*$ be the limiting reference point from Assumption \ref{assum:reference}. 
            Define $G^j$ as,
            \[
                G^j(x) = S(x, Y_j) - (S(X^*, Y_j) - \mathbb{E}_{P}S(X^*, Y)).
            \]
            Due to Assumption \ref{assum:reference}, $G^j$ are conditionally independent given $X^*$. Consider the difference,
            \begin{align*}
            \label{eq:ei-gi}
                &\left|\frac{1}{|\mathcal{J}_{(n, m(n))}|}\sum_{J \in \mathcal{J}_{(n, m(n))}} \left(v_i\cdot\frac{\sum_{j \in J}E^j_i(x)}{m(n)}\right)_{+} -  \left(v_i\cdot\frac{\sum_{j \in J}G^j_i(x)}{m(n)}\right)_{+} \right| \\
                &\quad \quad \le \frac{1}{m(n)|\mathcal{S}_{(n, m(n))}|}\sum_{S \in \mathcal{S}_{(n, m(n))}} \sum_{j \in S}\left|E^j_i(x) - G^j_i(x)\right| \\
                &\quad \quad = n^{-1}\sum_{j=1}^n |E^j_i(x) - G^j_i(x)|. \numberthis
            \end{align*}
            Let $(y_j)_{j=1}^n$ and $x_n^*$ be the observed dataset and reference point, respectively. Also, let $x^*$ be the observed value of $X^*$ such that $x_n^* \to x^*$. Let $s^j$ denote the function $x \mapsto S(x, y_j)$ for each $j$. However,
            \begin{align*}
                |E^j_i(x) - G^j_i(x)| &= \left|-s^j_i(x_n^*) + s^j_i(x^*) + n^{-1}\sum_{k=1}^ns^k_i(x_n^*) - \mathbb{E}_PS_i(x^*, Y)\right| \\
                &\le \left|s^j_i(x_n^*) - s^j_i(x^*)\right| + \left|n^{-1}\sum_{k=1}^ns^k_i(x_n^*) - \mathbb{E}_PS_i(x^*, Y) \right|
            \end{align*}
            for all $x \in \mathbb{R}$. A second-order Taylor expansion gives,
            \[
              s_i^j(x^*_n) - s^j_i(z) = \nabla s^{j}_i(z)\cdot(x^*_n - z) + (x^*_n - z)^T\nabla^{\otimes 2} s^{j}_i(z^j)(x^*_n - z)
            \]
            where $z^j = z + \theta^j(x^*_n - z)$ for some $\theta^j \in (0, 1)$. Under the assumption that third derivatives are bounded (Assumption \ref{assum:smoothness}),
            \[
              |s_i^j(x^*_n) - s^j_i(x^*)| \le \|\nabla s^{j}_i(x^*)\|\cdot\|x^*_n - x^*\| + M'\|x^*_n - x^*\|^2
            \]
            Then,
            \begin{align*}
               n^{-1}\sum_{j=1}^n |E^j_i(x) - G^j_i(x)| &\le n^{-1}\|x^*_n - x^*\|\sum_{j=1}^n\|\nabla s^{j}_i(x^*)\| + M'\|x^*_n - x^*\|^2 \\
               &\quad \quad \quad \quad + \left|n^{-1}\sum_{k=1}^ns^k_i(x_n^*) - \mathbb{E}_PS_i(x^*, Y) \right|
            \end{align*}
            Given that $\mathbb{E}[|\nabla S(x, Y)|] < \infty$ for all $x$, and $x^*_n \to x^*$, each term on the right hand side goes to $0$. Since the above does not depend on $x$, this convergence is uniform on compact sets. This implies that \eqref{eq:ei-gi} goes to $0$ uniformly almost surely. Now, given $X^* = x^*$, $G^j_i$s are i.i.d. And so, $\lambda^n_{\text{cv},i}(x, v; m(n))/n$ is asymptotically equivalent to,
            \[
                \frac{1}{|\mathcal{S}_{(n, m(n))}|}\sum_{S \in \mathcal{S}_{(n, m(n))}} \left(v_i\cdot(m(n))^{-1}\sum_{j \in S}G^j_i(x)\right)_{+}
            \]
            Once again, when $m(n) \to m < \infty$, there exists an $N$ such that for all $n \ge N$, $m(n) = m$ and the above is a $U$-statistic indexed by $(x,v)$. By Proposition~\ref{propo:slln_ustatistics} again, 
            \begin{align*}
                &\frac{\lambda^n_{\text{cv},i}(x, v; m(n))}{n} \\
                &\quad \xrightarrow{n \to \infty} \mathbb{E}_{P}\left[v_i \cdot \left(m^{-1}\sum_{j=1}^m S_i(x, Y_j) - S_i(X^*, Y_j) + \mathbb{E}_{P}S_i(X^*, Y)\right)_{+}\right], \numberthis \label{eq:limitrate3}
            \end{align*}
            uniformly on compact sets $\mathbb{P}$-almost surely. When $m(n) \to \infty$, a similar argument as for ZZ-SS gives,
            \begin{equation}
            \label{eq:limitrate4}
                \frac{\lambda^n_{\text{cv},i}(x, v; m(n))}{n} \xrightarrow{n \to \infty} (v_i\cdot\mathbb{E}_{P}(S_i(x, Y)))_{+}.
            \end{equation}
            This covers all the cases and, hence, proves Proposition \ref{propo:rates}. The limiting rate $\lambda_i$ is given by the right hand side of equations \eqref{eq:limitrate1}, \eqref{eq:limitrate2}, \eqref{eq:limitrate3}, and \eqref{eq:limitrate4}.
    \end{proof}

\section{Extended generator of Zig-Zag process}
\label{sec:generator}    
        Let $\weak$ denote the weak derivative operator in $\mathbb{R}^d$ such that for all $(x, v) \in \mathsf{E}$,
        \[
            f(x + tv, v) - f(x, v)= \int_0^t v\cdot\weak f(x + vs, v)\ ds, \quad t \ge 0.
        \] 
        When $f$ is differentiable in its first argument, $\weak f = \nabla f$, where $\nabla f$ denotes gradient of $f$.
        
        Recall that $\mathsf{F}_i,\ i = 1, \dots, d$ are flip-operators defined in \eqref{eq:flip_operator}. The extended generator $(\mathcal{L}, \mathcal{D}(\mathcal{L}))$ of Zig-Zag process $(Z_t)_{t \ge 0}$ with switching rates $\{\lambda_i\}_{i=1}^d$ is given by (see \cite{davis_piecewise-deterministic_1984, bierkens_zig-zag_2019}), 
        \begin{equation}
        \label{eq:zigzag_generator1}
            \mathcal{L}f(x,v) = v\cdot \weak f(x, v) + \sum_{i=1}^d\lambda_i(x, v)\{f(x, \mathsf{F}_i(v)) - f(x, v)\}, \quad (x,v) \in \mathsf{E}.
        \end{equation}
        for all $f \in \mathcal{D}(\mathcal{L})$, where
        \[
            \mathcal{D}(\mathcal{L}) = \{f \in C(\mathsf{E}): f(\cdot, v) \text{ is absolutely continuous for all } v \in \mathcal{V}\}.
        \]
        Here $C(\mathsf{E})$ denotes the class of continuous functions on $\mathsf{E}$ where $\mathsf{E}$ is equipped with with the product topology of the Euclidean topology on $\mathbb{R}^d$ and the discrete topology on $\mathcal{V}$.

\section{Fluid limits in the transient phase}
\label{sec:proof_transient}

    Let $\mathsf{a}$ and $\{m(n)\}_{n\in \mathbb{N}}$ be as fixed in Section~\ref{sec:fluid}. For $n \in \mathbb{N}$, let $(Z^n_{\mathsf{a},t})_{t\ge0} = (X^n_{\mathsf{a},t}, V^n_{\mathsf{a},t})_{t\ge0}$ be Zig-Zag process with switching rates $\{\lambda^n_{\mathsf{a}, i}\}_{i=1}^d$. In this section, we will prove Theorem~\ref{thm:transient}. Hereafter, throughout this section, we drop $\mathsf{a}$ from the subscripts and $m(n)$ from arguments for notational convenience.
    

        
        \subsection{Technical results}

         Let $1/2 < \delta < 1$. For arbitrary $f \in C_{c}^{\infty}(H)$, extend $f$ to $\mathbb{R}^d$ by $0$. Let $f_x$ denote the projection map $(x, v) \mapsto f(x)$. Let $\{P^n_t\}_{t \ge 0}$ denote the semigroup of $(Z^n_t)_{t \ge 0} = (X^n_t, V^n_t)_{t \ge 0}$. Define
        \[
            f_n(x, v) = \dfrac{1}{n^{-\delta}} \int_{0}^{n^{-\delta}} P^n_sf_x(x, v)\ ds, \quad (x, v) \in \mathsf{E}. \numberthis \label{eq:fn}
        \]
        Let $\|\cdot\|_{\infty}$ denote the sup-norm.

        \begin{lemma}
        \label{lemma:f_perturbation}
            Let $f \in C_{c}^{\infty}(H)$ be arbitrary and let $f_n$ be as in \eqref{eq:fn}. Then,
            \[
                \lim_{n \to \infty}\left\|f_x - f_n\right\|_{\infty} = 0, \quad \mathbb{P}\text{-almost surely}.
            \]
            Also, $\|f_n\|_{\infty} \le \|f\|_{\infty} < \infty$ for all $n$.
        \end{lemma}
        \begin{proof}
            Since $f \in C_{c}^{\infty}(H)$, there exists an $M$ such that for all $x, y$, $|f(x) - f(y)| \le M\|x - y\|$. Moreover, since the Zig-Zag process has unit speed, $\|X^n_s - x\| \le \sqrt{d}s$ for all $s$. Thus, for all $s$, $|f(x) - f(X^n_s)| \le M\sqrt{d}s$. For any $s \ge 0$, $P^n_sf_x(x, v) = \mathbb{E}_{(x,v)}f_x(X^n_s, V^n_s) = \mathbb{E}_{(x,v)}f(X^n_s)$. From \eqref{eq:fn}, using triangle inequality,
            \begin{align*}
                \left|f_x(x, v) - f_n(x, v)\right| \le \frac{1}{n^{-\delta}}\mathbb{E}_{(x, v)} \int_0^{n^{-\delta}} \left|f(x) - f(X^n_s)\right|\ ds \le M\sqrt{d}n^{-\delta}/2.
            \end{align*}
            The right-hand side goes to $0$ as $n \to \infty$ since $\delta > 0$. The claim follows. The second part follows since $\{P^n_t\}_{t\ge 0}$ is a contraction semigroup.
        \end{proof}

        Let $(\mathcal{L}^n,\mathcal{D}(\mathcal{L}^n))$ denote the extended generator of $(Z^n_t)_{t \ge 0} = (X^n_t, V^n_t))_{t \ge 0}$. For any $f \in C^{\infty}_c(H)$, extended to $\mathbb{R}^d$ by $0$, the function $f_x$ belongs to domain of the strong generator corresponding to $\mathcal{L}^n$ \cite{holderrieth_cores_2021}. Then, so does $f_n$ and,
        \[
            \mathcal{L}^nf_n(x, v) = \dfrac{1}{n^{-\delta}} \int_{0}^{n^{-\delta}} \mathcal{L}^nP^n_sf_x(x, v)\ ds = \dfrac{1}{n^{-\delta}}\int_{0}^{n^{-\delta}} P^n_s\mathcal{L}^nf_x(x, v)\ ds;
        \]
        see, e.g., \cite[][Proposition~14.10]{davis_markov_1993}. From \eqref{eq:zigzag_generator1} and the definition of semigroup,
        \[
            \mathcal{L}^nf_n(x, v) = \dfrac{1}{n^{-\delta}} \mathbb{E}_{(x,v)} \int_{0}^{n^{-\delta}} V^n_s\cdot \nabla f(X^n_s)\ ds, \quad (x, v) \in \mathsf{E}. \numberthis \label{eq:Lnfn}
        \]
        Thus, the function $f_n$ may be thought of as corresponding to discretizing the Zig-Zag process at regular time intervals of length $n^{-\delta}$ and then embedding the discrete chain back into continuous-time using an independent Poisson clock of rate $n^{\delta}$. Indeed, one can check that the Zig-Zag generator acting on $f_n$ coincides with the generator of this Poissonized discrete-time chain acting on $f_x$.

        Define
        \[
            G^{n, f}(x, v) = \dfrac{1}{n^{-\delta}} \mathbb{E}_{(x,v)} \int_{0}^{n^{-\delta}} V^n_s\cdot \nabla f(x)\ ds, \quad (x, v) \in \mathsf{E}. \numberthis \label{eq:Gnf}
        \]
        The next lemma shows that we can replace $\mathcal{L}^nf_n$ by $G^{n,f}$ as $n \to \infty$.

        \begin{lemma}
        \label{lemma:gen_approx_1}
            Let $f \in C^{\infty}_c(H)$ be arbitrary and let $f_n$ be as in \eqref{eq:fn}. Then,
            \[
                \lim_{n \to \infty}\left\|\mathcal{L}^nf_n - G^{n,f}\right\|_{\infty} = 0, \quad \mathbb{P}\text{-almost surely}.
            \]
        \end{lemma}
        \begin{proof}
            As in the proof of Lemma~\ref{lemma:f_perturbation}, since $f \in C^{\infty}_c(H)$, there exists $M_1 > 0$ such that $\|\nabla f(x) - \nabla f(X^n_s)\| \le M_1\sqrt{d}s$ for all $s > 0$. From \eqref{eq:Lnfn} and \eqref{eq:Gnf},
            \[
                |\mathcal{L}^nf_n(x, v) - G^{n,f}(x, v)| \le \frac{1}{n^{-\delta}}\mathbb{E}_{(x,v)}\int_0^{n^{-\delta}} \left|V^n_s\cdot (\nabla f(X^n_s) - \nabla f(x))\right|\ ds \le M_1dn^{-\delta}/2.
            \]
            The right-hand side goes to $0$ as $n \to \infty$ since $\delta > 0$. The claim follows.
        \end{proof}

        Lemma \ref{lemma:gen_approx_1} implies that we can replace $\mathcal{L}^nf_n$ by $G^{n, f}$ for all test functions $f$. 
        Recall that $(V^n_t)$ jumps at random times according to the time-varying rates $\{\lambda^n_{i}(X^n_t, \cdot)\}_{i=1}^d$. If $(\lambda^n_i)_{i=1}^d$ do not change much during a short time interval $(0, n^{-\delta})$, we can replace $V^n_t$ in \eqref{eq:Gnf} by another process $V^{n, x}_t$ (say) which evolves according to frozen time-homogeneous rates $\{\lambda^n_i(x, \cdot)\}_{i=1}^d$. By Proposition \ref{propo:rates}, $\lambda^n_i$ is of the order $n$ for any fixed $x$. Thus, the replacement process $(V^{n, x}_t)$ oscillates, on average, $n^{1-\delta}$ times in the time interval $(0, n^{-\delta})$. Since $n^{1-\delta} \to \infty$ as $n \to \infty$, the corresponding integral in \eqref{eq:Gnf} is then approximated by a law of large numbers for the replacement process $(V^{n, x}_t)$. 
        
         We now present a couple of results which will be used in the further approximation of $\mathcal{L}^nf_n$. Similar to the sets $H_{i}$ in \eqref{eq:setsH}, for each $i = 1, \dots, d$, define the sets,
        \begin{equation}
        \label{eq:hni}
            H^n_{i} := \left\{x \in \mathbb{R}^d: \lambda^n_{i}(x, \bs{-1}_d) + \lambda^n_i(x, \bs{1}_d) = 0 \right\}, \quad H^n := \mathbb{R}^d \setminus \cup_{i=1}^dH^n_{i}.
        \end{equation}
        Since the switching rates are continuous, $H^n$ is always open.

        \begin{lemma}
        \label{lemma:containment}
            Let $H^n$ and $H$ be as defined in \eqref{eq:hni} and \eqref{eq:setsH} respectively. Let $K \subset H$ be compact and define, for $\delta > 0$,
            \[
                K^{n,\delta} := \left\{x: d(x, K) \le \sqrt{d}n^{-\delta}\right\}, \quad n \ge 1.\numberthis \label{eq:kn_delta}
            \]
            Under the assumptions of Proposition~\ref{propo:rates}, with $\mathbb{P}$-probability $1$, there exists $N$ such that for all $n \ge N$, $K^{n,\delta} \subset H^n$.
            In particular, $K \subset H^n$ for all $n \ge N$.
        \end{lemma}
        \begin{proof}
            Let $K$ be an arbitrary compact subset of $H$ and let $r$ be such that $K_r := \{x: d(x, K) \le r\} \subset H$. Let $q^n_i(x) = \lambda^n_i(x, -\bs{1}_d) + \lambda^n_i(x, \bs{1}_d)$ and $q_i(x) = \lambda_i(x, -\bs{1}_d) + \lambda_i(x, \bs{1}_d)$. Let $c_i = \inf_{x\in K_r}q_i(x)$. Then $c_i > 0$ for each $i$ by definition of $H$.  Let $c = \min_{i}c_i > 0$. By Proposition~\ref{propo:rates}, there exists a $N_{2,i}$ for each $i$, such that, $\inf_{x \in K_r}q^n_i(x) > nc/2$. Let $N_2 = \max\{N_{2,i}, i = 1, \dots, d\}$. Then, for all $n \ge N_2$, it follows that $q^n_i(x) > 0$ for all $x \in K_r$ and $i = 1, \dots, d$. Thus, $K_r \subset H^n$ for all $n \ge N_2$. Let $N_1$ be such that $\sqrt{d}n^{-\delta} < r$ for all $n \ge N_1$. Then, $K^{n,\delta} \subset K_r$ for all $n \ge N_1$. And so, for all $n \ge \max\{N_1, N_2\}$, we get $K^{n, \delta} \subset H^n$. 
        \end{proof}

        \begin{lemma}
        \label{lemma:rates_lipschitz}
            Suppose Assumptions \ref{assum:smoothness} and \ref{assum:moment} hold. With $\mathbb{P}$-probability $1$, for each $i = 1, \dots, d$, there exists a $M_{2, i}: \mathbb{R}^d \to [0, \infty)^d$ and $N$ such that for all $n > N$, 
            \[
               \left|\lambda^n_i(x + h , v) - \lambda^n_i(x, v)\right| \le n\|h\|(1 + \|h\|)M_{2, i}(x),
            \]
            for all $h \in \mathbb{R}^d$ and $v \in \mathcal{V}$. Moreover, for any compact $K \subset \mathbb{R}^{d}$, $\sup_{x \in K}M_{2, i}(x) < \infty$.
        \end{lemma}
        \begin{proof}
            From \eqref{eq:subsampling_rates}, for any $x, h \in \mathbb{R}^d$, and all $i = 1, \dots, d$,
            \begin{align*}
              \left|\lambda^n_i(x + h, v) - \lambda^n_i(x, v)\right| &\le \frac{n}{m|\mathcal{S}_{(n, m)}|}\sum_{S \in \mathcal{S}_{(n, m)}}\sum_{j \in S}\left| E^j_i(x + h) - E^j_i(x)\right| \\
              &= \sum_{j = 1}^n \left| S_i(x + h, Y_j) - S_i(x, Y_j)\right|,
            \end{align*}
            $\mathbb{P}$-almost surely. The last equality is true irrespective of whether the process is ZZ-SS or ZZ-CV or canonical Zig-Zag. For any fixed $y$, a second order Taylor's expansion gives,
            \[
              S_i(x + h, y) - S_i(x, y) = \nabla S_i(x, y)\cdot h + h^T\nabla^{2} S_i(x', y)h/2
            \]
            where $x' = x + \theta h$ for some $\theta \in (0, 1)$. Under the assumption that third derivatives (Assumption~\ref{assum:smoothness}) are bounded,
            \[
              |S_i(x + h, y) - S_i(x, y)| \le \|\nabla S_i(x, y)\|\cdot\|h\| + M'\|h\|^2
            \]
            Consequently, for all $x, y \in \mathbb{R}^d$ and $i = 1, \dots, d$,
            \[
               \left|\lambda^n_i(x + h , v) - \lambda^n_i(x, v)\right| \le \|h\|\sum_{j=1}^n\|\nabla S_i(x, Y_j)\| + M'n\|h\|^2, \quad \mathbb{P}\text{-almost surely}.
            \]

            Due to Assumption \ref{assum:moment}, $\mathbb{E}_{P}\left\| \nabla S_i(x, Y) \right\| < \infty$ for all $x$. Then the law of large numbers implies there exists a constant $M''_i(x)$ such that, $\sum_{j=1}^n\|\nabla S_i(x, Y_j)\| \le nM''_i(x)$, $\mathbb{P}$-almost surely. This gives,
            \[
               \left|\lambda^n_i(x + h , v) - \lambda^n_i(x, v)\right| \le n\|h\|(1 + \|h\|)(M''_i(x) + M'), \quad \mathbb{P}\text{-almost surely}.
            \]
            In particular, take $M''_i(x) = 2\mathbb{E}_{P}\left\| \nabla S_i(x, Y) \right\|$. The smoothness assumption (Assumption \ref{assum:smoothness}) implies that for any compact set $K \subset \mathbb{R}^d$, $\sup_{x \in K} M''_i(x) < \infty$. 
            Setting $M_2 = (M''_i(x) + M')$ yields the desired result.
        \end{proof}

        The next lemma replaces the velocity process in \eqref{eq:Gnf} by the frozen-rate driven $(V^{n,x}_t)$ process.

        \begin{lemma}
        \label{lemma:Gn-Gnx}
           Suppose the assumptions of Proposition~\ref{propo:rates} hold. For each $x \in \mathbb{R}^d$, and $n$, let $(V^{n, x}_t)_{t\ge 0}$ be continuous-time Markov process on $\mathcal{V}$ with infinitesimal generator,
            \[
                G^n_xg(v) = \sum_{i=1}^d \lambda^n_i(x, v)(g(\mathsf{F}_i(v)) - g(v)), \quad v \in \mathcal{V}. \numberthis \label{eq:Gnx}
            \]
           Let $\bar{\mathbb{E}}_{x, v}$ denote the expectation with respect to $(V^{n, x}_t)_{t\ge 0}$ started at $V^{n,x}_0 = v$. Let $f \in C^{\infty}_c(H)$ be arbitrary and let $G^{n,f}$ be as in \eqref{eq:Gnf}.
           Then, $\mathbb{P}$-almost surely,
            \[
                \lim_{n \to \infty} \sup_{(x,v) \in \mathsf{E}}\left|G^{n,f}(x, v) - \left(\frac{1}{n^{-\delta}}\bar{\mathbb{E}}_{x, v} \int_0^{n^{-\delta}} V^{n,x}_s \ ds \right)\cdot \nabla f(x)\right| = 0. \numberthis \label{eq:ln-vnx} 
            \]
        \end{lemma}
        \begin{proof}
            Let $f \in C^{\infty}_c(H)$ be arbitrary and let $K \subset H$ be the compact support of $f$. Let $M$ be such that $\|\nabla f(x)\| \le M$ for all $x$. From \eqref{eq:Gnf}, and using Cauchy-Schwartz inequality,
            \begin{align*}
                &\sup_{(x,v) \in \mathsf{E}}\left|G^{n,f}(x, v) - \left(\frac{1}{n^{-\delta}}\bar{\mathbb{E}}_{x, v} \int_0^{n^{-\delta}} V^{n,x}_s \ ds \right)\cdot \nabla f(x)\right| \\
                &\quad \quad \le Mn^{\delta}\sup_{(x,v) \in K\times \mathcal{V}}\left\|\mathbb{E}_{(x,v)} \int_{0}^{n^{-\delta}} V^n_s\ ds -\bar{\mathbb{E}}_{x, v} \int_0^{n^{-\delta}} V^{n,x}_s \ ds\right\|, \numberthis \label{eq:gn-gnx1}
            \end{align*}
            $\mathbb{P}$-almost surely. Thus, it is enough to show that the right-hand side goes to $0$.
            Fix an observed dataset and the sequence of reference points in the almost sure set of Lemma~\ref{lemma:containment} and Lemma~\ref{lemma:rates_lipschitz}. For any given $(x, v) \in K \times \mathcal{V}$, we will reconstruct the parent Zig-Zag process $Z^n_t$ along with the $V^{n, x}_t$ process on a common probability space uptil time $n^{-\delta}$ and with initial conditions $Z^n_0 = (x, v)$ and $V^{n,x}_0 = v$.

            Let $K^{n,\delta}$ be as in \eqref{eq:kn_delta}. By Lemma~\ref{lemma:containment}, there exits $N$ such that for all $n \ge N$, $K^{n, \delta} \subset H^n$. By construction, $K^{n, \delta}$ contains every position reachable by Zig-Zag process starting from $K$ before time $n^{-\delta}$. Thus, for all $n \ge N$, position paths of Zig-Zag process starting from $K$ upto time $n^{-\delta}$  lie in the `good' region $H^n$.
            In the following, suppose $n \ge N$. 
            
            Let $q^n_i(x) = \lambda^n_i(x, \bs{1}_d) + \lambda^n_i(x, -\bs{1}_d)$ for $i = 1, \dots, d$ and $q^n(x) = \sum_{i=1}^d q^n_i(x)$. Define
            \[
                \nu^n = \sup\left\{q^n(x): x \in K^{n, \delta}\right\}.
            \]
            By construction of $H^n$, $\nu^n > 0$ for $n \ge N$ and $\nu^n < \infty$ by continuity of the switching rates.            
            For $x \in K$, define, for each $i = 1, \dots, d$,
            \begin{align*}
                p_{0, i}^{n, x}(y, v) := \frac{\min\{\lambda^n_i(y, v), \lambda^n_i(x, v)\}}{2\nu^n},\\
                p_{1, i}^{n, x}(y, v) := \frac{(\lambda^n_i(y, v) - \lambda^n_i(x, v))_{+}}{2\nu^n},\\
                p_{2, i}^{n, x}(y, v) := \frac{(\lambda^n_i(y, v) - \lambda^n_i(x, v))_{-}}{2\nu^n},
            \end{align*}
            for $y \in K^{n,\delta}$ and $v \in \mathcal{V}$. We can chose an auxiliary measurable space $(S, \mathcal{S})$, a collection of probability measures $\{\mathbb{P}^{n, x}, x \in K\}$ on $(S, \mathcal{S})$, and a measurable function $F:\mathbb{R}^d \times\mathcal{V}\times\mathcal{V}\times S \to \mathbb{R}^d\times\mathcal{V}\times\mathcal{V}$ satisfying the following for all $x \in K$: for any $y \in K^{n,\delta}$ and $v\in \mathcal{V}$,
            \begin{align*}
                \mathbb{P}^{n, x}(u \in S: F(y, v, v, u) = (y, \mathsf{F}_i(v), \mathsf{F}_i(v))) &= p_{0, i}^{n, x}(y, v), \\
                \mathbb{P}^{n, x}(u \in S: F(y, v, v, u) = (y, \mathsf{F}_i(v), v)) &= p_{1, i}^{n, x}(y, v), \\
                \mathbb{P}^{n, x}(u \in S: F(y, v, v, u) = (y, v, \mathsf{F}_i(v))) &= p_{2, i}^{n, x}(y, v),
            \end{align*}
            and,
            \begin{align*}
                \mathbb{P}^{n, x}(u \in S: F(y, v, v, u) = (y, v, v)) &= 1 - \sum_{i=1}^d(p^{n,x}_{0, i}(y, v) + p^{n,x}_{1, i}(y, v) + p^{n,x}_{2, i}(y, v)) \\
                &= 1 - \frac{1}{2\nu^n}\sum_{i=1}^d\max\{\lambda^n_i(y, v), \lambda^n_i(x, v)\}.
            \end{align*}
            for all $i = 1, \dots, d$. Also, for any $v' \in \mathcal{V}$, $v' \neq v$,
            \begin{gather*}
                \mathbb{P}^{n, x}(u \in S: F(y, v, v', u) = (y, \mathsf{F}_i(v), \mathsf{F}_j(v'))) = \frac{\lambda^n_i(y, v)}{2\nu^n}\cdot \frac{\lambda^n_j(x, v')}{2\nu^n}, \\
                \mathbb{P}^{n, x}(u \in S: F(y, v, v', u) = (y, \mathsf{F}_i(v), v')) = \frac{\lambda^n_i(y, v)}{2\nu^n}\cdot \left(1 - \frac{\sum_{i=1}^d\lambda^n_i(x, v')}{2\nu^n}\right), \\
                \mathbb{P}^{n, x}(u \in S: F(y, v, v', u) = (y, v, \mathsf{F}_j(v'))) = \left(1 - \frac{\sum_{i=1}^d\lambda^n_i(y, v)}{2\nu^n}\right)\cdot \frac{\lambda^n_j(x, v')}{2\nu^n}, \\
                \mathbb{P}^{n, x}(u \in S: F(y, v, v', u) = (y, v, v')) = \left(1 - \frac{\sum_{i=1}^d\lambda^n_i(y, v)}{2\nu^n}\right)\cdot \left(1 - \frac{\sum_{i=1}^d\lambda^n_i(x, v')}{2\nu^n}\right).
            \end{gather*}
            for all $i, j \in 1, \dots, d$. Finally, for any $y \notin K^{n,\delta}$ and $v, v' \in \mathcal{V}$,
            \[
                \mathbb{P}^{n, x}(u \in S: F(y, v, v', u) = (x + vn^{-\delta}, v, v')) = 1
            \]
            
            Fix $(x, v) \in K \times \mathcal{V}$. Let $(N^n_t)_{t \ge 0}$ be a time-homogeneous Poisson process with intensity $2\nu^n$. Independently of $(N^n_t)_{t\ge 0}$, let $\{U_k\}_{k=1}^{\infty}$ be a sequence of i.i.d. $S$-valued random variables with common distribution $\mathbb{P}^{n, x}$. Let $T_1, T_2, \dots$ be the sequence of arrival times for $(N^n_t)_{t\ge 0}$ with $T_0 = 0$. Set $(X^{n}_0, V^{n}_0) = (x, v)$ and $V^{n, x}_0 = v$ and recursively define for $k \ge 1$,
            \begin{gather*}
                (X^n_t, V^n_t, V^{n, x}_t) = (X^n_{T_{k-1}} + V^n_{T_{k-1}}(t - T_{k-1}), V^n_{T_{k-1}}, V^{n, x}_{T_{k-1}}), \quad t \in [T_{k-1}, T_{k}),\\
                (X^n_{T_k}, V^n_{T_{k}}, V^{n, x}_{T_k}) = F(X^n_{T_{k-1}} + V^n_{T_{k-1}}(T_{k} - T_{k-1}), V^n_{T_{k-1}}, V^{n, x}_{T_{k-1}}, U_k).
            \end{gather*}
            By construction, the process $(X^n_t, V^n_t, V^{n, x}_t)_{t < n^{-\delta}}$ is a PDMP with event rate $2\nu^n$ and transition kernel determined by the law of the random variable $F$. At each event time $T_k$, if $V^n_{T^-_k} = V^{n, x}_{T^-_k}$, there are three possibilities: both $V^n$ and $V^{n, x}$ jump together to the same location, or only one of them jumps, or neither of them jumps. On the other hand if $V^n_{T^-_k} \neq V^{n, x}_{T^-_k}$, then they both jump independently of each other with certain probabilities. The probabilities are chosen such that marginally, $(X^n_t, V^n_t)_{t < n^{-\delta}}$ is a Zig-Zag process on $\mathbb{R}^d \times \mathcal{V}$ with initial state $(X^n_0, V^n_0) = (x, v)$ and rates $\{\lambda^n_i\}_{i=1}^d$. Moreover, the process $(V^{n, x}_t)_{t \le n^{-\delta}}$ is marginally a continuous time Markov process on the $\mathcal{V}$ space with $V^{n, x}_0 = v$ and generator \eqref{eq:Gnx}.

            Let $\tau^n_c$ be the first time less than or equal to $n^{-\delta}$ such that $V^n_s$ and $V^{n, x}_s$ are not equal, i.e.,
            \[
                \tau^n_c := \inf\{0 \le s < n^{-\delta}: V^n_s \neq V^{n, x}_s\}.
            \]
            If no such time exists, set $\tau^n_c = n^{-\delta}$. We will call this time as the {\it decoupling time}. Let $\mathbb{P}^{n,x}_v$ denote the law of the coupled processes and let $\mathbb{E}^{n,x}_v$ be the corresponding expectation. Then, under the construction above,
            \begin{align*}
                n^{\delta}\left\|\mathbb{E}_{x, v} \int_0^{n^{-\delta}} V^n_s \ ds -  \bar{\mathbb{E}}_{x, v} \int_0^{n^{-\delta}} V^{n,x}_s \ ds\right\| & = n^{\delta}\left\|\mathbb{E}^{n, x}_v \left(\int_0^{n^{-\delta}} V^n_s - V^{n,x}_s \ ds \right)\right\| \\
                &\le n^{\delta}\mathbb{E}^{n, x}_v \left(\int_{\tau^n_c}^{n^{-\delta}} \left\|V^n_s - V^{n,x}_s \right\|\ ds \right)\\
                &\le 2\sqrt{d}\mathbb{E}^{n, x}_v \left(1 - \frac{\tau^n_c}{n^{-\delta}}\right)\\
                &\le 2\sqrt{d}\mathbb{P}^{n, x}_v \left(\tau^n_c < n^{-\delta}\right). \numberthis \label{eq:v-vnx_bound}
            \end{align*}
            The result will follow if we can show that the probability of decoupling before time $n^{-\delta}$ goes to $0$ as $n \to \infty$, uniformly over $x \in K$.

            Let $p^n(s)$ be the probability that $V^n$ and $V^{n, x}$ decouple at time $s$ given that they were coupled till time $s$. 
            Then, for $s < n^{-\delta}$, from Lemma \ref{lemma:rates_lipschitz},
            \begin{align*}
                p^n(s) = p^{n,x}_{1}(X^n_s, V^n_{s^-}) + p^{n,x}_{2}(X^n_s, V^n_{s^-}) &= \frac{1}{2\nu^n}\sum_{i=1}^d\left|\lambda^n_i(X^n_s, V^n_{s^-}) - \lambda^n_i(x, V^n_{s^-})\right| \\
                &\le \frac{1}{2\nu^n} \sqrt{d}n^{1-\delta}(1 + \sqrt{d}n^{-\delta})\sum_{i=1}^d M_{2, i}(x), \numberthis \label{eq:pns_bound}
            \end{align*}
            for all $s < n^{-\delta}$, it holds that $\|X^n_s - x\| \le \sqrt{d}n^{-\delta}$. By construction, the decoupling can happen only at event times of the Poisson process $(N^n_t)$. Then, using tower property,
            \begin{align*}
                \Prob^{n,x}_v(\tau^n_c < n^{-\delta}) = \mathbb{E}^{n,x}_v\left[\sum_{k=1}^{N^{n}_{n^{-\delta}}} \Prob^{n,x}_v(\tau^n_c = T_k \mid T_1, T_2 \dots, )\right]
                \le \mathbb{E}^{n,x}_v\left[\sum_{k=1}^{N^{n}_{n^{-\delta}}} p^n(T_k)\right]\numberthis \label{eq:probtaun}.
            \end{align*}
            Using \eqref{eq:pns_bound} in \eqref{eq:probtaun},
            \begin{align*}
                 \Prob^{n,x}_v(\tau^n_c < n^{-\delta}) &\le \frac{1}{2\nu^n} \sqrt{d}n^{1-\delta}(1 + \sqrt{d}n^{-\delta})\left(\sum_{i=1}^d M_{2, i}(x)\right)\mathbb{E}^{n,x}_v\left(N^{n}_{n^{-\delta}}\right) \\
                 &= \sqrt{d}n^{1-2\delta}(1 + \sqrt{d}n^{-\delta})\left(\sum_{i=1}^d M_{2, i}(x)\right), 
            \end{align*}
            where the last equality comes from the fact $\mathbb{E}^{n,x}_v(N^{n}_{n^{-\delta}}) = 2\nu^nn^{-\delta}$. Combining this with \eqref{eq:gn-gnx1} and \eqref{eq:v-vnx_bound},
            \begin{align*}
                &\left\|G^{n,f}(x, v) - \left(\frac{1}{n^{-\delta}}\bar{\mathbb{E}}_{x, v} \int_0^{n^{-\delta}} V^{n,x}_s \ ds \right)\cdot \nabla f(x)\right\|_{\infty} \\
                &\quad \quad \le 2Mdn^{1-2\delta}(1 + \sqrt{d}n^{-\delta})\sup_{x\in K}\left(\sum_{i=1}^d M_{2, i}(x)\right)
            \end{align*}
            Since $M_{2, i}$'s are bounded on compact sets by Lemma \ref{lemma:rates_lipschitz} and $\delta > 1/2$, the right-hand side goes to $0$ uniformly on compact sets as $n \to \infty$. The claim follows.
        \end{proof}

        Next, we will approximate the integral in \eqref{eq:ln-vnx} by a law of large numbers for the replacement process $(V^{n,x}_t)$. For each $n$, define $b^n$ coordinate-wise by,
            \begin{equation}
            \label{eq:drift_n}
                b^n_i(x) = \begin{cases}
                \dfrac{\lambda^n_i(x, -\bs{1}_d) - \lambda^n_i(x, \bs{1}_d)}{\lambda^n_i(x, -\bs{1}_d) + \lambda^n_i(x, \bs{1}_d)}, & x \notin H^n_i, \\
                0, & \text{otherwise}.
            \end{cases}
            \end{equation}

        \begin{lemma}
        \label{lemma:zz_trans_2}
           Suppose the assumptions of Proposition~\ref{propo:rates} hold. Let $b^n$ be as defined in \eqref{eq:drift_n}. For each $x \in \mathbb{R}^d$, and $n$, let $(V^{n,x}_t)$ be as in Lemma~\ref{lemma:Gn-Gnx} and let $\bar{\mathbb{E}}_{x, v}$ denote the expectation with respect to $(V^{n, x}_t)_{t\ge0}$ started at $V^{n,x}_0 = v$. For all $f \in C^{\infty}_c(H)$ extended to $\mathbb{R}^d$ by $0$, $\mathbb{P}$-almost surely,
            \[
                \lim_{n \to \infty}\sup_{(x,v) \in \mathsf{E}}\left|\left(\frac{1}{n^{-\delta}}\bar{\mathbb{E}}_{x, v} \int_0^{n^{-\delta}} V^{n,x}_s \ ds \right)\cdot \nabla f(x) - b^n(x) \cdot  \nabla f(x)\right| = 0.
            \]
        \end{lemma}
        \begin{proof}
            Let $f \in C^{\infty}_c(H)$ be arbitrary and let $K \subset H$ be the compact support of $f$. As in the proof of Lemma~\ref{lemma:Gn-Gnx}, it is enough to show that,
            \[
                \lim_{n \to \infty}\sup_{(x, v) \in K \in \mathcal{V}}\left\|\frac{1}{n^{-\delta}}\bar{\mathbb{E}}_{x, v} \int_0^{n^{-\delta}} V^{n,x}_s \ ds  - b^n(x)\right\| = 0, \quad \mathbb{P}\text{-almost surely}.
            \]
            
            From \eqref{eq:subsampling_rates}, observe that $\lambda^n_i$ depends only on the $i$-th coordinate of $v$. Similarly, the $i$-th flip operator $\mathsf{F}_i$ also only depends on the $i$-th coordinate of $v$. Thus, $G^n_x$ in \eqref{eq:Gnx} can be written as a sum of $d$ independent operators each acting on only one of the coordinates of $v$. In process speak, for each $i$, $(V^{n, x}_s)_i$ is a two-state continuous-time Markov process on $\{-1, 1\}$ with rate $\lambda^n_i(x, \pm \bs{1}_d)$, independently of other coordinates. Moreover since,
            \begin{align*}
              \left\|\frac{1}{n^{-\delta}}\bar{\mathbb{E}}_{x, v} \int_0^{n^{-\delta}} V^{n,x}_s \ ds - b^n(x)\right\| \le \sum_{i=1}^d \left|\frac{1}{n^{-\delta}}\bar{\mathbb{E}}_{x, v} \int_0^{n^{-\delta}} (V^{n,x}_s)_{i} \ ds - b^n_i(x)\right|, \numberthis \label{eq:vnx_split}
            \end{align*}
            it is sufficient to analyze each two-state process independently.

            Fix $(x,v) \in K \times \mathcal{V}$ and let $i \in \{1, \dots, d\}$ be arbitrary. Fix an observed dataset in the almost sure set of Lemma~\ref{lemma:containment}. Let $q^n_i(x) = \lambda^n_i(x, -\bs{1}_d) + \lambda^n_i(x, \bs{1}_d)$. From the proof of Lemma~\ref{lemma:containment}, there exists an $N$ and $c > 0$ such that for all $n \ge N$, $q^n_i(x) > nc/2$ for all $x \in K$ and $i = 1, \dots, d$. The $(V^{n, x}_t)_i$ process is ergodic on $\{-1, +1\}$ with corresponding stationary distribution $\pi^{n, x}_i$ given by,
            \[
              \pi^{n, x}_i(\{+1\}) = 1 - \pi^{n, x}_i(\{-1\}) = \frac{\lambda^n_i(x, -\bs{1}_d)}{q^n_i(x)},
            \]          
            and for all $t < n^{-\delta}$,
            \[
                \bar{\mathbb{E}}_{x, v} (V^{n,x}_t)_i = b^n_i(x) + 2v_i\pi^{n,x}_i(\{-v_i\})\cdot e^{- t q^n_i(x)}.
            \]
            Thus,
            \begin{align*}
                \left|\frac{1}{n^{-\delta}}\bar{\mathbb{E}}_{x, v} \int_0^{n^{-\delta}} (V^{n,x}_s)_{i} \ ds - b^n_i(x)\right| &= \left|\frac{1}{n^{-\delta}} \int_0^{n^{-\delta}} 2v_i\pi^{n,x}_i(\{-v_i\})\cdot e^{- sq^n_i(x)} \ ds \right|\\
                &\le 2\frac{(1 -  e^{- q^n_i(x)n^{-\delta}})}{n^{-\delta}q^n_i(x)} \\
                &\le 4n^{\delta - 1}/c.\numberthis \label{eq:bound1}
            \end{align*}
            Using this in \eqref{eq:vnx_split}, for all $(x, v) \in K \times \mathcal{V}$,
            \[
                \left\|\frac{1}{n^{-\delta}}\bar{\mathbb{E}}_{x, v} \int_0^{n^{-\delta}} V^{n,x}_s \ ds - b^n(x)\right\| \le 4dn^{\delta - 1}/c.
            \] 
            The right-hand side goes to uniformly on $K$ since $1/2 < \delta < 1$. The claim follows. 
        \end{proof}

        Now, we show that the sequence of finite sample drifts $b^n$ converges to the asympotic drift $b$ uniformly on the compact sets.

        \begin{lemma}
        \label{lemma:drifts} 
            Let $b^n$  and $b$ be as defined in \eqref{eq:drift_n} and \eqref{eq:limiting_drift} respectively. Under the assumptions of Proposition~\ref{propo:rates}, for all compact $K \subset H$, and all $i = 1, \dots, d$
            \[
                \lim_{n \to \infty}\sup_{x \in K}|b^n_i(x) - b_i(x)\| = 0, \quad \mathbb{P}\text{-almost surely.}
            \] 
        \end{lemma}
        \begin{proof}
            For an arbitrary $i$, $b^n_i$ can be rewritten as,
            \begin{equation}
                    b^n_i(x) = \begin{cases}
                    \dfrac{\lambda^n_i(x, -\bs{1}_d)/n - \lambda^n_i(x, \bs{1}_d)/n}{\lambda^n_i(x, -\bs{1}_d)/n + \lambda^n_i(x, \bs{1}_d)/n}, & x \notin H^n_i, \\
                    0, & \text{otherwise}.
                \end{cases}
            \end{equation}
            By Proposition~\ref{propo:rates}, there exists a set $A \subset \Omega$ with $\Prob(A) = 1$ such that for all $\omega \in A$, the numerator and the denominator in $b^n_i$ converge individually to $\lambda_i(x, -\bs{1}_d) - \lambda_i(x, \bs{1}_d)$ and $\lambda_i(x, -\bs{1}_d) + \lambda_i(x, \bs{1}_d)$ respectively, uniformly on all compact subsets of $\mathbb{R}^d$. For each such $\omega$, let $K$ be an arbitrary compact subset of $H$. Since $\lambda_i$ is continuous, there exists an $\epsilon > 0$ such that $\lambda_i(x, -\bs{1}_d) + \lambda_i(x, \bs{1}_d) >\epsilon$ for all $x \in K$. Then, eventually, $(\lambda^n_i(x, -\bs{1}_d) + \lambda^n_i(x, \bs{1}_d))/n > \epsilon/2$ uniformly on $K$. Thus, the fraction $b^n_i$ converges uniformly to $b_i$ on $K$ for all $i = 1, \dots, d$ by Proposition~\ref{propo:rates}. 
        \end{proof}

        For any $f \in C^1(\mathbb{R}^d)$ define
        \[
            g_f(x) = \begin{cases}
                b(x)\cdot \nabla f(x), & x \in H,\\
                0, & x \notin H. \numberthis \label{eq:Lf}
            \end{cases}
        \]      
        Let $(g_f)_x$ denote its extension to $\mathsf{E}$, i.e., $(g_f)_x(x, v) = g_f(x)$. We finally combine all the lemmas to get the convergence of $\mathcal{L}^nf_n$ to $g_f$ for all $f \in C^{\infty}_c(H)$.

        \begin{lemma}
        \label{lemma:ln-bf}
            Under the assumptions of Proposition~\ref{propo:rates}, for all $f \in C^{\infty}_{c}(H)$,
            \[
                \lim_{n \to \infty}\left\|\mathcal{L}^nf_n - (g_f)_x\right\|_{\infty} = 0, \quad \mathbb{P}\text{-almost surely}.
            \]
        \end{lemma}
        \begin{proof}
            Combining Lemma~\ref{lemma:gen_approx_1}, \ref{lemma:Gn-Gnx}, and \ref{lemma:zz_trans_2}, for any $f \in C^{\infty}_{c}(H)$,
            \[
                \lim_{n \to \infty}\sup_{(x,v) \in \mathsf{E}}\left|\mathcal{L}^nf_n(x,v) - b^n(x)\cdot \nabla f(x)\right| = 0, \quad \mathbb{P}\text{-almost surely}. \numberthis \label{eq:bound_11}
            \]
            Let $K \subset H$ be the compact support of $H$ and let $M$ be such that $\|\nabla f(x)\| \le M$ for all $x$. Using triangle inequality, Cauchy-Schwartz inequality, and from \eqref{eq:Lf},
            \begin{align*}
                &\left\|\mathcal{L}^nf_n - (g_f)_x\right\|_{\infty} \\
                &\quad \le \sup_{(x,v) \in \mathsf{E}}\left|\mathcal{L}^nf_n(x,v) - b^n(x)\cdot \nabla f(x)\right| + \sup_{(x,v) \in \mathsf{E}}\left|b^n(x)\cdot \nabla f(x) - (g_f)_x(x,v)\right| \\
                &\quad \le \sup_{(x,v) \in \mathsf{E}}\left|\mathcal{L}^nf_n(x,v) - b^n(x)\cdot \nabla f(x)\right| + M\sup_{x \in K}\left\|b^n(x) - b(x)\right\|
            \end{align*}
            The claim follows from \eqref{eq:bound_11} and Lemma~\ref{lemma:drifts}.    
        \end{proof}

        We are now ready to prove Theorem~\ref{thm:transient}.

    \subsection{Proof of Theorem~\ref{thm:transient}}
        Fix an observed dataset, a sequence of reference points, and a realization of $X^*$, in the common almost sure set of Lemma~\ref{lemma:f_perturbation} and Lemma~\ref{lemma:ln-bf} and for which the initial distributions converge weakly to some $\delta_{x_0}$ for some $x_0 \in H$. 
        Throughout this proof, we work conditionally on these fixed quantities. In particular, $b$, $H$, and $\tau$ in \eqref{eq:zz_ode} are deterministic.
                
        Let $T > 0$ be arbitrary. By hypothesis that $X^n_0$ converge weakly, we have that the sequence $\{X^n_0\}_{n \in \mathbb{N}}$ is tight. Furthermore, by construction,
        \[
            w_n(\delta) := \sup\left\{\|X^n_s - X^n_t\|: 0 \le s,t \le T,\ |s -t|\le \delta\right\} \le \delta\sqrt{d} \quad \text{ for all } \delta > 0, \ n \in \mathbb{N}.
        \]
        Thus, by \cite[][Theorem 7.3]{billingsley_convergence_1999}, the sequence of paths $\{(X^n_t)_{0 \le t\le T}\}_{n \in \mathbb{N}}$ is tight, and hence relatively compact, in $C([0,T], \mathbb{R}^d)$.

        Let $A = \{(f, g_f): f\in C^{\infty}_c(H)\}$ where $g_f$ is as defined in \eqref{eq:Lf} and $f$ is extended to $\mathbb{R}^d$ by $0$. 
        Let $T < \tau$ be arbitrary. Since $b$ is locally Lipschitz on $H$, $(X_t)_{0\le t\le T}$ is the unique solution to the continuous martingale problem for $(A, x_0)$ on $[0, T]$. 
        
        For each $(f, g_f) \in A$, Lemma~\ref{lemma:f_perturbation} and Lemma~\ref{lemma:ln-bf} provide pairs $(f_n, \mathcal{L}^nf_n) \in (\mathcal{L}^n, \mathcal{D}(\mathcal{L}^n))$ satisfying  $\sup_n\|f_n\|_{\infty} < \infty$, and,
        \[
            \lim_{n \to \infty}\|f_n - f_x\|_{\infty} = \lim_{n\to \infty}\|\mathcal{L}^nf_n - (g_f)_x\|_{\infty} = 0.
        \]
        By \cite[][Chapter~4, Theorem~8.10 and Corollary~8.12]{ethier_markov_1986}, $(X^n_t)_{0 \le t \le T}$ converges weakly (in uniform topology on $C([0, T], \mathbb{R}^d)$) to $(X_{t})_{0 \le t \le T}$ as $n \to \infty$. Since the argument holds for any fixed dataset, sequence of reference points, and a realization of $X^*$, in an almost sure set, the $\mathbb{P}$-almost surely claim follows. \qed

\section{Trivial fluid limit in the stationary phase}

    \subsection{Proof of Theorem~\ref{thm:canonical}}
    \label{sec:trivial_proof}
        Let $t > 0$ and $\epsilon > 0$ be arbitrary. Let $h = \epsilon/(2\sqrt{d})$ and $N = \lceil t/h\rceil$. Then, for any $k < N$, and all $u \in [kh, (k+1)h]$, $\|X^n_{u} - X^n_{kh}\| \le \epsilon/2$ since the position component moves with unit speed. Thus, if $\sup_{s \le t}\|X^n_s - \hat{x}_n\| > \epsilon$, then $\|X^n_{kh} - \hat{x}_n\| > \epsilon/2$ at least for one $k$ such that $kh < t$. Then,
        \begin{align*}
            \mathbb{P}\left(\sup_{0 \le s \le t} \|X^n_s - \hat{x}_n\| > \epsilon \right) &\le \mathbb{P}\left(\exists k: kh < t, \ \|X^n_{kh} - \hat{x}_n\| > \epsilon/2 \right)\\
            &\le \sum_{k: kh < t}\mathbb{P}\left(\|X^n_{kh} - \hat{x}_n\| > \epsilon/2\right).
        \end{align*}
        But from the stationarity of Zig-Zag process, it follows that $\mathbb{P}(\|X^n_{kh} - \hat{x}_n\| > \epsilon/2) = \Pi^n(\{x: \|x - \hat{x}_n\| > \epsilon/2\})$. Using this in the above displayed inequality,
        \[
            \mathbb{P}\left(\sup_{0 \le s \le t} \|X^n_s - \hat{x}_n\| > \epsilon \right) \le (N + 1)\Pi^n(\{x: \|x - \hat{x}_n\| > \epsilon/2\}).
        \]
        The right-hand side goes to $0$ as $n \to \infty$ by Assumption~\ref{assum:mle_consistent}(iii).

        For second part, observe that,
        \[  
            \sup_{0\le s\le t} \|X^n_s - x_0\| \le \sup_{0 \le s \le t}\|X^n_s - \hat{x}_n\| + \|\hat{x}_n - x_0\|.
        \]
        Then, the result follows from first part and Assumption~\ref{assum:mle_consistent}(i). \qed

\section{Proofs for ZZ-SS in the stationary phase}
\label{sec:proof_zzss}

    We will prove Theorem \ref{theo:zzss_stat_fixed} following the approach of \cite{bierkens_limit_2017}. Throughout this section, we will suppose that Assumptions~\ref{assum:smoothness}, \ref{assum:unique_minimiser}, \ref{assum:mle_consistent}, all hold. Once again, we suppress the dependence on $\omega$ in all of the following.

            The following observation will be useful later and follows directly from the smoothness of $S$ and the consistency of the sequence $\hat{x}_n$.

   \begin{lemma}
    \label{lemma:In_convergence}
        Recall that $U_n = -\log \pi^{(n)}$ is the potential function. Define 
        \[
            I(\hat{x}_n) := \frac{1}{n}\nabla^2U_n(x) = \frac{1}{n}\sum_{j=1}^n -\nabla^2_x \log f(Y_j; \hat{x}_n), \quad n \in \mathbb{N}. \numberthis \label{eq:ixn}
        \] 
        Suppose Assumptions \ref{assum:smoothness}, \ref{assum:unique_minimiser}, and \ref{assum:mle_consistent} hold. Let $I(x_0)$ be as defined in Assumption~\ref{assum:unique_minimiser}. Then,
        \[
            I(\hat{x}_n) \to I(x_0).
        \]
        as $n \to \infty$ in $\Prob-$probability.
    \end{lemma}
    \begin{proof}
        Note that $-\nabla^2_x \log f(y; x) = \left(\left(-\frac{\partial \log f(y; x)}{\partial x_i \partial x_j}\right)\right)_{i, j=1, \dots, d}$ is a $d \times d$ matrix. We will show the convergence element-wise. For any arbitrary $k, l \in \{1, \dots, d\}$, let $S'_{kl}$ denote the function $(x, y) \mapsto -\frac{\partial \log f(y; x)}{\partial x_k \partial x_k}$. A Taylor's expansion about $x_0$ gives,
        \[
            S'_{kl}(\hat{x}_n, y) = S'_{kl}(x_0, y) + \nabla S'_{kl}(x_0 + \theta'(y)(\hat{x}_{n} - x_0), y)(\hat{x}_n - x_0),
        \]
        for some $\theta'(y) \in (0, 1)$. Then, by Assumption \ref{assum:smoothness}
        \begin{align*}
            |S'_{kl}(\hat{x}_n, y) - S'_{kl}(x_0, y)|  = | \nabla s'_{kl}(x_0 + \theta'(y)(\hat{x}_{n} - x_0), y)(\hat{x}_n - x_0)| \le M|\hat{x}_n - x_0|,
        \end{align*}
        for all $y \in \mathbb{R}^d$. Consequently,
        \[
            \left|\frac{1}{n}\sum_{j=1}^nS'_{kl}(\hat{x}_n, Y_j) - \frac{1}{n}\sum_{j=1}^n S'_{kl}(x_0, Y_j)\right| \le M|\hat{x}_n - x_0|
        \]
        The result then follows from the consistency of $\hat{x}_n$ and the law of large numbers.
    \end{proof}
        \subsection{Technical lemmas}
        \label{sec:technical_ss}
        Recall, from \eqref{eq:rates_ss_stat}, that the ZZ-SS switching rates in the reparameterized coordinate for a subsample of size $m(n)$ are given by,
        \[
            \lambda^n_{\text{ss}, i}(\xi, v; m(n)) = \frac{n}{m(n)|\mathcal{J}_{(n, m)}|}\sum_{J \in \mathcal{J}_{(m, n)}} \left(v_i \cdot \sum_{j \in J} S(\hat{x}_n + n^{-1/2}\xi, Y_j)\right)_{+},
        \]
        for all $ (\xi, v) \in \mathsf{E},\ i =1, \dots d$.

    Suppose $m(n) \to m < \infty$ and let $N$ be such that $m(n) = m$ for all $n \ge N$. In the following, $n \ge N$. We drop the subscript `ss' and $m$ from arguments for notational convenience.  
    For each $i = 1, \dots, d$, define
    \[
        q^n_i(\xi) = \lambda^n_i(\xi, -\bs{1}_d) + \lambda^n_i(\xi, \bs{1}_d), \quad \xi \in \mathbb{R}^d\numberthis \label{eq:hni_stat}.
    \]

    \begin{lemma}
    \label{lemma:tilde_hn}
        Let $A = ((A_{ij}))_{ij\le d}$ be the diagonal matrix in \eqref{eq:A_def}. For all $R > 0$, and each $i = 1, \dots, d$,
            \[
                \lim_{n \to \infty}\sup_{\|\xi\| \le R}\left| \frac{q^n_i(\xi)}{n} - \frac{2}{A_{ii}}\right| = 0, \quad \text{in }\Prob\text{-probability}.
            \]
    \end{lemma}
    \begin{proof}
        Consider,
            \begin{equation*}
                \frac{q^n_i(\xi)}{n} = \frac{\lambda^n_i(\xi, -\bs{1}_d) + \lambda^n_i(\xi, \bs{1}_d)}{n}
                = \frac{1}{m|\mathcal{J}_{(n, m)}|}\sum_{J \in \mathcal{J}_{(m, n)}}\left| \sum_{j \in J} S^j_i(\hat{x}_n + n^{-1/2}\xi, Y_j)\right|.
            \end{equation*}  
            Let $x_n' = \hat{x}_n + n^{-1/2}\xi - x_0$. For any fixed $y$, by a second-order Taylor expansion about $x_0$,
            \[
              S_i(\hat{x}_n + n^{-1/2}\xi, y) = S_i(x_0, y) + x_n'\cdot\nabla S_i(x_0, y) + \frac{1}{2}(x_n')^T \nabla^{2}S_i(x_0 + \theta x', y)x_n', \numberthis \label{eq:taylor1}
            \]
            for some $\theta \in (0, 1)$.
            Then, from \eqref{eq:taylor1}, there exist $\theta_j$s such that,
            \begin{align*}
                &\left|\frac{q^n_i(\xi)}{n} - \frac{1}{m|\mathcal{J}_{(n, m)}|}\sum_{J \in \mathcal{J}_{(m, n)}}\left|\sum_{j \in J} S_i(x_0, Y_j)\right|\right| \\
                &\quad \quad \le \frac{1}{n}\sum_{j=1}^n\left|S_i(\hat{x}_n + n^{-1/2}\xi, Y_j) - S_i(x_0, Y_j)\right| \\
                &\quad \quad = \frac{1}{n}\sum_{j=1}^n\left|x_n'\cdot\nabla S_i(x_0, Y_j) + \frac{1}{2}(x_n')^T \nabla^{2}S_i(x_0 + \theta_j x', Y_j)x_n'\right| \\
                &\quad \quad \le |x'_n| \cdot \left(\frac{1}{n}\sum_{j=1}^n\left|\nabla S_i(x_0, Y_j)\right|\right) + \frac{1}{2}M'|x'_n|^2,
            \end{align*}
            where the last inequality is due to Assumption~\ref{assum:smoothness}. By the consistency of $\hat{x}_n$ (Assumption \ref{assum:mle_consistent}), $x'_n$ goes to $0$ in $\mathbb{P}-$probability, uniformly in $\xi$. Also, by law of large numbers, $n^{-1}\sum_{j=1}^n\left|\nabla S_i(x_0, Y_j)\right| \to \mathbb{E}|\nabla S_i(x_0, Y)|$, $\mathbb{P}$-almost surely. Then, by Assumption~\ref{assum:unique_minimiser} and Slutsky's theorem, the right-hand side goes to $0$ uniformly in $\xi$. Then, by Proposition \ref{propo:slln_ustatistics}, for all $i = 1, \dots, d$,
            \[
                \frac{q^n_i(\xi)}{n} \to \mathbb{E}_P\left[\left|m^{-1}\sum_{j=1}^m S_i(x_0, Y_j) \right|\right], \quad \text{in }\Prob\text{-probability},
            \] 
            uniformly on compact sets. The result follows from \eqref{eq:A_def}.
    \end{proof}

        Since $P$ has a continuous density with respect to a non-atomic measure, $2/A_{ii}$ is strictly positive for any $i$. A consequence of Lemma~\ref{lemma:tilde_hn} is that on any closed ball, with probability going to $1$ as $n \to \infty$, $q^n_i$ is strictly positive and bounded away from $0$.
        

        Let $\eta: [0, \infty) \to [0,1]$ be a smooth function such that $\eta(u) = 0$ for $u \le 1$ and $\eta(u) = 1$ for $u \ge 2$. For $\epsilon > 0$, let $r_{\epsilon}: [0, \infty) \to \mathbb{R}$ be the function $u \mapsto u^{-1}\eta(u/\epsilon)$. Then, $r(u) = 0$ for all $u \le \epsilon$ and $r(u) = 1/u$ for all $u > 2\epsilon$. Since $\eta$ is smooth, $r(u)$ goes to $0$ smoothly as $u$ goes from $2\epsilon$ to $\epsilon$. Moreover, $|r_{\epsilon}(u)| \le 1/\epsilon$ for all $u$. For each $i = 1, \dots, d$, define
        \[
            h^n_i(\xi) = r_{n^{-1/4}}\left(q^n_i(\xi)\right), \quad \xi \in \mathbb{R}^d. \numberthis \label{eq:hni_stat1}
        \]
        If $q^n_i(\xi) \ge 2n^{-1/4}$, then $h^n_i(\xi) = 1/q^n_i(\xi)$, i.e., $h^n_i$ is a smooth version of $1/q^n_i$ which agrees with $1/q^n_i$ when $q^n_i$ is bounded away from $0$. Furthermore, $h^n_i(\xi) \le n^{1/4}$ for all $\xi$. An immediate consequence are the following results.

        \begin{lemma}
        \label{lemma:absolute_continuity}
            For all $n$ and $i$, let $h^n_i$ be as in \eqref{eq:hni_stat1}. For all $(\xi, v) \in \mathsf{E}$, the function $t \mapsto h^n_i(\xi + vt)$ is absolutely continuous on $(0, \infty)$.
        \end{lemma}

        \begin{proof}
            First, observe that for any $\epsilon$, the function $r_{\epsilon}$ is globally Lipschitz by construction. From the smoothness assumption on $S$ (Assumption~\ref{assum:smoothness}) and from the definition in \eqref{eq:hni_stat}, the function $t \mapsto q^n_i(\xi + n^{1/2}vt)$ is absolutely continuous on $[0, \infty)$ for all $(\xi, v)$. Since $h^n_i(\xi) = r_{n^{-1/4}}(q^n_i(\xi))$ is a composition of a globally Lipschitz function with an absolutely continuous function, it is also absolutely continuous.
        \end{proof}

        In the spirit of the proof of \cite[][Theorem 5.1]{bierkens_limit_2017}, for each $i = 1, \dots, d$, define
        \[
            f^n_i(\xi, v) := \xi_i + \sqrt{n}v_ih^n_i(\xi), \numberthis \label{eq:fni}
        \]
        where $h^n_i$ is given by \eqref{eq:hni_stat1}. Let $(\mathcal{L}^n, \mathcal{D}(\mathcal{L}^n))$ denote the extended generator of the reparameterised Zig-Zag process $(U^n_{t})_{t \ge 0} = (\xi^n_t, V^n_t)_{t \ge 0}$. From \eqref{eq:zigzag_generator1}, for any $f \in  \mathcal{D}(\mathcal{L}^n)$,
        \begin{equation}
        \label{eq:gen_ss_stat}
            \mathcal{L}^nf(\xi, v) = n^{1/2}v\cdot \weak f(\xi,v) + \sum_{i=1}^d \lambda^n_i(\xi, v)\left(f(\xi, \mathsf{F}_i(v)) - f(\xi, v)\right), \quad (\xi, v) \in \mathsf{E}.
        \end{equation}
    
        \begin{lemma}
        \label{lemma:zz_ss_2}
            For each $n$, the functions $f^n_i$ and their products belong to the domain of the extended generator $\mathcal{L}^n$ of the Zig-Zag process, $\mathbb{P}$-almost surely. For all $(\xi, v) \in \mathsf{E}$,
            \[
                \mathcal{L}^nf^n_i(\xi, v) = n^{1/2}(v_i - 2v_i\lambda^n_i(\xi, v)h^n_i(\xi)) + nv_i(v \cdot \weak h^n_i(\xi)) \numberthis \label{eq:gen_stat}.
            \]
            Let $g^n_{ij} = f^n_if^n_j$. Then,
            \[
                \mathcal{L}^ng^n_{ij}(\xi, v) = 
                (f^n_i\cdot\mathcal{L}^nf^n_j + f^n_j\cdot\mathcal{L}^nf^n_i)(\xi, v) + 4n1_{\{i = j\}}\lambda^n_i(\xi,v)(h^n_i(\xi))^2\numberthis \label{eq:gen_stat1}.
            \]
        \end{lemma}
        \begin{proof}
            Fix $i$ and $n$ arbitrarily. From Lemma~\ref{lemma:absolute_continuity}, the function $t \mapsto f^n_i(\xi + n^{1/2}vt, v)$ is absolutely continuous on $[0, \infty)$ for all $(\xi, v) \in \mathsf{E}$. Let $T_1, T_2, \dots$ be the switching times for the Zig-Zag $Z^n = (\xi^n_t, v^n_t)_{t \ge 0}$. The jump measure of the process is given by, $\mu = \sum_{i=1}^{\infty} \delta_{(Z^n_{T_i}, T_i)}$.
            Define $\mathfrak{B}f^n$ coordinate-wise by, $\mathfrak{B}f^n_i(\xi, v, s) = f^n_i(\xi, v) - f^n_i(\xi^n_{s^{-}}, v^n_{s^{-}})$.
            Then for any $t \ge 0$,
            \begin{align*}
                \int_{\mathsf{E} \times [0, t)} |\mathfrak{B}f^n_i(z, s)|\mu(dz, ds) 
                &= \sum_{T_k \le t} |f^n_i(\xi^n_{T_k}, v^n_{T_k}) - f^n_i(\xi^n_{T_k}, v^n_{T_{k-1}})| 
                \le 2\sqrt{n} \sum_{T_k \le t} h^n_i(\xi^n_{T_k})
            \end{align*}
            from \eqref{eq:fni}. And thus,
            \[
                \mathbb{E} \int_{\mathsf{E} \times [0, t)} |\mathfrak{B}f^n_i(z, s)|\mu(dz, ds) \le 2\sqrt{n}\mathbb{E}\sum_{T_k \le t} h^n_i(\xi^n_{T_k}) < \infty,
            \]
            because $h^n_i < \infty$ and the switching rates are continuous, hence bounded on $[0, t]$. By \cite[][Theorem 5.5]{davis_piecewise-deterministic_1984}, $f^n_i$ belongs to the domain of the extended generator. Similar arguments hold for arbitrary products.

            Equations \ref{eq:gen_stat} and \ref{eq:gen_stat1} follow from \eqref{eq:gen_ss_stat} after routine calculations.
        \end{proof}

        The following series of lemmas establish control on the terms in \eqref{eq:gen_stat} and \eqref{eq:gen_stat1}.
    
        \begin{lemma}
        \label{lemma:zz_ss_11}
            Let $A = ((A_{ij}))_{ij\le d}$ be the diagonal matrix in \eqref{eq:A_def}. For all $R > 0$, and each $i = 1, \dots, d$,
            \[
                \lim_{n \to \infty}\sup_{\|\xi\| \le R}\left| 2nh^n_i(\xi) - A_{ii}\right| = 0, \quad \text{in }\Prob\text{-probability}.
            \]
        \end{lemma}
        \begin{proof}
            Fix $R$ and $i = 1, \dots, d$. Consider the set, 
            \[
                C_n = \left\{\omega: \sup_{\|\xi\| \le R}\left|q^n_i(\xi)/n - 2/A_{ii} \right| \le 1/A_{ii}\right\}. \numberthis \label{eq:sets_Cn}
            \]
            On this set, $n/A_{ii} \le q^n_i(\xi) \le 3n/A_{ii}$ for all $\|\xi\| < R$. Let $N_1$ be such that for all $n \ge N_1$, $2n^{-1/4} \le n/A_{ii}$. For such $n$, on the set $C_n$, $h^n_i = 1/q^n_i$ and,
            \[
                \sup_{\|\xi\| \le R}\left| 2nh^n_i(\xi) - A_{ii}\right| =  \sup_{\|\xi\| \le R}\left| \frac{2n - q^n_{i}(\xi)A_{ii}}{q^n_i(\xi)}\right| \le A^2_{ii}\sup_{\|\xi\| \le R}\left|q^n_i(\xi)/n - 2/A_{ii} \right|
            \]
            Let $\epsilon > 0$ be arbitrary. Then,
            \begin{align*}
                \mathbb{P}\left(\sup_{\|\xi\| \le R}\left| 2nh^n_i(\xi) - A_{ii}\right| < \epsilon\right) &\ge \mathbb{P}\left(\sup_{\|\xi\| \le R}\left| q^n_i(\xi)/n - 2/A_{ii}\right| < \frac{\epsilon}{A^2_{ii}}, C_n\right) \\
                &= \mathbb{P}\left(\sup_{\|\xi\| \le R}\left| q^n_i(\xi)/n - 2/A_{ii}\right| < \min\left\{\frac{\epsilon}{A^2_{ii}}, \frac{1}{A_{ii}} \right\}\right)
            \end{align*}
            The right-hand side goes to $1$ as $n \to \infty$ by Lemma~\ref{lemma:tilde_hn} and the result follows.
        \end{proof}

        Let $x_0$ and $I(x_0)$ be as in the Assumption~\ref{assum:unique_minimiser}. Let $I_{\cdot i}(x_0)$ denote the $i$-th coloumn of $I(x_0)$. Define $b: \mathbb{R}^d \to \mathbb{R}^d$ coordinate-wise by,
            \[
                b_i(\xi) = \frac{-A_{ii}\left(\xi \cdot I_{\cdot i}(x_0)\right)}{2}, \quad  \xi \in \mathbb{R}^d,\ i = 1,\dots, d, \numberthis \label{eq:bi_stat}
            \]
            where $A = ((A_{ij}))_{ij\le d}$ is the diagonal matrix in \eqref{eq:A_def}. 
            Define $b^n:\mathbb{R}^d \to \mathbb{R}^d$ coordinate-wise by,
            \[
               b^n_i(\xi) = h^n_i(\xi)\left(\lambda^n_i(\xi, -\bs{1}_d) - \lambda^n_i(\xi, \bs{1}_d)\right), \quad \xi \in \mathbb{R}^d. \numberthis \label{eq:bni_stat}
            \]
            where $h^n_i$ is as defined in \eqref{eq:hni_stat1}. We now approximate $b$ by $n^{1/2}b^n$.

        \begin{lemma}
        \label{lemma:zz_ss_1}
            For all $R > 0$, and each $i = 1, \dots, d$,
            \[
                \lim_{n \to \infty}\sup_{\|\xi\| \le R}\left| n^{1/2}b^n_i(\xi) - b_i(\xi)\right| = 0, \quad \text{in }\Prob\text{-probability}.
            \]
        \end{lemma}
        \begin{proof}
            Let $i$ be arbitrary, decompose $b^n_i$ as, 
            \begin{align*}
                n^{1/2}b^n_i(\xi) &= \left(\frac{\lambda^n_i(\xi, -\bs{1}_d) - \lambda^n_i(\xi, \bs{1}_d)}{n^{1/2}}\right)nh_i^n(\xi).
            \end{align*}
            From Lemma~\ref{lemma:zz_ss_11}, it is enough to show that,
            \[
                \lim_{n \to \infty}\sup_{|\xi| \le R}\left| \frac{\lambda^n_i(\xi, -\bs{1}_d) - \lambda^n_i(\xi, \bs{1}_d)}{n^{1/2}} - \left(-\xi \cdot I_{\cdot i}(x_0)\right)\right| = 0, \quad \text{in }\Prob\text{-probability}.
            \]
            Consider,
            \begin{align*}
                \frac{\lambda^n_i(\xi, -\bs{1}_d) - \lambda^n_i(\xi, \bs{1}_d)}{n^{1/2}} &= \frac{-1}{n^{1/2}}\left(\frac{n}{m|\mathcal{J}_{(n, m)}|}\sum_{J \in \mathcal{J}_{(m, n)}} \sum_{j \in J} S_i(\hat{x}_n + n^{-1/2}\xi, Y_j)\right) \\
                    &=\frac{-1}{n^{1/2}}\sum_{j=1}^n S_i(\hat{x}_n + n^{-1/2}\xi, Y_j).
            \end{align*}
            Once again, a second-order Taylor expansion about $\hat{x}_n$ gives,
            \begin{align*}
                &\frac{-1}{n^{1/2}}\sum_{j=1}^n S_i(\hat{x}_n + n^{-1/2}\xi, Y_j) \\
                    &\quad = \frac{-1}{n^{1/2}}\sum_{j=1}^n \left(S_i(\hat{x}_n, Y_j) + n^{-1/2}\xi\cdot\nabla S_i(\hat{x}_n, Y_j) + \frac{n^{-1}}{2}\xi^T \nabla^{2}S_i(\hat{x}_n + \theta_{j}n^{-1/2}\xi, Y_j)\xi\right) \\
                    &\quad = o_{p}(1) + \frac{-1}{n}\xi \cdot\sum_{j=1}^n \nabla S_i(\hat{x}_n, Y_j) +  O_{p}(n^{-1/2}),
            \end{align*}
            for some $\theta_{j} \in (0, 1)$. 
            The last equality follows due to the assumption that $\sum S(\hat{x}_n, Y_j) = o_{p}(n^{1/2})$ (Assumption~\ref{assum:mle_consistent}(ii)) and the assumption that $\nabla^2S_i(x, y)$ is uniformly bounded (Assumption~\ref{assum:smoothness}). By Lemma \ref{lemma:In_convergence}, 
            \begin{equation}
            \label{eq:bn1}
                \frac{\lambda^n_i(\xi, -\bs{1}_d) - \lambda^n_i(\xi, \bs{1}_d)}{n^{1/2}} \to -\xi \cdot \mathbb{E}[\nabla S_i(x_0, Y)] \quad \text{in }\Prob\text{-probability}.
            \end{equation}
            The uniform convergence on the set $\{|\xi|\le R\}$ follows by the observation that $\xi$ appears only as a multiplicative factor.            
        \end{proof}

        \begin{lemma}
        \label{lemma:zzss_4}
            Let $\weak$ denote the weak derivative operator with respect to the $\xi$ coordinate. For all $R>0$, and each $i = 1, \dots, d$,
            \[
                \lim_{n \to \infty}\sup_{\|\xi\| \le R}\left|n \weak h^n_i(\xi)\right| = 0, \quad \text{in } \mathbb{P}\text{-probability.}
            \]
        \end{lemma}
        \begin{proof}
            Let $i$ and $R$ be arbitrary. As in the previous proof, consider the set $C_n = \{\omega: q^n_i(\xi) > n/A_{ii}\  \forall \ \|\xi\| \le R \}$ and let $n$ be large enough such that $n/A_{ii} > 2n^{-1/4}$. On this set, $h^n_i(\xi) = 1/q^n_i(\xi)$ for all $\|\xi\| \le R$. 
            Let $\weak_k$ denote the $k$-th partial weak derivative operator. 
            For $J \in \mathcal{J}_{(n,m)}$, let $E^J_i(\xi)$ denote the sum $\sum_{j \in J}S_i(\hat{x}_n + n^{-1/2}\xi, Y_j)$. Let $k$ be arbitrary. From \eqref{eq:hni_stat} and \eqref{eq:bni_stat}, on the set $C_n$,
        \begin{align*}
          \weak_kh^n_i(\xi) &= \weak_k \left(\frac{n}{m|\mathcal{J}_{(n, m)}|}\sum_{J \in \mathcal{J}_{(m, n)}} \left|E^J_i(\xi)\right|\right)^{-1} \\
          &= -(h_i^n(\xi))^2 \cdot \left(\frac{n}{m|\mathcal{J}_{(n, m)}|}\sum_{J \in \mathcal{J}_{(m, n)}} \weak_k \left(\left|E^J_i(\xi)\right|\right)\right)\\
          &= -(h_i^n(\xi))^2 \cdot \left(\frac{n}{m|\mathcal{J}_{(n, m)}|}\sum_{J \in \mathcal{J}_{(m, n)}} \mathrm{sgn}(E^J_i)\cdot \weak_k(E^J_i)\right)\\
          &= -(h_i^n(\xi))^2 \cdot \left(\frac{n}{m|\mathcal{J}_{(n, m)}|}\sum_{J \in \mathcal{J}_{(m, n)}} \mathrm{sgn}(E^J_i)\cdot n^{-1/2}\sum_{j \in J}S'_{ki}(\hat{x}_n + n^{-1/2}\xi, Y_j)\right),
        \end{align*}
        where $S'_{kl}$ is as defined in the proof of Lemma \ref{lemma:In_convergence}. The above gives,
        \begin{align*}
          n|\weak_kh^n_i(\xi)| &\le n^{1/2}|h_i^n(\xi)|^2 \left(\sum_{j = 1}^n \left|S'_{ki}(\hat{x}_n + n^{-1/2}\xi, Y_j)\right|\right).
        \end{align*}
        Let $x'_n = (\hat{x}_n + n^{-1/2}\xi + x_0)$. Using first-order Taylor's expansion about $x_0$,
        \begin{align*}
          \sum_{j = 1}^n \left|S'_{ki}(\hat{x}_n + n^{-1/2}\xi, Y_j)\right|
          &\le \sum_{j = 1}^n |S'_{ki}(x_0, Y_j)| + |\nabla S'_{ki}(x_0 + \theta_jx'_n, Y_j)\cdot x'_n|,
        \end{align*}
        for some $\theta_j \in (0, 1)$. Then, by Assumption~\ref{assum:smoothness},
        \begin{align*}
          n|\weak_kh^n_i(\xi)|
          &\le n^{3/2}|h_i^n(\xi)|^2 \left(\frac{1}{n}\sum_{j = 1}^n |S'_{ki}(x_0)| + M'\|x'_n\|\right).
        \end{align*}
        By Assumption~\ref{assum:mle_consistent}, $x'_n$ converges to $0$ in $\mathbb{P}$-probability uniformly in $\xi$. By the law of large numbers and Assumption~\ref{assum:unique_minimiser}), the term inside the parentheses converges to a finite quantity. Also, $n^{3/2}|h^n_i|^2$ goes to $0$ uniformly from Lemma~\ref{lemma:zz_ss_11}. Since the above is true for arbitrary $k$, it follows that on the sets $C_n$, $n|\weak h^n_i(\xi)|$ goes to $0$ uniformly on $\{\xi: \|\xi\| \le R\}$ in $\mathbb{P}$-probability. Moreover, from Lemma~\ref{lemma:tilde_hn}, the sets $C_n$ satisfy $\mathbb{P}(C_n) \to 1$ as $n \to \infty$. Thus, the result follows.
        \end{proof}

        \begin{lemma}
        \label{lemma:zz_ss_7}
            For all $R > 0$, and each $i = 1, \dots, d$, and all $v \in \{-1,+1\}^d$,
                \[
                    \lim_{n \to \infty}\sup_{\|\xi\| \le R}\left| \mathcal{L}^nf^n_i(\xi, v) - b_i(\xi)\right| = 0, \quad \text{in }\Prob\text{-probability}.
                \]
        \end{lemma}
        \begin{proof}
            Fix $R$ and $i$. Consider the event $C_n = \{\omega: q^n_i(\xi) > n/A_{ii}\  \forall \ \|\xi\| \le R \}$ and let $n$ be large enough such that $n/A_{ii} > 2n^{-1/4}$. By Lemma~\ref{lemma:tilde_hn}, $\mathbb{P}(C_n) \to 1$ as $n \to \infty$. On $C_n$, $h^n_{i}(\xi) = 1/q^n_i(\xi)$ for large enough $n$. For such $n$, on the event $C_n$,
            \begin{align*}
                v_i - 2v_i\lambda^n_i(\xi, v)h^n_i(\xi) = v_i\left(1 - \frac{2\lambda^n_i(\xi, v)}{q^n_i(\xi)}\right) = v_i\left(\frac{\lambda^n_i(\xi, -\bs{1}_d) + \lambda^n_i(\xi, \bs{1}_d) - 2\lambda^n_i(\xi, v)}{\lambda^n_i(\xi, -\bs{1}_d) + \lambda^n_i(\xi, \bs{1}_d)}\right),
            \end{align*}
            from \eqref{eq:hni_stat}. But also, $\lambda^n_i(\xi, -\bs{1}_d) + \lambda^n_i(\xi, \bs{1}_d) = \lambda^n_i(\xi, \mathsf{F}_i(v)) + \lambda^n_i(\xi, v)$. Using this in the numerator,
            \[
                v_i - 2v_i\lambda^n_i(\xi, v)h^n_i(\xi) = \frac{v_i\lambda^n_i(\xi, \mathsf{F}_i(v)) - v_i\lambda^n_i(\xi, v)}{\lambda^n_i(\xi, -\bs{1}_d) + \lambda^n_i(\xi, \bs{1}_d)} = \frac{\lambda^n_i(\xi, -\bs{1}_d) - \lambda^n_i(\xi, \bs{1}_d)}{\lambda^n_i(\xi, -\bs{1}_d) + \lambda^n_i(\xi, \bs{1}_d)} = b^n_{i}(\xi), \numberthis \label{eq:lambdah}
            \]
            where the second equality follows from \eqref{eq:rates}. From \eqref{eq:gen_stat} and \eqref{eq:lambdah},
            \begin{align*}
                \left| \mathcal{L}^nf^n_i(\xi, v) - b_i(\xi)\right| &= \left| n^{1/2}(v_i - 2v_i\lambda^n_i(\xi, v)h^n_i(\xi)) + nv_i(v \cdot \weak h^n_i(\xi)) - b_i(\xi)\right| \\
                &\le |n^{1/2}b^n_i(\xi) - b_i(\xi)| + \sqrt{d}|n\weak h^n_i(\xi)|,
            \end{align*}
            where the last inequality is due to Cauchy-Schwartz and triangle inequality.
            
            Let $\epsilon > 0$ be arbitrary. Then,
            \begin{align*}
                \mathbb{P}\left(\sup_{\|\xi\| \le R}\left| \mathcal{L}^nf^n_i(\xi, v) - b_i(\xi)\right| > \epsilon \right) &\le \mathbb{P}\left(\sup_{\|\xi\| \le R}\left| n^{1/2}b^n_i(\xi) - b_i(\xi)\right| > \epsilon/2 \right) \\
                &\quad + \mathbb{P}\left(\sup_{\|\xi\| \le R}\left|n \weak h^n_i(\xi)\right| > \epsilon/2\sqrt{d} \right) + (1 - \mathbb{P}(C_n)).
            \end{align*}
            The first term on the right-hand side goes to $0$ by Lemma~\ref{lemma:zz_ss_1}. The second term goes to $0$ by Lemma~\ref{lemma:zzss_4}. The third term goes to $0$ by Lemma~\ref{lemma:tilde_hn}. The result thus follows.         
        \end{proof}

\begin{lemma}
\label{lemma:product_gen}
    Let $A = ((A_{ij}))_{ij\le d}$ be the diagonal matrix in \eqref{eq:A_def}. For all $R >0$, $v \in \{-1,+1\}^d$, and all $1 \le i, j \le d$,
    \[
        \lim_{n \to \infty}\sup_{\|\xi\| \le R}\left|(\mathcal{L}^ng^n_{ij} - f^n_i\mathcal{L}^nf^n_j - f^n_j\mathcal{L}^nf^n_i)(\xi, v) - A_{ij}\right| = 0, \quad \text{in } \mathbb{P}\text{-probability.}
    \]
\end{lemma}

\begin{proof}
    For all $i \neq j$, the claim holds trivially by \eqref{eq:gen_stat1} and since $A_{ij} = 0$. Let $i$ be arbitrary. From \eqref{eq:gen_stat1},
    \[
        \left|(\mathcal{L}^ng^n_{ij} - f^n_i\mathcal{L}^nf^n_j - f^n_j\mathcal{L}^nf^n_i)(\xi, v) - A_{ij}\right| = \left|4n\lambda^n_i(\xi, v)(h^n_i(\xi))^2 - A_{ii}\right|
    \]
    As in the previous proof, consider the set $C_n = \{\omega: q^n_i(\xi) > n/A_{ii}\  \forall \ \|\xi\| \le R \}$ and let $n$ be large enough such that $n/A_{ii} > 2n^{-1/4}$. On this set, $h^n_i(\xi) = 1/q^n_i(\xi)$ for all $\|\xi\| \le R$. Multiplying by $v_i$ throughout \eqref{eq:lambdah} gives $2\lambda^n_i(\xi, v)h^n_i(\xi) = 1 - v_ib^n_i(\xi)$. Thus,
    \[
        4n\lambda^n_i(\xi)(h^n_i(\xi))^2 = 2nh^n_i(\xi) -2nv_ib^n_i(\xi)h^n_i(\xi).
    \]
    Furthermore,
    \[
        \left|4n\lambda^n_i(\xi, v)(h^n_i(\xi))^2 - A_{ii}\right| \le \left|2nh^n_i(\xi) - A_{ii}\right| + \left|2nb^n_i(\xi)h^n_i(\xi)\right|
    \]
    Let $\epsilon > 0$ be arbitrary. As in the previous proofs, we get,
    \begin{align*}
        &\mathbb{P}\left(\sup_{\|\xi\| \le R}\left|(\mathcal{L}^ng^n_{ij} - f^n_i\mathcal{L}^nf^n_j - f^n_j\mathcal{L}^nf^n_i)(\xi, v) - A_{ij}\right| > \epsilon\right) \\
        &\quad \le \mathbb{P}\left(\sup_{\|\xi\| \le R}\left|4n\lambda^n_i(\xi, v)(h^n_i(\xi))^2 - A_{ii}\right| > \epsilon\right) + (1 - \mathbb{P}(C_n)) \\
        &\quad \le \mathbb{P}\left(\sup_{\|\xi\| \le R}\left|2nh^n_i(\xi) - A_{ii}\right| > \epsilon/2\right) + \mathbb{P}\left(\sup_{\|\xi\| \le R}\left|2nb^n_i(\xi)h^n_i(\xi)\right| > \epsilon/2\right) + (1 - \mathbb{P}(C_n)).
    \end{align*}
    The first term on the right-hand side goes to $0$ by Lemma~\ref{lemma:zz_ss_11}. The last term goes to $0$ by Lemma~\ref{lemma:tilde_hn}. The middle term goes to $0$ by combining Lemma~\ref{lemma:zz_ss_11} and Lemma~\ref{lemma:zz_ss_1} with Slutsky's theorem. The result follows.
\end{proof}

        We are now ready to prove Theorem~\ref{theo:zzss_stat_fixed}.

\subsection{Proof of Theorem \ref{theo:zzss_stat_fixed}}
\label{sec:proo_stat_zzss}
        In this section, we prove Theorem~\ref{theo:zzss_stat_fixed} following the approach of \cite{bierkens_limit_2017}. Let $f^n_i$ be as defined in \eqref{eq:fni}. Define 
        \[
            Y^{n, i}_t := f^n_i(\xi^n_t, V^n_t) \quad \text{and} \quad j^{n, i}_t := \mathcal{L}^nf^n_i(\xi^n_t, V^n_t).
        \]
        It follows that, for all $i = 1, \dots, d$,
        \[
            M^{n, i}_t = Y^{n, i}_t - \int_0^t j^{n, i}_s \ ds,
        \]
        is a local martingale with respect to the filtration $\mathcal{F}^n_t$ generated by $(\xi^n_t, V^n_t)_{t \ge 0}$. Consequently, the vector process $M^n_t = (M^{n, 1}_t, \dots, M^{n, d}_t)$ is a local martingale. For all $i, j$, define
        \[
            g^n_{ij}(\xi, v) := f^n_i(\xi, v)f^n_j(\xi, v).
        \] 
        Then $g^n_{ij}$ belongs to the domain of $\mathcal{L}^n$. We get that,
        \[
          N^{n, ij}_t := Y^{n, i}_tY^{n, j}_t - \int_0^t \mathcal{L}^n g^n_{ij}(\xi^n_s, V^n_s) ds,
        \]
        is a local martingale with respect to the filtration $\mathcal{F}^n_t$. We now decompose the outer product of the local martingale $M^n_t$ into a local martingale term and a remainder. To this end, for all $i$, define $J^{n,i}_t = \int_0^t j^{n,i}_s\ ds$. Using integration by parts,
        \begin{align*}
          M^{n,i}_tM^{n,j}_t &= (Y^{n,i}_t - J^{n,i}_t)(Y^{n,j}_t - J^{n,j}_t) \\
          &\hspace{-3em}= Y^{n,i}_tY^{n,j}_t - J^{n,i}_tY^{n,j}_t - Y^{n,i}_tJ^{n,j}_t + J^{n,i}_tJ^{n,j}_t\\
          &\hspace{-3em}= Y^{n,i}_tY^{n,j}_t - \int_0^t J^{n,i}_s\ dY^{n,j}_s - \int_0^t Y^{n,j}_s j^{n,i}_s\ ds - \int_0^t J^{n,j}_s\ dY^{n,i}_s - \int_0^t Y^{n,i}_s j^{n,j}_s\ ds \\
          &\quad \quad + \int_0^t J^{n,i}_s j^{n,j}_s\ ds + \int_0^t J^{n,j}_s j^{n,i}_s\ ds \\
          &\hspace{-3em} =  Y^{n,i}_tY^{n,j}_t - \int_0^t Y^{n,j}_s j^{n,i}_s\ ds - \int_0^t Y^{n,i}_s j^{n,j}_s\ ds - \int_0^t J^{n,i}_s\ dM^{n,j}_s - \int_0^t J^{n,j}_s\ dM^{n,i}_s \\
          &\hspace{-3em} = N^{n, ij}_t  - \int_0^t J^{n,i}_s\ dM^{n,j}_s - \int_0^t J^{n,j}_s\ dM^{n,i}_s \\
          &\quad \quad + \int_0^t \mathcal{L}^n g^n_{ij}(\xi^n_s, v^n_s) ds  - \int_0^t Y^{n,j}_s j^{n,i}_s\ ds - \int_0^t Y^{n,i}_s j^{n,j}_s\ ds. 
        \end{align*}
        It follows that, for all $i, j$,
        \[
          H^{n, ij}_t := M^{n,i}_tM^{n,j}_t - \int_0^t \mathcal{L}^n g^n_{ij}(\xi^n_s, V^n_s) ds + \int_0^t Y^{n,j}_s j^{n,i}_s\ ds + \int_0^t Y^{n,i}_s j^{n,j}_s\ ds,
        \]
        is a local martingale with respect to the filtration $\mathcal{F}^n_t$. 
    
        The weak convergence result will be proven using Theorem 4.1, Chapter 7 of \cite{ethier_markov_1986}. To this end, define, for all $n$, $B^n_t := \xi^n_t - M^n_t$. Write $B^{n, i}$ for the $i$-th component. Then,
        \[
            B^{n,i}_t = -\sqrt{n}(v^n_t)_ih^n_i(\xi^n_t) + \int_0^t j^{n,i}_s\ ds, \quad i = 1, \dots, d.
        \]
        Also define, for all $n$, a symmetric $d \times d$ matrix valued process $A^n_t = ((A^{n, ij}_t))$ as,
        \begin{align*}
            A^{n, ij}_t = M^{n,i}_tM^{n,j}_t - H^{n, ij}_t &= \int_0^t \mathcal{L}^n g^n_{ij}(\xi^n_s, V^n_s) ds - \int_0^t Y^{n,j}_s j^{n,i}_s\ ds - \int_0^t Y^{n,i}_s j^{n,j}_s\ ds\\
              &= \int_0^t (\mathcal{L}^ng^n_{ij} - f^n_i\mathcal{L}^nf^n_j - f^n_j\mathcal{L}^nf^n_i)(\xi^n_s, V^n_s) \ ds. \numberthis \label{eq:Ant}
        \end{align*}
        Define stopping time $\tau^n_r := \inf\{t \ge 0: \|\xi^n_t\| \ge r \text{ or } \|\xi^n(t-)\| \ge r\}$. Let $b$ and $A$ be as defined in \eqref{eq:bi_stat} and \eqref{eq:A_def} respectively. Let $T, R \ge 0$ and $t \le T \wedge \tau^n_R$.

        By properties of the Zig-Zag process, $(B^n_t)$ has right-continuous paths. Moreover, for all $i$, $n$, and $t \ge 0$,
        \[
            |B^{n, i}_t - B^{n, i}_{t^-}| = |-\sqrt{n}(V^n_t)_ih^n_i(\xi^n_t) + \sqrt{n}(V^n_{t^-})_ih^n_i(\xi^n_t)| \le 2\sqrt{n}h^n_i(\xi_t).
        \]
        Then, for all $R > 0$, it follows from Lemma~\ref{lemma:zz_ss_11} that,
        \[
            \sup_{\|\xi\| \le R} \|B^n_t - B^n_{t^-}\|^2 \le \sum_{i=1}^d \sup_{\|\xi\| \le R}4n(h^n_i(\xi))^2 \to 0,  \numberthis \label{eq:bni_continuity}
        \]
        as $n \to \infty$ in $\mathbb{P}$-probability. Let $b$ be as defined in \eqref{eq:bi_stat}. For an arbitrary $i$, consider,
        \begin{align*}
            \left|B^{n,i}_t - \int_0^t b_i(\xi^n_s)\ ds \right|
            &= \left| -\sqrt{n}(V^n_t)_ih^n_i(\xi^n_t) + \int_0^t \mathcal{L}^nf^n_i(\xi^n_s, V^n_s)\ ds - \int_0^t b_i(\xi^n_s)\ ds\right|
            \\
            &\le \int_0^t\left|\mathcal{L}^nf^n_i(\xi^n_s, V^n_s) -  b_i(\xi^n_s)\right|\ ds + \left|n^{1/2}h^n_i(\xi^n_t)\right|.
        \end{align*}
        This gives,
        \begin{align*}
            \sup_{t \le T \wedge \tau^n_R}\left|B^{n,i}_t - \int_0^t b_i(\xi^n_s)\ ds \right| &\le T\sup_{|\xi| \le R}\left| \mathcal{L}^nf^n_i(\xi, v) - b_i(\xi)\right| + T\sup_{|\xi| \le R}\left|n^{1/2}h^n_i(\xi)\right|.
        \end{align*}
        The first term on the right-hand side goes to $0$ in $\mathbb{P}$-probability by Lemma~\ref{lemma:zz_ss_7}. The second term goes to $0$ in $\mathbb{P}$-probability by Lemma~\ref{lemma:zz_ss_11}. Since $i$ was arbitrary, it follows that for all $i$,
        \begin{equation}
        \label{eq:limit_Bn}
            \sup_{t \le T \wedge \tau^n_R}\left|B^{n,i}_t - \int_0^t b_{i}(\xi^n_s)\ ds \right| \xrightarrow{n \to \infty} 0,\quad \text{in }\Prob\text{-probability}.
        \end{equation}
        
        
        Now, consider $(A^n_t)$ defined in \eqref{eq:Ant}. By construction, $(A^n_t)$ has continuous paths for all $n$ and so, $A^n_t = A^n_{t^-}$ for all $t$. By \eqref{eq:gen_stat1}, $A^n_t$ is a diagonal matrix for all $t$ and for all $t > s \ge 0$, $A^n_t - A^n_s$ is non-negative definite. Let $A = ((A_{ij}))_{ij\le d}$ be the diagonal matrix in \eqref{eq:A_def}. Consider, for any $i,j$,
        \begin{align*}
            \left|A^{n, ij}_t - \int_0^t A_{ij}\ ds \right| &= \left| \int_0^t (\mathcal{L}^ng^n_{ij} - f^n_i\mathcal{L}^nf^n_j - f^n_j\mathcal{L}^nf^n_i)(\xi^n_s, V^n_s) \ ds - tA_{ij}\right| \\
            &\le \int_0^t \left| (\mathcal{L}^ng^n_{ij} - f^n_i\mathcal{L}^nf^n_j - f^n_j\mathcal{L}^nf^n_i)(\xi^n_s, V^n_s)- A_{ij}\right|\ ds
        \end{align*}
        This gives,
        \begin{equation*}
            \sup_{t \le T \wedge \tau^n_R}\left|A^{n,ij}_t - \int_0^t A_{ij}\ ds \right| \le T\sup_{|\xi| \le R}\left| (\mathcal{L}^ng^n_{ij} - f^n_i\mathcal{L}^nf^n_j - f^n_j\mathcal{L}^nf^n_i)(\xi, v)- A_{ij}\right|.
        \end{equation*}
        The right-hand side goes to $0$ in $\mathbb{P}$-probability by Lemma~\ref{lemma:product_gen}. Thus, for all $i, j$,
        \begin{equation}
        \label{eq:limit_An}
            \sup_{t \le T \wedge \tau^n_R}\left|A^{n,ij}_t - \int_0^t A_{ij}\ ds \right| \xrightarrow{n \to \infty} 0,\quad \text{in }\Prob\text{-probability}.
        \end{equation}

        For any subsequence $(n_k)$, let $(n_{k_l})$ be a further sub-sequence such that the convergence in \eqref{eq:bni_continuity}, \eqref{eq:limit_Bn}, and \eqref{eq:limit_An} is almost sure. Then by Theorem 4.1, Chapter 7 of \cite{ethier_markov_1986}, along subsequence $(n_{k_l})$,  $(\xi^{n}_t)$ converges weakly in Skorohod topology to the solution of the martingale problem of the operator $(\mathcal{A}, \mathcal{D}(\mathcal{A}))$, where $\mathcal{D}(\mathcal{A}) = C^2_c(\mathbb{R}^d)$ and for $h \in \mathcal{D}(\mathcal{A})$,
        \[
          \mathcal{A}h(\xi) = b(\xi)\cdot\nabla h(\xi) + \frac{1}{2}\text{Tr}(\nabla^{\otimes 2}h(\xi)A).
        \] 
        Since the martingale problem above is well-posed, the result follows. \qed

\section{Proofs for ZZ-CV in the Stationary phase}
\label{sec:proof_zzcv}

Throughout this section, we will suppose that Assumptions~\ref{assum:smoothness}, \ref{assum:unique_minimiser}, \ref{assum:mle_consistent}, and \ref{assum:reference_stat}, all hold. Once again, we suppress the dependence on $\omega$ in all of the following. We will prove Theorem~\ref{theo:zzcv_stat_fixed} and Theorem~\ref{theo:zzcv_stat_varying} by proving the uniform convergence of generators in $\mathbb{P}$-probability.

\subsection{Technical lemmas} 
\label{sec:technical_cv}
    Recall, from \eqref{eq:cv_rate_xi}, that the $i$-th ZZ-CV switching rate in $\xi$-parameter is given by,
    \begin{equation}
        \lambda^n_{\text{cv}, i}(\xi, v; m(n)) = \frac{n}{m(n)|\mathcal{J}_{(n,m)}|}\sum_{J \in \mathcal{J}_{(n,m)}}\left(v_i \cdot \sum_{j \in J}E_i^j(\xi)\right)_{+}, \label{eq:rescaledrate}
    \end{equation}
    where,
    \[
        E_{\text{cv}}^j(\xi) = S(\hat{x}_n + n^{-1/2}\xi, Y_j) - S(X_n^*, Y_j) + n^{-1}\nabla U_n(X_n^*).
    \]
    We will establish the uniform convergence of switching rates. Define $\xi^*_n = n^{1/2}(X_n^* - \hat{x}_n)$ for all $n$. Let $I(\hat{x}_n)$ be as defined in \eqref{eq:ixn} and let $I_{\cdot i}(\hat{x}_n)$ denote the $i-$th column of $I(\hat{x}_n)$. We drop the subscript `cv' and $m(n)$ from arguments for notational convenience.
    
    \subsubsection{Case 1: $m(n) \to m < \infty$}
        Suppose $m(n) \to m < \infty$ and let $N$ be such that $m(n) = m$ for all $n \ge N$. In the following suppose $n \ge N$. Let $I(\hat{x}_n)$ be as defined in \eqref{eq:ixn} and let $I_{\cdot i}(\hat{x}_n)$ denote the $i$-th column of $I(\hat{x}_n)$. 
     Define, for $i = 1, \dots, d$,
        \begin{equation}
        \label{eq:An_seq}
            F_i^j(\xi) = \xi\cdot \nabla S_i(\hat{x}_n, Y_j) - \xi_n^*\cdot\left(\nabla S_i(\hat{x}_n, Y_j) - I_{\cdot i}(\hat{x}_n)\right), \quad j = 1, \dots, n.
        \end{equation}
        and correspondingly,
        \[
            \tilde{\lambda}^n_{i}(\xi, v) = \frac{1}{m|\mathcal{J}_{(n,m)}|}\sum_{J \in \mathcal{J}_{(n,m)}}\left(v_i \cdot \sum_{j \in J}F_i^j(\xi)\right)_{+}. \numberthis \label{eq:tilde_stat}
        \]
        The next approximation holds irrespective of the subsample size, i.e., whether $m(n) = m$ for a fixed $m \in \mathbb{N}$ or $m(n) \to \infty$ as $n \to \infty$.
        \begin{lemma}
        \label{lemma:zz_cv_1}
            Let $K \subset \mathbb{R}^d$ be compact. Then, for all $i = 1, \dots, d$, 
            \[
               \sup_{\xi \in K} \left|n^{-1/2}\lambda^{n}_{i}(\xi, v) - \tilde{\lambda}^n_{i}(\xi, v) \right| \to 0, \quad \text{in }\Prob-\text{probability}.
            \]
        \end{lemma}
        \begin{proof}
            Let $i$ and $n$ be arbitrary. Observe that,
            \begin{align*}
                \left| n^{-1/2}\lambda^n_{i}(\xi, v) - \tilde{\lambda}^n_{i}(\xi, v)\right| 
                &\le \frac{n}{m|\mathcal{J}_{(n,m)}|}\sum_{J \in \mathcal{J}_{(n,m)}}\sum_{j \in J}\left|n^{-1/2}E_i^j(\xi) - n^{-1}F_i^j(\xi)\right| \\
                &= \sum_{j=1}^n \left|n^{-1/2}E_i^j(\xi) - n^{-1}F_i^j(\xi)\right|.
            \end{align*}
            Using a second-order Taylor expansion about $\hat{x}_n$ we have for each $j$,
            \[
                E^j_i(\xi) = n^{-1/2}F^j_i(\xi) + R^j_i(\xi),
            \]
            where $R^j_i(\xi)$ is a remainder term given by,
            \begin{align*}
                R^j_i(\xi) &= \frac{n^{-1}}{2}\xi^T \nabla^{2}S_i(\hat{x}_n + \theta_{j}n^{-1/2}\xi, Y_j)\xi - \frac{n^{-1}}{2}(\xi_n^*)^T\nabla^{2}S_i(\hat{x}_n + \theta^*_{j}n^{-1/2}\xi^*_n, Y_j)\xi_n^* \\
                &\quad \quad + n^{-1}\sum_{j=1}^n S_i(\hat{x}_n, Y_j) + \frac{n^{-2}}{2}(\xi_n^*)^T\sum_{k=1}^n \nabla^{2}S_i(\hat{x}_n + \theta^*_{k}n^{-1/2}\xi^*_n, Y_k)\xi_n^*.
            \end{align*}
            for some $\theta_j, \theta^*_j \in (0, 1)$.
            Thus for all $j$, by Assumption \ref{assum:smoothness} and by Assumption~\ref{assum:mle_consistent} (ii),
            \begin{equation*}
                \left|n^{-1/2}E^j_i(\xi) - n^{-1}F^j_{i}(\xi)\right| = \left|n^{-1/2}R_i^j(\xi)\right| \le \frac{n^{-3/2}}{2}\|\xi\|^2 M' + o_{p}(n^{-1}) + n^{-3/2}\|\xi_n^*\|^2M'.
            \end{equation*}
            This gives,
            \begin{align*}
                \left| n^{-1/2}\lambda^n_{i}(\xi, v) - \tilde{\lambda}^n_{i}(\xi, v)\right| &\le \frac{n^{-1/2}}{2}\|\xi\|^2 M' + o_{p}(1) + n^{-1/2}\|\xi_n^*\|^2M'.
            \end{align*}
            Since, by Assumption~\ref{assum:reference_stat}, $\xi^*_n \to \xi^*$ in $\mathbb{P}$-probability with $\mathbb{E}\|\xi^*\|^2 < \infty$, the result follows for all $i = 1, \dots, d$.
        \end{proof}

        First, we will prove the convergence when $m(n) = m$ for a fixed $m \in \mathbb{N}$.
    
        \begin{lemma}
        \label{lemma:zz_cv_2}
           Let $K \subset \mathbb{R}$ be compact. In \eqref{eq:rescaledrate}, let $m(n) = m$ for a fixed $m \in \mathbb{N}$. Let $\{\lambda_i\}_{i=1}^d$ be the asymptotic switching rates defined in \eqref{eq:limrate}. Then, for all $i = 1, \dots, d$,
            \[
                \sup_{\xi \in K} \left|n^{-1/2}\lambda^{n}_{i}(\xi, v) - \lambda_{i}(\xi, v) \right| \to 0, \quad \text{in }\Prob-\text{probability}.
            \]
        \end{lemma}
        \begin{proof}
            Let $i$ be arbitrary. Due to Lemma \ref{lemma:zz_cv_1}, it is sufficient to show the result for $\tilde{\lambda}^{n}_{i}(\xi, v)$ instead of $n^{-1/2}\lambda_{i}(\xi, v)$, where $\tilde{\lambda}^n_i$ is defined in \eqref{eq:tilde_stat}. Define, for $j = 1, \dots, n$,
            \begin{equation}
            \label{eq:gi}
                G^j_i(\xi) = \xi \cdot \nabla S_i(x_0, Y_j) - \xi^*\cdot\left(\nabla S_i(x_0, Y_j) - I_{\cdot i}(x_0)\right),
            \end{equation}
            where $I_{\cdot i}(x_0)$ denotes the $i$-th column of $I(x_0) = -\mathbb{E}_{P}\nabla^2_x\log f(Y; x_0)$. Then, as in the proof of Lemma~\ref{lemma:zz_cv_1},
            \begin{align*}
                \left|\tilde{\lambda}^n_{i}(\xi, v) - \frac{1}{m|\mathcal{J}_{(n,m)}|}\sum_{J \in \mathcal{J}_{(n,m)}}\left(v_i \cdot \sum_{j \in J}G^j_i(\xi)\right)_{+} \right| 
                &\le \frac{1}{n} \sum_{j=1}^n |F^j_i(\xi) - G^j_i(\xi)|. \numberthis \label{eq:temp1}
            \end{align*}
            A first-order Taylor expansion about $x_0$ gives,
            \[
                \nabla S_i(\hat{x}_n, Y_j) = \nabla S_i(x_0, Y_j) + \nabla^{ 2}S_i(x_0 + \theta_{j}(\hat{x}_{n} - x_0), Y_j) \cdot (\hat{x}_n - x_0),
            \]
            for some $\theta_{j} \in (0, 1)$. From \eqref{eq:An_seq}, we write $F^j_i$ as,
            \begin{align*}
                F^j_i(\xi) &= \xi\cdot \nabla S_i(\hat{x}_n, Y_j) - \xi_n^*\cdot\left(\nabla S_i(\hat{x}_n, Y_j) - I_{\cdot i}(\hat{x}_n)\right) \\
                &= (\xi - \xi^*_n) \cdot \left(\nabla S_i(x_0, Y_j) + \nabla^{2}S_i(x_0 + \theta_{i, j}(\hat{x}_{n} - x_0), Y_j) \cdot (\hat{x}_n - x_0)\right) + \xi^*_n\cdot I_{\cdot i}(\hat{x}_n)
            \end{align*}
            Then, from \eqref{eq:gi}, and Assumption~\ref{assum:smoothness},
            \begin{align*}
                &\left|F^j_i(\xi) - G^j_i(\xi)\right| \\
                &\quad = \left|(\xi - \xi_n^*)^T\nabla^{2}S_i(x_0 + \theta_{i, j}(\hat{x}_{n} - x_0), Y_j)(\hat{x}_n - x_0) - (\xi_n^* - \xi^*)\cdot \nabla S_i(x_0, Y_j) \right.\\
                &\quad \quad \quad \quad \left. - (\xi_n^* - \xi^*)\cdot \nabla S_i(x_0, Y_j) + \xi^*_n \cdot I_{\cdot i}(\hat{x}_n) - \xi^*\cdot I_{\cdot i}(x_0) \right| \\
                &\quad \le M'\|\xi - \xi_n^*\|\|\hat{x}_n - x_0\| + \|\xi^*_n - \xi^*\|\|\nabla S_i(x_0, Y_j)\| + |\xi^*_n \cdot I_{\cdot i}(\hat{x}_n) - \xi^*\cdot I_{\cdot i}(x_0)|.
            \end{align*}
            Putting this back in \eqref{eq:temp1}, we get,
            \begin{align*}
                &\left|\tilde{\lambda}^n_{i}(\xi, v) - \frac{1}{m|\mathcal{J}_{(n,m)}|}\sum_{J \in \mathcal{J}_{(n,m)}}\left(v_i \cdot \sum_{j \in J}G^j_i(\xi)\right)_{+} \right| \\
                &\quad\le M'\|\xi - \xi_n^*\|\|\hat{x}_n - x_0\| + |\xi^*_n \cdot I_{\cdot i}(\hat{x}_n) - \xi^*\cdot I_{\cdot i}(x_0)| + \|\xi^*_n - \xi^*\|\frac{1}{n}\sum_{j=1}^n\|\nabla S_i(x_0, Y_j)\|.
            \end{align*}
            Under Assumption~\ref{assum:unique_minimiser}, $n^{-1}\sum_{j=1}^{n}\|\nabla S_i(x_0, Y_j)\| \to \mathbb{E}\|\nabla S_i(x_0)\| < \infty$ almost surely as $n \to \infty$. By Assumption~\ref{assum:reference_stat}, $\xi^*_n \to \xi^*$ and by Lemma \ref{lemma:In_convergence}, $I(\hat{x}_n) \to I(x_0)$. Together, this gives,
            \[
                \sup_{\xi \in K} \left|\tilde{\lambda}^{n}_{i}(\xi, v) - \frac{1}{m|\mathcal{J}_{(n,m)}|}\sum_{J \in \mathcal{J}_{(n,m)}}\left(v_i \cdot \sum_{j \in J}G^j_i(\xi)\right)_{+} \right| \to 0,
            \]
            in $\Prob-$probability. 
            
            Finally, observe that given $\xi^*$ and for each fixed $i$, $G^j_i$s are conditionally i.i.d. random variables. Also given $\xi^*$, for each fixed $(\xi, v)$ and $i$,
            \[
                \frac{1}{m|\mathcal{J}_{(n,m)}|}\sum_{J \in \mathcal{J}_{(n,m)}}\left(v_i \cdot \sum_{j \in J}G^j_i(\xi)\right)_{+}
            \]
            is a $U$-statistic of degree $m$ with kernel $f_m(y_1, \dots, y_m)$ given by,
            \[
                f_m(y_1, \dots, y_m) = \frac{1}{m}\left(v_i \cdot \sum_{j=1}^m \xi\cdot \nabla s_i(x_0; y_j) - \xi^*\cdot (\nabla s_i(x_0; y_j) - I_{\cdot i}(x_0))\right)_{+}.
            \]  
            Then, conditioned on $\xi^*$, Proposition \ref{propo:slln_ustatistics} gives convergence for each $(\xi, v)$. Furthermore, since $\xi$ and $v$ appear only as multiplicative factors, the strong law is uniform on compact sets. Finally, since $\xi^*$ is independent of the data sequence by Assumption~\ref{assum:reference_stat}, the result follows.  
        \end{proof}

        \subsubsection{Case 2: $m(n) \to \infty$}
        Next, we prove convergence for $m(n) \to \infty$ as $n \to \infty$.

        \begin{lemma}
        \label{lemma:zz_cv_3}
           Let $K \subset \mathbb{R}$ be compact. In \eqref{eq:rescaledrate}, suppose $m(n) \to \infty$ as $n \to \infty$. Let $\{\lambda_i\}_{i=1}^d$ be the asymptotic switching rates defined in \eqref{eq:limrate_can}. Then, for all $i = 1, \dots, d$,
            \[
                \sup_{\xi \in K} \left|n^{-1/2}\lambda^{n}_{i}(\xi, v) - \lambda_{i}(\xi, v) \right| \to 0, \quad \text{in }\Prob-\text{probability}.
            \]
        \end{lemma}
        \begin{proof}
            Let $i$ be arbitrary. Following the steps in Lemma \ref{lemma:zz_cv_1} and Lemma \ref{lemma:zz_cv_2}, it is possible to once again show that for any $K \subset \mathbb{R}^d$,
            \[
                \sup_{\xi \in K} \left|\tilde{\lambda}^{n}_{i}(\xi, v) - \frac{1}{m(n)|\mathcal{J}_{(n,m(n))}|}\sum_{J \in \mathcal{J}_{(n,m(n))}}\left(v_i \cdot \sum_{j \in J}G^j_i(\xi)\right)_{+} \right| \to 0,
            \]
            for all $i = 1, \dots, d$, where $G^j$ is defined in \eqref{eq:gi}. Following a similar argument as in the proof of Proposition \ref{propo:rates}, one can further show that whenever $m(n) \to \infty$ as $n \to \infty$, the difference
            \[
                \left|\frac{1}{|\mathcal{J}_{(n,m(n))}|}\sum_{J \in \mathcal{J}_{(n,m(n))}}\left(v_i \cdot m(n)^{-1}\sum_{j \in J}G^j_i(\xi)\right)_{+} - \left(v_i \cdot n^{-1} \sum_{j = 1}^n G^j_i(\xi)\right)_{+}\right|
            \]
            goes to $0$ uniformly on compact sets in $\Prob-$probability. This implies that for all $i$ and compact $K \subset \mathbb{R}^d$,
            \[
                \sup_{\xi \in K} \left|n^{-1/2}\lambda^{n}_{i}(\xi, v) - n^{-1}\left(v_i \cdot \sum_{j = 1}^{n} G^j_i(\xi)\right)_{+} \right| \to 0, \quad \text{in } \Prob-\text{probability}. 
            \] 
            Also, the strong law of large numbers gives, 
            \begin{align*}
                &\frac{1}{n}\left(v_i \cdot \sum_{j = 1}^{n} G^j_i(\xi)\right)_{+} \\
                &\quad \xrightarrow{n \to \infty} \left(v_i \mathbb{E}_{Y}\left[\xi\cdot \nabla S_i(x_0; Y) - \xi^*\cdot(\nabla S_i(x_0; Y) - \mathbb{E}\nabla S_i(x_0))\right]\right)_{+},
            \end{align*}
            almost surely. But, 
            \[
                \mathbb{E}_{Y}\left[\xi\cdot \nabla S_i(x_0; Y) - \xi^*\cdot(\nabla S_i(x_0; Y) - \mathbb{E}\nabla S_i(x_0))\right] = \xi\cdot \mathbb{E}_{Y} \nabla S_i(x_0; Y)
            \]
            almost surely. And thus we get, for all $i$,
            \[
                \sup_{\xi \in K} \left|n^{-1/2}\lambda^{n}_{i}(\xi, v) - (v_i (\xi\cdot\mathbb{E} \nabla S_i(x_0; Y))_{+} \right| \to 0, \quad \text{in } \Prob-\text{probability}.
            \] 
            The result follows.
        \end{proof}

\subsection{Convergence to Zig-Zag processes}
\label{sec:proof_stat_zzcv}

        Let $\mathcal{L}^n$ be the infinitesimal generator for the rescaled process $U^n_{n^{-1/2}t} = (\xi^n_{n^{-1/2}t}, V^n_{n^{-1/2}t})$. Then, from \eqref{eq:zigzag_generator1},
        \[
            \mathcal{L}^nf(\xi, v) = v\cdot\nabla_{\xi}f(\xi, v) + \sum_{i=1}^d n^{-1/2}\lambda^n_{i}(\xi, v)\{f(x, F_i(v)) - f(x, v)\}, \numberthis \label{eq:gen_cv_stat}
        \]
        for an arbitrary test function $f$. 

        \begin{proof}[\it Proof of Theorem \ref{theo:zzcv_stat_fixed}]
            Let $f$ be such that $f(\cdot, v) \in \mathcal{C}_c^{\infty}(\mathbb{R}^d)$ for each $v$. The infinitesimal generator $\mathcal{L}$ of the Zig-Zag process with switching rates \eqref{eq:limrate} is given by,
            \[
                \mathcal{L}f(\xi, v) = v\cdot\nabla_{\xi}f(\xi, v) + \sum_{i=1}^d\lambda_{i}(\xi, v)\{f(x, F_i(v)) - f(x, v)\}.  \numberthis \label{eq:gen_cv_stat_asymp}.
            \]
            From \eqref{eq:gen_cv_stat} and \eqref{eq:gen_cv_stat_asymp},
            \begin{align*}
                \sup_{\xi \in \mathbb{R}^d} \left|\mathcal{L}^{n}f(\xi, v) - \mathcal{L}f(\xi, v)\right| &\le \sup_{\xi \in K} 2\|f\|_{\infty}\sum_{i=1}^d\left|n^{-1/2}\lambda_{i}(\xi, v) - \lambda_{i}(\xi, v)\right|,
            \end{align*}
            where $\|f\|_{\infty}$ denotes the sup-norm of $f(\cdot, v) \in \mathcal{C}_c^{\infty}(\mathbb{R}^d)$ for each $v$, $\|f\|_{\infty} < \infty$. Then, by Lemma~\ref{lemma:zz_cv_2},
            \[
                \sup_{\xi \in \mathbb{R}^d} \left|\mathcal{L}^{n}f(\xi, v) - \mathcal{L}f(\xi, v)\right| \to 0, \quad \text{in } \Prob-\text{probability}. \numberthis \label{eq:zigzag_gen_converge}.
            \]

            For any subsequence $(n_k)$ of $(n)_{n=1}^{\infty}$, let $(n_{k_l})$ be a further subsequence such that the convergence in \eqref{eq:zigzag_gen_converge} is almost sure along $(n_{k_{l}})$. Then by Corollary 8.7(f), Chapter 4 of \cite{ethier_markov_1986} and since $\mathcal{C}_c^{\infty}$ is a core for $\mathcal{L}$ \cite{holderrieth_cores_2021}, $U^{n}$ converges weakly (in Skorohod topology) to $U$, $\Prob-$almost surely along the subsequence $(n_{k_l})$. This in turn implies weak convergence (in Skorohod topology) in $\Prob-$probability of $U^n$ to $U$ as $n \to \infty$.
        \end{proof}

        \begin{proof}[\it Proof of Theorem \ref{theo:zzcv_stat_varying}] 
            The proof is exactly the same as the proof of Theorem~\ref{theo:zzcv_stat_fixed} with switching rates in the expression for generator in \eqref{eq:gen_cv_stat_asymp} given by \eqref{eq:limrate_can} and the convergence in \eqref{eq:zigzag_gen_converge} following from Lemma~\ref{lemma:zz_cv_3}. We omit further details.
        \end{proof}

\section{Complexity bounds via weak convergence}
\label{sec:complexity}
    In this section, we will prove Corollaries~\ref{corr:complexity_ss} and \ref{corr:complexity_cv}. First, we provide a slight generalization of \cite[][Theorem 1]{roberts_complexity_2016}.

   \subsection{A general complexity bound via weak convergence}
    
    Let $\|\dots\|_{\text{KR}}$ be as defined in \eqref{eq:kr_distance}.
    Then, $\|\dots\|_{\text{KR}}$ is a norm \cite[][Lemma~1]{roberts_complexity_2016}, and the metric $\Delta(\mu, \nu) := \|\mu - \nu\|_{\text{KR}}$ metrises weak convergence of probability measures \cite[][Proposition~1]{roberts_complexity_2016} on separable metric spaces endowed with their Borel $\sigma$-algebras. Also, by \cite[][Lemma~2]{roberts_complexity_2016} and \eqref{eq:kr_distance}, $\|\cdots\|_{\text{KR}} \le 2$. 
    Let $\|\dots\|_{\text{TV}}$ denote the total variation distance, i.e., for any two probability measures $\mu$ and $\nu$,
    \[
        \|\mu - \nu\|_{\text{TV}} = \sup_{|f| \le 1}|\mu(f) - \nu(f)|.
    \]
    Then, $\|\dots\|_{\text{KR}} \le \|\dots\|_{\text{TV}}$. The following result is an easy generalization of \cite[][Theorem 1]{roberts_complexity_2016}. In the statement below, both $E$ and $F$ are separable metric spaces endowed with their Borel $\sigma$-algebras.

    \begin{theorem}
    \label{theo:complexity}
        For each $n \in \mathbb{N}$, let $\{Y^{n}_t\}_{t \ge 0}$ be a time-homogeneous Markov process on $F$, and let $\mu^n$ be the corresponding stationary distribution. For some continuous $\eta: F \to E$, let $X^n = \{X^n_t\}_{t \ge 0}$ be defined by $X^n_t = \eta(Y^n_t)$. Let $\nu^n = \mu^n \circ \eta^{-1}$. Suppose $\mu^n$ converges to $\mu$ as $n \to \infty$ in total variation distance and let $\nu = \mu \circ \eta^{-1}$. Suppose, for $\mu$-almost all $y$, $X^n \mid Y^n_0 = y$ converges weakly in the Skorohod topology to $X^y$ as $n \to \infty$ where $X^y = \{X^y_t\}_{t \ge 0}$ is a c\'adl\'ag process in $E$. Suppose $X^y$ converges, either weakly or in total variation distance, to $\nu$ as $t \to \infty$ for $\mu$-almost all $y$. Then, for any $\epsilon > 0$, there are $N < \infty$ and $T< \infty$ such that
        \[
            \mathbb{E}_{Y^n_0 \sim \mu^n}\left\|\text{Law}(X^n_t \mid Y^n_0) - \nu^n\right\|_{\text{KR}} < \epsilon, \quad t \ge T, n \ge N.
        \]
    \end{theorem}

    We will prove Theorem~\ref{theo:complexity} analogously to the proof of \cite[][Theorem~1]{roberts_complexity_2016}. 

    \begin{lemma}
    \label{lemma:decreasing_fn}
        Under the assumptions of Theorem~\ref{theo:complexity}, $\mathbb{E}_{Y^n_0 \sim \mu^n}\left\|\text{Law}(X^n_t \mid Y^n_0) - \nu^n\right\|_{\text{KR}}$ is a non-increasing function of $t$.
    \end{lemma}
    \begin{proof}
        The proof is same as the proof of \cite[][Lemma~4]{roberts_complexity_2016} and follows from the Markov property of $\{Y^n_t\}_{t\ge 0}$ and stationarity of $\mu^n$.
    \end{proof}

    \begin{propo}
    \label{propo:complexity2}
        Under the assumptions of Theorem~\ref{theo:complexity}, for all $\epsilon > 0$, there are $N < \infty$ and $T < \infty$ such that,
        \[
            \mathbb{E}_{Y^n_0 \sim \mu^n}\left\|\text{Law}(X^n_T \mid Y^n_0) - \nu^n\right\|_{\text{KR}} < \epsilon, \quad n \ge N.
        \]
    \end{propo}
    \begin{proof}
        Let,
        \[
            A_m = \left\{y \in F: \sup_{t \ge m}\|\text{Law}(X^y_t) - \nu\|_{\text{KR}} \le \frac{\epsilon}{16}\right\}.
        \]
        Then, $A_{m+1} \supseteq A_m$. Moreover, since $X^y$ converges, either weakly or in total variation distance to $\nu$ for $\mu$-almost all $y$, $\lim_{m \to \infty}\mu(A_m) = \mu(\cup_{m}A_m) = 1$ by \cite[][Lemma~3]{roberts_complexity_2016} and the continuity property of probability measures. Let $m_0 < \infty$ be such that $\mu(A_{m_0}) \ge 1 - \epsilon/32$. Let $T \ge m_0$ be such that $T$ is a continuity point of $X^y$ for $\mu$-almost all $y$, i.e., $\mathbb{P}(X^y_{T} = X^y_{T^-}) = 1$ for $\mu$-almost all $y$. Such a $T$ exists since for any $y$, $X^y$ has c\'adl\'ag paths and thus at most countably many discontinuities. Tonelli's theorem then implies that the set of times for which $X^y_t$ is continuous for $\mu$-almost all $y$ is of full Lebesgue measure.

        Now, for this fixed $T$, let
        \[
            B_m = \left\{y \in F: \sup_{n \ge m}\|\text{Law}(X^{n}_T \mid Y^n_0 = y) - \text{Law}(X^y_T)\|_{\text{KR}} \le \frac{\epsilon}{16}\right\}.
        \]
        Then, once again $B_{m+1} \supseteq B_m$. Since, $X^n \mid Y^n_0 = y$ converges weakly in Skorohod topology to $X^y$ for $\mu$-almost all $y$, and since $\|\dots\|_{\text{KR}}$ metrises weak convergence, $\lim_{m \to \infty}\mu(B_m) = \mu(\cup_{m}B_m) = 1$. Let $N_1$ be such that $\mu(B_{N_1}) \ge 1 - \epsilon/32$.

        Since $\mu^n$ converges in total variation distance to $\mu$, there exists $N_2 < \infty$ such that $\|\mu^n - \mu\|_{\text{TV}} < \epsilon/4$ for all $n \ge N_2$. Moreover, since the total variation distance contracts under push-forward and since $\|\dots\|_{\text{KR}} \le \|\dots\|_{\text{TV}}$, we get, 
        \[
            \|\nu^n - \nu\|_{\text{KR}} \le \|\nu^n - \nu\|_{\text{TV}} \le \|\mu^n - \mu\|_{\text{TV}} < \epsilon/4
        \]
        for all $n \ge N_2$. Let $N = \max\{N_1,N_2\}$. Then, $\mu(B_N) \ge \mu(B_{N_1}) \ge 1 - \epsilon/32$. 
        
        For this fixed $N$ and $T$, for any $n \ge N$,
        \begin{align*}
            \mathbb{E}_{Y^n_0 \sim \mu}\left\|\text{Law}(X^n_T \mid Y^n_0) - \nu^n\right\|_{\text{KR}} &=  \mathbb{E}_{Y^n_0 \sim \mu}\left(1_{Y^n_0 \in A_{m_0} \cap B_{N}}\left\|\text{Law}(X^n_T \mid Y^n_0) - \nu^n\right\|_{\text{KR}}\right) \\
            &\quad + \mathbb{E}_{Y^n_0 \sim \mu}\left(1_{Y^n_0 \notin A_{m_0} \cap B_{N}}\left\|\text{Law}(X^n_T \mid Y^n_0) - \nu^n\right\|_{\text{KR}}\right) \\
            &\le \left[\frac{\epsilon}{16} + \frac{\epsilon}{16} + \frac{\epsilon}{4}\right] + 2 \left[\frac{\epsilon}{32} + \frac{\epsilon}{32}\right]\\
            &\le \epsilon/2, \numberthis \label{eq:complexity_bound_2}
        \end{align*}
        where for the first term we have used the triangle inequality twice, and for the second term, we used the fact that $\|\dots\|_{\text{KR}} \le 2$.

        Finally, from \eqref{eq:complexity_bound_2} and triangle inequality, for all $n \ge N$,
        \begin{align*}
            \mathbb{E}_{Y^n_0 \sim \mu^n}\left\|\text{Law}(X^n_T \mid Y^n_0) - \nu^n\right\|_{\text{KR}} &\le \mathbb{E}_{Y^n_0 \sim \mu}\left\|\text{Law}(X^n_T \mid Y^n_0) - \nu^n\right\|_{\text{KR}} + 2\left\|\mu^n - \mu\right\|_{\text{TV}} \\
            &< \mathbb{E}_{Y^n_0 \sim \mu}\left\|\text{Law}(X^n_T \mid Y^n_0) - \nu^n\right\|_{\text{KR}} + \epsilon/2\\
            &\le \epsilon.
        \end{align*}
        The result follows.
    \end{proof}

    \begin{proof}[Proof of Theorem~\ref{theo:complexity}]
        The proof follows by combining Lemma~\ref{lemma:decreasing_fn} and Proposition~\ref{propo:complexity2}.
    \end{proof}

    \subsection{Complexity bounds for ZZ-SS and ZZ-CV in the stationary phase}        
    \label{sec:proof_complexity_stat}
    We are now ready to prove Corollaries~\ref{corr:complexity_ss} and \ref{corr:complexity_cv}. Recall that $\Pi^{(n)}$ is the random posterior with Lebesgue density \eqref{eq:pin}.  Let $\Pi^{(n)}_{\sharp \xi}$ denote the push-forward of $\Pi^{(n)}$ under the reparameterization \ref{eq:rescaling}. Let $I(x_0)$ be as in Assumption~\ref{assum:unique_minimiser} and let $\Pi = N(0, (I(x_0))^{-1})$, i.e., the $d$-dimensional distribution with mean vector $0$ and covariance matrix $(I(x_0))^{-1}$. Define
    \[
        \Pi^{(n)}_{\mathsf{E}} = \Pi^{(n)}_{\sharp \xi} \otimes \text{Uniform}\left(\{-1,+1\}^d\right), \quad \Pi_{\mathsf{E}} = \Pi \otimes \text{Uniform}\left(\{-1,+1\}^d\right). \numberthis \label{eq:extended_target}
    \]
    By Assumption~\ref{assum:mle_consistent} (iii), as $n\to \infty$
    \[
        \|\Pi^{(n)}_{\mathsf{E}} - \Pi_{\mathsf{E}}\|_{\text{TV}} = \|\Pi^{(n)}_{\sharp \xi} - \Pi\|_{\text{TV}} \to 0, \quad \text{in } \mathbb{P}-\text{probability}. \numberthis \label{eq:mu_converge}
    \]
     Let $\{U^n_{t}\}_{t \ge 0} = \{(\xi^n_t, V^n_t)\}_{t \ge 0}$ be a Zig-Zag process defined, conditional on the dataset $\bs{Y}^{(n)}$ (and possibly a reference point $\xi^{*}_n$), as a random element on $(\Omega, \mathcal{F}, \mathbb{P})$ to have stationary distribution $\Pi^{(n)}_{\mathsf{E}}$.

    \begin{corollary}
    \label{corr:complexity}
        Suppose, for $\Pi_{\mathsf{E}}$-almost all $u$, given $U^n_0 = u$, $\xi^n = \{\xi^n_t\}_{t \ge 0}$ converges weakly in the Skorohod topology to a c\'adl\'ag process $\{\xi^u_t\}_{t \ge 0}$, in $\mathbb{P}$-probability. Suppose $\{\xi^u\}_{t \ge 0}$ converges, either weakly or in total variation distance, to $\Pi$ as $t \to \infty$ for $\Pi_{\mathsf{E}}$-almost all $u$, and $\mathbb{P}$-almost surely. Then, for any $\epsilon > 0$ and $\delta > 0$, there are $N < \infty$ and $T< \infty$, such that
        \[
            \mathbb{P}\left(\mathbb{E}_{U^n_0 \sim \Pi^{(n)}_{\mathsf{E}}}\left\|\text{Law}(\xi^n_t \mid U^n_0) - \Pi^{(n)}_{\sharp \xi}\right\|_{\text{KR}} < \epsilon \ \forall\ t \ge T \right) > 1- \delta, \quad n \ge N.
        \]
    \end{corollary}
    \begin{proof}
        Suppose the claim is not true. There exists an $\epsilon> 0$ and $\delta > 0$ and an increasing subsequence $\{n_k\}$ such that,
        \[
            \mathbb{P}\left(\mathbb{E}_{U^{n_k}_0 \sim \Pi^{(n_k)}_{\mathsf{E}}}\left\|\text{Law}(\xi^{n_k}_t \mid U^{n_k}_0) - \Pi^{({n_k})}_{\sharp \xi}\right\|_{\text{KR}} < \epsilon \ \forall\ t \ge k \right) \le 1- \delta. \numberthis \label{eq:corr_f4_2}
        \]
        
        From the hypothesis and from \eqref{eq:mu_converge}, there exists a further subsequence $\{n_{k_{l}}\}_{l=1}^{\infty}$ such that for $\mathbb{P}$-almost all $\omega$, $\{\xi^{n_{k_l}}_t\}$ given $U^{n_{k_l}}_0 = u$ converges weakly in the Skorohod topology to $\{\xi^u_t\}_{t \ge 0}$, as $l \to \infty$ for $\Pi_{\mathsf{E}}$-almost all $u$. As in the proof of Proposition~\ref{propo:complexity2}, this is possible due to a combination of bounded convergence and Tonelli's theorem applied to a bounded metric characterizing weak convergence in Skorohod topology. By \eqref{eq:mu_converge}, $\{n_{k_l}\}_{l=1}^{\infty}$ can be chosen that $\Pi^{(n_{k_l})}_{\mathsf{E}}$ converges in total variation distance to $\Pi_{\mathsf{E}}$ as $l \to \infty$, $\mathbb{P}$-almost surely. For every $\omega$ in this $\mathbb{P}$-almost sure set, by Theorem~\ref{theo:complexity}, there exist $L < \infty$ and $T < \infty$, possibly dependent on $\omega$, such that,
        \[
            \mathbb{E}_{U^{n_{k_l}}_0 \sim \Pi^{(n_{k_l})}_{\mathsf{E}}}\left\|\text{Law}(\xi^{n_{k_l}}_t \mid U^{n_{k_l}}_0) - \Pi^{(n_{k_l})}_{\sharp \xi}\right\|_{\mathrm{KR}} < \epsilon/2, \quad  t \ge T, \ l \ge L, \numberthis \label{eq:corr_f4_1}
        \]
        Since $k_l \to \infty$ as $l \to \infty$, for each $\omega$, $l \ge L$ and $k_l \ge T$ for all sufficiently large $l$.
        Consequently,
\[
    \mathbf{1}\left\{
        \mathbb{E}_{U^{n_{k_l}}_0 \sim
            \Pi^{(n_{k_l})}_{\mathsf{E}}}
        \left\|
            \operatorname{Law}(\xi^{n_{k_l}}_t
                \mid U^{n_{k_l}}_0)
            - \Pi^{(n_{k_l})}_{\sharp \xi}
        \right\|_{\mathrm{KR}}
        < \epsilon \ \forall\, t \ge k_l
    \right\}
    \longrightarrow 1,
    \qquad \mathbb{P}\text{-almost surely}.
\]
Bounded convergence therefore gives
\[
    \mathbb{P}\left(
        \mathbb{E}_{U^{n_{k_l}}_0 \sim
            \Pi^{(n_{k_l})}_{\mathsf{E}}}
        \left\|
            \operatorname{Law}(\xi^{n_{k_l}}_t
                \mid U^{n_{k_l}}_0)
            - \Pi^{(n_{k_l})}_{\sharp \xi}
        \right\|_{\mathrm{KR}}
        < \epsilon \ \forall\, t \ge k_l
    \right)
    \longrightarrow 1.
\]
This contradicts the upper bound $1-\delta$ along the original
subsequence and proves the claim.
    \end{proof}

        \begin{proof}[Proof of Corollary~\ref{corr:complexity_ss}]
        For every initial condition $U^n_{\text{ss}, 0} = (x, v) \in \mathsf{E}$, Theorem~\ref{theo:zzss_stat_fixed} implies that $\{\xi^n_{\text{ss}, t}\}_{t \ge 0}$ converges weakly in Skorohod topology to $\{\xi_t\}_{t \ge 0}$, in $\mathbb{P}$-probability. The limiting process $\{\xi_t\}_{t \ge 0}$ in Theorem~\ref{theo:zzss_stat_fixed} is an Ornstein-Uhlenbeck process with invariant distribution $\Pi$ and initial condition $\xi_0= x$, for each $\omega \in \Omega$. It is well-known (see, e.g., \cite[][]{pavliotis_stochastic_2014}) that this process converges in total variation distance to $\Pi$ as $t \to \infty$ for $\Pi$-almost all $x$. The conditions for Corollary~\ref{corr:complexity} are met and the result follows.
    \end{proof}

    \begin{proof}[Proof of Corollary~\ref{corr:complexity_cv}]
        The proof is same as the proof of Corollary~\ref{corr:complexity_ss}. For every initial condition $(x, v) \in \mathsf{E}$, in both Theorem~\ref{theo:zzcv_stat_fixed} and \ref{theo:zzcv_stat_varying}, $\{U^n_{\text{cv}, n^{-1/2}t}\}_{t \ge 0}$ converges weakly in Skorohod topology to $\{U_t\}_{t \ge 0} = \{(\xi_t, V_t)\}_{t \ge 0}$, in $\mathbb{P}$-probability. The limiting process is a Zig-Zag process with invariant distribution $\Pi$ and initial condition $(x, v)$ for each $\omega \in \Omega$, although with different switching rates. In particular, $\{\xi^n_{\text{cv}, n^{-1/2}t}\}_{t \ge 0}$ converges weakly in Skorohod topology to $\{\xi_t\}_{t \ge 0}$, in $\mathbb{P}$-probability. The ergodicity of Zig-Zag process (see \cite[][Theorem~1]{bierkens_ergodicity_2019}) implies that $\{\xi_t\}_{t \ge 0}$ converges in total variation distance to $\Pi$ as $t \to \infty$ for $\Pi$-almost all $x$. The conditions for Corollary~\ref{corr:complexity} are met and the result follows.
    \end{proof}

\section{A mean absolute deviation inequality for sums of i.i.d. random variables}
\label{sec:mad_inequality}
    Let $b(n,p;k)$ denote the binomial probability,
    \[
        b(n, p; k) = \binom{n}{k}p^{k}(1 - p)^{n-k}, \quad  n \in \mathbb{N}, p \in [0,1], k = 0,1\dots, n.\numberthis \label{eq:binom}
    \]
    Let $B_{n,p}$ denote the $\text{Binomial}(n, p)$ random variable.

    \begin{lemma}
    \label{lemma:binom}
        Let $n \in \mathbb{N}$. For all $p \in (0,1)$ and $k \in \{0, 1, \dots, n\}$,
        \[
            b(n,p;k) \le \left(1 - \frac{k}{n}\right)(1 -p) + \frac{k}{n}p. \numberthis \label{eq:binom_lemma}
        \]
    \end{lemma}
    \begin{proof}
       For $k = 0$ and $k = n$, the claim follows from,
       \[
            (1 -p)^n \le (1-p), \quad p^n \le p, \quad p \in (0,1).
        \]
        For $k \in \{1, \dots, n-1\}$, calculating adjacent binomial probabilities,
        \[
            b(n-1,p;k-1) = \frac{k}{np}b(n,p;k), \quad b(n-1,p;k) = \frac{n - k}{n(1 -p)}b(n,p;k).
        \]
        Thus, since the total probability is bounded by $1$,
        \[
            \left(\frac{k/n}{p} + \frac{1 - k/n}{1 -p}\right)b(n,p;k) = b(n-1,p;k-1) + b(n-1,p;k) \le 1
        \]
        The above implies,
        \[
            b(n, p; k) \le \left(\frac{k/n}{p} + \frac{1 - k/n}{1 -p}\right)^{-1}
        \]
        for all $p \in (0,1)$ and $k \in \{1, \dots, n-1\}$. The right-hand side is a weighted harmonic mean and, since it is bounded above by the corresponding weighted arithmetic mean, the claim follows.
    \end{proof}

    \begin{propo}
    \label{propo:binom}
         Let $n \in \mathbb{N}$. For all $p \in (0, 1)$, and $a \in \mathbb{R}$,
         \[
            \mathbb{E}|B_{n,p} - na| \ge \mathbb{E}|B_{1,p} - a|.
         \]
    \end{propo}
    \begin{proof}
        First consider $a \notin [0, 1]$. For any $p \in (0,1)$, both $B_{n,p} - na$ and $B_{1,p} - a$ have the same constant sign almost surely. And so,
        \[
            \frac{\mathbb{E}|B_{n,p} - na|}{\mathbb{E}|B_{1,p} - a|} = \frac{|\mathbb{E}B_{n,p} - na|}{|\mathbb{E}B_{1,p} - a|} = n \ge 1.
        \]
        
        Now suppose $a \in [0, 1]$. For any $p \in (0,1)$, 
        \[
            \mathbb{E}|B_{1,p} - a| = a(1 - p) + (1 - a)p = 1 - (1-a)(1-p) - ap. \numberthis \label{eq:denom_G}
        \]
        And,
        \[
            \mathbb{E}|B_{n,p} - na| = \sum_{j=0}^n |j - na|b(n,p;j).
        \]
        For fixed $p$, the function $a \mapsto \mathbb{E}|B_{n,p} - na| - \mathbb{E}|B_{1,p} - a|$ is piecewise linear in $a$ with breakpoints at $a = j/n$, $j = 0, 1, \dots, n$. The minimum or maximum of $\mathbb{E}|B_{n,p} - na| - \mathbb{E}|B_{1,p} - a|$ occurs at one of these breakpoints. It is, hence, enough to show that,
        \[
            \mathbb{E}|B_{n,p} - k| \ge \mathbb{E}|B_{1,p} - k/n|
        \]
        for all $p \in (0,1)$ and $k \in \{0, 1, \dots, n\}$. Consider,
        \begin{align*}
            \mathbb{E}|B_{n,p} - k| &= \sum_{j=0}^{\infty} \mathbb{P}(|B_{n,p} - k| \ge j) \ge 1 -\mathbb{P}(|B_{n,p} - k| = 0) = 1 - b(n,p;k),
        \end{align*}
        for all $k \in \{0, 1, \dots, n\}$ and $p \in (0,1)$. But, by Lemma~\ref{lemma:binom} and \eqref{eq:denom_G},
        \begin{align*}
            \mathbb{E}|B_{n,p} - k| &\ge 1 - b(n,p;k) \ge 1- (1 - k/n)(1-p) -kp/n = \mathbb{E}|B_{1,p} - k/n|.
        \end{align*}
        The claim follows.
    \end{proof}

    \begin{propo}
    \label{propo:mad_general}
        Let $n \in \mathbb{N}$. Let $X, X_1, \dots, X_n$ be i.i.d. random variables such that $\mathbb{E}|X| < \infty$. Then, for all $n \in \mathbb{N}$,
        \[
            \mathbb{E}|X_1 + \cdots + X_n| \ge \mathbb{E}|X|.
        \]
    \end{propo}
    \begin{proof}
        Let $n$ be arbitrary. Let $Y, Y_1, \dots, Y_n$ be i.i.d. random variables with $\mathbb{P}(Y = 1) = p = 1 - \mathbb{P}(Y = 0)$ for some $p \in (0,1)$. From \cite[][Theorem 1]{bshouty_extrema_1993} and \cite{mattner_mean_2003},
        \[
            \frac{\mathbb{E}|X_1 + \cdots + X_n|}{\mathbb{E}|X|} \ge \min\left\{\frac{\mathbb{E}|Y_1 + \dots + Y_n - na|}{\mathbb{E}|Y - a|}: p \in (0, 1), \ a \in \mathbb{R}\right\}.
        \]
        Since $\sum_{j=1}^n Y_j \sim \text{Binomial}(n, p)$, the result follows from Proposition~\ref{propo:binom}.
    \end{proof}

\end{appendix}

\begin{acks}[Acknowledgments]
    We acknowledge fruitful discussions with Paul Fearnhead, Sebastiano Grazzi and Kengo Kamatani. We thank the anonymous referees for their insightful and constructive feedback.
\end{acks}

\begin{acks}[Statement on the use of artificial intelligence tools]
    The authors used GPT-6.1 Astra in finding proofs of Propositions~\ref{propo:binom} and~\ref{propo:mad_general}, and for proof-reading the final versions of the paper.
\end{acks}

\begin{funding}
    SA was supported by the EPSRC studentship grant EP/W523793/1, and by the Independent Research Fund Denmark (DFF) through the Sapere Aude Starting Grant (5251-00032B). JB was funded by the research programme ‘Zigzagging through computational barriers’ with project
number 016.Vidi.189.043, which was financed by the Dutch Research Council (NWO). GOR acknowledges the financial support provided by UKRI grant EP/Y014650/1, as part of the ERC Synergy project OCEAN. GOR was also supported by EPSRC grants Bayes for Health (R018561), CoSInES (R034710), PINCODE (EP/X028119/1), and EP/V009478/1.
%
\end{funding}

\bibliographystyle{imsart-number} 
\bibliography{references}       


\end{document}